\documentclass[10pt, final, twocolumn, twoside, romanappendices]{IEEEtran}
\usepackage{bigints,multirow,cite,algorithm,algorithmic,amsbsy,amsmath,amssymb,epsfig,bbm,mathrsfs,fancyhdr,fancyvrb,subfigure,url,color}

\newcommand{\EX}[1]{{\mathbb{E}}\left\{{#1}\right\}}
\newcommand{\EXs}[2]{{\mathbb{E}}_{\scriptsize{#1}}\!\!\left\{{#2}\right\}}

\newcommand{\beq}{\begin{equation}}
\newcommand{\eeq}{\end{equation}}
\newcommand{\beqn}{\begin{eqnarray}}
\newcommand{\eeqn}{\end{eqnarray}}

\newtheorem{theorem}{\bfseries{Theorem}} 
\newtheorem{lemma}{\bfseries{Lemma}}

\newtheorem{remark}{\bfseries{Remark}}

\newenvironment{proof}[1][Proof:]{\begin{trivlist}
\item[\hskip \labelsep {\bfseries #1}]}{\end{trivlist}}
\newenvironment{definition}[1][Definition]{\begin{trivlist}
\item[\hskip \labelsep {\bfseries #1}]}{\end{trivlist}}

\newcommand{\qed}{\nobreak \ifvmode \relax \else
      \ifdim\lastskip<1.5em \hskip-\lastskip
      \hskip1.5em plus0em minus0.5em \fi \nobreak
      \vrule height0.75em width0.5em depth0.25em\fi}
						\usepackage[margin=0.7 in]{geometry}

      \usepackage[table,xcdraw]{xcolor}

\newcommand{\paperTitle}{ Modeling and Analysis of Cellular Networks using Stochastic Geometry: A Tutorial}

\definecolor{BLUE}{rgb}{0,0,1}

\begin{document}

\title{\paperTitle}


\author{
	\vspace{0.2cm}
    Hesham~ElSawy, \textit{Member,~IEEE},
    Ahmed Sultan-Salem, \textit{Member,~IEEE},\\
    Mohamed-Slim~Alouini, \textit{Fellow,~IEEE}, and
    Moe~Z.~Win, \textit{Fellow,~IEEE}    \\
    \thanks{
	}
\thanks{H.~ElSawy, A.~Sultan-Salem, and M.-S.~Alouini are with the King Abdullah University of Science and Technology (KAUST), Thuwal, Makkah Province, Saudi Arabia 
		(email: \{hesham.elsawy, ahmed.salem, slim.alouini\}@kaust.edu.sa).} 
\thanks{M.~Z.~Win is with the Laboratory for Information and Decision Systems (LIDS), Massachusetts Institute of Technology, Cambridge, MA 02139, USA (e-mail: moewin@mit.edu).}
    }
    
\maketitle

\thispagestyle{empty}

\setcounter{page}{1}

\begin{abstract}
This paper presents a tutorial on stochastic geometry (SG) based analysis for cellular networks. This tutorial is distinguished by its depth with respect to wireless communication details and its focus on cellular networks. The paper starts by modeling and analyzing the baseband interference in a basic cellular network model. Then, it characterizes signal-to-interference-plus-noise-ratio (SINR) and its related performance metrics. In particular, a unified approach to conduct error probability, outage probability, and rate analysis is presented. {Although the main focus of the paper is on cellular networks, the presented unified approach applies for other types of wireless networks that impose interference protection around receivers}. The paper then extends the baseline unified approach to capture cellular network characteristics (e.g., frequency reuse, multiple antenna, power control, etc.). It also presents numerical examples associated with demonstrations and discussions. Finally, we point out future research directions. 

\end{abstract}

\section{Introduction}

Stochastic geometry (SG) has  succeeded to provide a unified mathematical paradigm to model different types of wireless networks, characterize their operation, and understand their behavior~\cite{moe_win, survey_martin, survey_h}. The main strength of the analysis based on SG, hereafter denoted as SG analysis, can be attributed to its ability to capture the spatial randomness inherent in wireless networks. Furthermore, SG models can be naturally extended to account for other sources of uncertainties such as fading, shadowing, and power control. In some special cases, SG analysis can lead to closed-form expressions that govern system behavior. These expressions enable the understanding of network operation and provide insightful design guidelines, which are often difficult to get from computationally intensive simulations.



{SG analysis for wireless networks can be traced back to the late 70's~\cite{1st, 2nd, 3rd, 4th, 5th}. At that point in time, SG was first used to design the transmission ranges and strategies in multi-hop ad hoc networks. 
Then, SG was used to characterize the aggregate interference coming from a Poisson field of interferers~\cite{sousa1990, mathar1995, stable_2}.\footnote{The Poisson field of interferers means that the transmitters are randomly, independently, and uniformly scattered in the spatial domain, in which the number of transmitters in any bounded region in the space is a Poisson random variable.} Despite the existence of a large number of interfering sources, it is shown in \cite{sousa1990, mathar1995, stable_2} that the central limit theorem does not apply, and consequently, the aggregate interference does not follow the Gaussian distribution. This is due to the prominent effect of distance-dependent path-loss attenuation, which makes the aggregate interference dominated by proximate interferers. The research outcome in \cite{sousa1990, moe_win_error, moe_win_error_2, sousa1992, moe_win} has shown that the aggregate interference follows the $\alpha$-stable distribution~\cite{stable_book1, stable_1, stable_2}, which is more impulsive and heavy tailed than the Gaussian distribution~\cite{stable_book2}. In fact, the aggregate interference has been characterized by generalizing shot-noise theory in higher dimensions~\cite{now_martin, stable_2, shot_noise1, shot_noise2, shot_noise3}. Such characterization has set the foundations for SG analysis, enriched the literature with valuable results, and helped to understand the behavior of several wireless technologies in large-scale setups \cite{1st, 2nd, 3rd, 4th, 5th, sousa1990, sousa1992, cog_2, moe_win_error, moe_win_error_2, Sec_moe, Sec_moe_P1, Sec_moe_P2, Nikias95, ilow1998, ambike1994, win2006error, zanella_1, UWB1, Spec_out_1, shot_noise3}. However, these results are confined to ad hoc networks with no spectrum access coordination schemes. In wireless networks with coordinated spectrum access, the aforementioned analysis presents pessimistic results.}

{Due to the shared nature of the wireless spectrum, along with the reliability requirement for communication links, spectrum access is usually coordinated to mute interference sources {nearby receivers}. This can be achieved by separating nearby transmissions over orthogonal resources (i.e., time, frequency, or codes). However, due to the scarcity of resources and the high demand for wireless communication, the wireless resources are reused over the spatial domain. The receivers are usually protected from interference resulting from spatial frequency reuse by interference exclusion regions. Cellular networks, which are the main focus of this tutorials, impose interference protection for users' terminal via the cellular structure. This intrinsic property of cellular network should be incorporated into analysis. Furthermore, several medium access control protocols exist in ad hoc networks (e.g., carrier sensing multiple access) that impose interference protection around receivers. Accounting for the interference protection around receivers, the aggregate interference is neither $\alpha$-stable nor Gaussian distributed~\cite{Martin_finite}. In fact, there is no closed-form expression for the interference distribution if interference protection is incorporated into analysis. This makes characterizing and understanding the interference behavior a challenging task. This tutorial shows detailed step-by-step interference characterization using stochastic geometry. It also shows the interference effect on important wireless communication key performance indicators such as error probability and transmission rate. Since interference coordination is elementary for several types of wireless networks, the analysis in this paper can be extended to other types of wireless networks that impose interference protection around receivers.} 

\subsection{Using SG for Cellular Networks}

SG was mostly confined to ad hoc and sensor networks to account for their intrinsic spatial randomness. In contrast, cellular networks were mostly assumed to be spatially deployed according to an idealized hexagonal grid. {Motivated by its tractability, attempts to promote SG to model cellular networks can be traced back to the late 90's \cite{cellular_bac1, cellular_tim}}. However, success was not achieved until a decade later \cite{tractable_app, martin_ppp, valid}. {The theoretical and statistical studies presented in \cite{tractable_app, martin_ppp, valid} revealed that cellular networks deviate from the idealized hexagonal grid structure and follows and irregular topology that randomly changes from one geographical location to another.} The authors in \cite{tractable_app} show that the signal-to-interference-plus-noise-ratio (SINR) experienced by users in a simulation with actual base station (BS) locations is upper bounded by the SINR of users in idealistic grid network, and lower bounded by the SINR of users in random network. Interestingly, the random network provides a lower bound that is as tight as the upper bound provided by the idealized grid network. However, the lower bound is preferred due to the tractability provided by SG. The authors in \cite{martin_ppp} show that the spatial patterns exhibited  by actual BS locations in different geographical places can be accurately fitted to random spatial patterns obtained via SG. Furthermore, the results in \cite{martin_ppp} confirm the tight lower bound provided by the random network to the users' SINR in simulations with actual BS locations. Finally, the authors in \cite{valid} show that the SINR in grid network converges to the SINR of random network in a strong shadowing environment.

Exploiting the tractability of SG, several notable {results are obtained for cellular networks}. For instance, the downlink baseline operation of cellular networks is characterized in \cite{tractable_app, valid, martin_ppp, eid_app, Gil_marco, Ginibre_martin, Ginibre_2}. Extensions to multi-tier case are provided in \cite{k_tier, shadowing_letter, Rate_marco, cellular_tim2, Tony_hetnet,  Assoc_singh, Het_Net_heath, Block_heath, Loc_het, Sayandev_1, Sayandev_2, Ali_pcp}. The uplink case is characterized in \cite{uplink_H,  Laila_Uplink, uplink_harpreet, uplink2_jeff, uplink_alamouri, Uplink_lett, uplink_zolfa, Giniber_uplink}. Range expantion and load balancing are studied in \cite{Offloading_jeff, Joint_singh, load_aware_harpreet, Halim_1, User_Assoc_Wei}. Relay-aided cellular networks are characterized in \cite{Relay_martin, Relay_marco}. Cognitive and self-organizing cellular networks are studied in \cite{cog_h, cog_h2, on_cog, on_ch, Cog_tony, Bennis, Bennis_2, Prap1, Prap2}. Cellular networks with multiple-input multiple-output (MIMO) antenna system are investigated in \cite{STBC_harpret, Mimo_load_harpreet, Mimo_ordering, Corr_MRC_harpreet, Corr2_MRC_harpreet, eid_Mimo, marco_Mimo, asymptotic_SE, Dist_antenna_Zhang, SDMA_Zhang, MIMO_heath, Dist_antenna_Wei, Massive_harpret, Laila_MIMO}. Cooperation, coordination, and interference cancellation in cellular networks are characterized in \cite{Gaus_approx1, coordination1_martin, Cooperation1_martin, CoMP_martin, CoMP_sakr, Pairwise_Bacc, Cancel_Tony, Int_align_tony, Interf_null_Zhang, Non_coh_singh, Coop1_heath, Int_cancel_heath}. Energy efficiency, energy harvesting, and BS sleeping for green cellular operation are studied in \cite{harvest_harpreet, Harvest_sakr, Harvest2_sakr, Sleep_tony, Ener_eff_tony, Energy_Zhang, Optiml_energy, Energy_partial}. Millimeter (mmW) based communication in cellular network is characterized in \cite{mmW_marco, mmW_jeff, mmW_heath, Esma}. In-band full-duplex communication for cellular networks is studied in \cite{Full_tony, FD_cellular2, h_FD, Harvesting2016AlAmmouri, Itsikiantsoa2015}. Interference correlation  across time and space in cellular networks is studied in \cite{correlated_HetNet, Corr_marco}. The additional interference imposed via underlay device-to-device (D2D) communication in cellular networks is characterized in \cite{D2D_h, Jeffery_D2D, D2D_tony, D2D2_tony, Spatial_d2d}. Mobility and cell boundary cross rate are studied in \cite{Mobility_jeff, Mobility_liang, Adve_1, Rabe}. Cloud radio access network and backhuling in cellular networks are studied in \cite{Backhaul_tony, Cloud_1, cloud_heath}. Last but not least, the physical layer security and secrecy in the context of cellular networks are  characterized in \cite{Security_harpreet, Security_Zhou, Security2_jeff, Bennis_sec}. By virtue of the results in \cite{Rabe, h_FD, Harvesting2016AlAmmouri, Itsikiantsoa2015, Laila_MIMO, Laila_Uplink, cellular_bac1, cellular_tim, tractable_app, valid, martin_ppp, eid_app, Gil_marco, Ginibre_martin, Ginibre_2, k_tier, shadowing_letter, Rate_marco, cellular_tim2, Tony_hetnet,  Assoc_singh, Het_Net_heath, Block_heath, Loc_het, Sayandev_1, Sayandev_2, Ali_pcp, uplink_H,  uplink_harpreet, uplink2_jeff, uplink_alamouri, Uplink_lett, uplink_zolfa, Giniber_uplink, Offloading_jeff, Joint_singh, load_aware_harpreet, Halim_1, User_Assoc_Wei, Relay_martin, Relay_marco, cog_h, cog_h2, on_cog, on_ch, Cog_tony, Bennis, Bennis_2, Prap1, Prap2, STBC_harpret, Mimo_load_harpreet, Mimo_ordering, Corr_MRC_harpreet, Corr2_MRC_harpreet, eid_Mimo, marco_Mimo, asymptotic_SE, Dist_antenna_Zhang, SDMA_Zhang, MIMO_heath, Dist_antenna_Wei, Massive_harpret, Gaus_approx1, coordination1_martin, Cooperation1_martin, CoMP_martin, CoMP_sakr, Pairwise_Bacc, Cancel_Tony, Int_align_tony, Interf_null_Zhang, Non_coh_singh, Coop1_heath, Int_cancel_heath, harvest_harpreet, Harvest_sakr, Harvest2_sakr, Sleep_tony, Ener_eff_tony, Energy_Zhang, Optiml_energy, Energy_partial, mmW_marco, mmW_jeff, mmW_heath, Esma, Full_tony, FD_cellular2, correlated_HetNet, Corr_marco, D2D_h, Jeffery_D2D, D2D_tony, D2D2_tony, Spatial_d2d, Mobility_jeff, Mobility_liang, Adve_1, Backhaul_tony, Cloud_1, cloud_heath, Security_harpreet, Security_Zhou, Security2_jeff, Bennis_sec}, SG based modeling for cellular networks is widely accepted by both academia and industry.

\subsection{Motivation \& Contribution}

Due to the expanding interest in SG analysis, it is required to have a unified and deep, yet simple, tutorial that introduces SG analysis for beginners in this field. Although there are excellent resources that present SG analysis for wireless networks \cite{baccelli_vol1, baccelli_vol2, martin_book, now_martin, now_jeff, survey_martin, moe_win, survey_h}, this tutorial is discriminated by introducing the error rate analysis and the focus on cellular networks. The monographs \cite{baccelli_vol1, baccelli_vol2, martin_book} present an advanced level treatment for SG and delve into details related to SG theory.  In \cite{baccelli_vol1, baccelli_vol2, martin_book, now_martin, now_jeff, survey_martin} many transceiver characteristics (e.g., modulation scheme, constellation size, matched filtering, signal recovery technique, etc.) are abstracted and the aggregate interference is treated as the sum of the powers of the interfering signals, and hence, the analysis is limited to outage probability, defined as the probability that the SINR goes below a certain threshold, and ergodic rate, defined by the seminal Shannon's formula. The tutorial in \cite{moe_win} delves into fine wireless communication details and presents error probability analysis. However, it is focused on ad hoc networks. The authors in \cite{survey_h} survey the SG related cellular networks literature without delving into the analysis details. 

In contrast to \cite{baccelli_vol1, baccelli_vol2, martin_book, now_martin, now_jeff, survey_martin}, the presented tutorial delves into the wireless communication aspects and exposes the necessary material from SG theory. Hence, it is more suited for those with wireless communication background. Furthermore, the presented tutorial is focused on the cellular network which is not the case in \cite{baccelli_vol1, baccelli_vol2, martin_book, now_martin, now_jeff, survey_martin, moe_win}. This tutorial also discusses the Gaussian signaling approximation that is taken for granted in the literature. To the best of the authors’ knowledge, this is the first time that the accuracy of the Gaussian signaling approximation in large-scale networks is discussed and analytically quantified. Finally, the tutorial elaborates the reasons for the pessimistic performance evaluation obtained via SG analysis and points out potential solutions.

This tutorial is organized as follows. Section~\ref{sec:overview} gives an overview of SG. Section~\ref{sec:model} introduces a basic system model which is used to introduce SG analysis. Section~\ref{sec:PPP} motivates the Poisson point process for network abstraction. Details about exact error probability analysis using the Poisson point process are presented in Section~\ref{sec:exact}. Section~\ref{sec:GCB} introduces the Gaussian signaling  assumption for the interfering symbols and discusses its effect. Section~\ref{sec:out} shows the abstracted outage probability analysis that is commonly used in the literature and highlights its implicit assumptions. Section~\ref{sec:advanced} relaxes basic assumptions in the basic system model, illustrates how SG analysis can be extended to capture realistic network setups, and provides numerical examples with discussions. Future research directions are highlighted in Section~\ref{sec:PPs} before the paper is concluded in Section~\ref{sec:conc}.

\textbf{Notation:} throughout the paper we use $\mathbb{P}\{\cdot\}$ to denote probability, $\EXs{X}{\cdot}$ to denote the expectation over the random variable $X$, $\EX{.}$ to denote the expectation over all random variables in $\{\cdot\}$, $\kappa_n(X)$ to denote the $n^{th}$ cumulant of the random variable $X$, $\overset{D}{=}$ to denote the equivalence in distribution, $\sim$ to denote the distribution, and $\mathcal{CN}(a, b)$ to denote the circularly symmetric complex Gaussian distribution with mean $a$ and variance $b$. The notations $f_X(\cdot)$,  $F_X(\cdot)$,  $\varphi_{X}(\cdot)$, and  $\mathcal{L}_{X}(\cdot)$  are used to denote the probability density function (PDF), the cumulative distribution function (CDF), the characteristic function (CF), and the Laplace transform\footnote{With a slight abuse of notation, we denote the LT of the PDF of a random variable $X$ by the LT of $X$. The LT of $X$ is defined as $\mathcal{L}_X(s) = \EX{e^{-sX}}$.} (LT), respectively, for the random variable $X$. The indicator function is denoted as $\mathbbm{1}_{\{\cdot\}}$, which takes the value $1$ when the statement $\{\cdot\}$ is true and $0$ otherwise. The set of real numbers is denoted as $\mathbb{R}$, the set of integers is denoted as $\mathbb{Z}$, the set of complex numbers is denoted as $\mathbb{C}$, in which the imaginary unit is denoted as $\jmath = \sqrt{-1}$, the magnitude of a complex number is denoted as $\vert \cdot \vert$, the complex conjugate is denoted as $(\cdot)^*$, and the Hermitian conjugate ${(\cdot)}^H$. The Euclidean norm is denoted as $ \left\|.\right\|$. $\gamma(a,b)=\int_0^b x^{a-1} e^{-x} dx $ is the lower incomplete gamma function,  $\text{erfc}(x)=\frac{2}{\sqrt{\pi}} \int_x^{\infty} e^{-t^2} \mathrm{d}t$ is the complementary error function, $(a)_n = \frac{\Gamma(a+n)}{\Gamma(a)}$ is the Pochhammer symbol, ${_1}F_1(a;b;x)=\sum_{n=0}^\infty \frac{(a)_n}{(b)_n}\frac{x^n}{n!}$ is the Kummer confluent hypergeometric function, and ${_2}F_1(a,b;c;x)=\sum_{n=0}^\infty \frac{(a)_n (b)_n}{(c)_n}\frac{x^n}{n!}$ is the Gauss hypergeometric function~\cite{abramowitz_stegun, integrals_book}. 


\section{Overview of SG Analysis} \label{sec:overview}

Before delving into the modeling details, we first give a broad overview of the SG  analysis as well as its outcome. In practice, cellular networks are already deployed and, for a given city, the locations of BSs are already known. However, in SG analysis, we are not concerned with the performance of a specific realization of the cellular network at a specific geographical location. Instead, we are interested in a general analytical model that applies on average for all cellular networks' realizations. For instance, if we want to analyze the effect of in-band full-duplex communication in cellular networks, instead of repeating the analysis for each and every geographical setup of the cellular networks, we can obtain general performance analysis, guidelines, and design insights that apply when averaging over all distinct realizations. Hence, from the analysis point of view, the locations of the BSs are considered unknown. Furthermore, following the studies in \cite{tractable_app, martin_ppp}, the locations of the BSs are considered random. Abstracting each BS location to a point in the Euclidean space, SG models the BS locations by a point process (PP)~\cite{spatial_pp_ch, SG_applications, pp_book_II, pp_book_I}, which describes the random spatial patterns formed by points in space. Then, according to the properties of the selected PP, the analysis is conducted. Note that, the selected PP, as will be shown later, and the associated analysis should capture the general properties of cellular network.

We are interested in the performance of a randomly selected user and/or the average performance of all users, i.e., the average performance of users in all locations. As discussed above, from the network perspective, we are interested in the average performance over all cellular network realizations. Such an average performance metric is denoted as {\em spatially averaged} (SA) performance, which is the main outcome from SG analysis.\footnote{{Formally, the spatial averaging technique and interpretation depend on the type of the PPs. If the PP is stationary (i.e., translation invariant) and spatially ergodic, then the averaging is w.r.t. the PP distribution and the result is location independent. On the other hand, if the PP is stationary but spatial ergodicity does not hold, then the expectation is done w.r.t. the Palm distribution of the PP \cite{martin_book} and the result is location independent. Finally, if the PP is neither stationary nor spatially ergodic, then the expectation is done w.r.t. the Palm distribution of the PP and the result is location dependent. Further discussion about this subject can be found in \cite[Chapter 8]{martin_book}.}} Examples of the SA performance metrics of interest in cellular networks are:
 
\begin{itemize}
\item \textbf{Outage probability:} defined as the probability that the SINR goes below a certain threshold (T), $\mathbb{P}\{{\rm SINR} <T\}$.
\item \textbf{Ergodic capacity:} defined by $\EX{\log(1+{\rm SINR})}$. {The ergodic
capacity measures the long-term achievable
rate averaged over all channel and interference (i.e., network realization) states~\cite{and_capacity}.}
\item \textbf{Symbol error probability (SEP):} defined as the probability that the decoded symbol is not equal to the transmitted symbol.
\item \textbf{Bit error probability (BEP):}  defined as the probability that the decoded bit is not equal to the transmitted bit.
\item \textbf{Pairwise error probability (PEP):} defined as the probability that the decoded symbol is $s_i$ given that $s_j$ is transmitted ignoring other possible symbols. 
\item \textbf{Handover rate} defined as the number of cell boundaries crossed over per unit time.
\end{itemize}

In SG analysis, we obtain expressions that relate the aforementioned performance metrics to the cellular network parameters and design variables. Such expressions are then used to understand the network behavior in response to the network parameters and/or design variables. This helps to obtain insights into the network operation and extract design guidelines. Note that, by SG analysis, the obtained design insights hold on average for all realizations of cellular networks. Hence, in light of the obtained expressions, communication system engineers can perform tradeoff studies and take informed design decisions before facing practical implementation issues. It is worth noting that the spatially averaged performance metrics obtained by SG analysis can be interpreted in different ways. {For instance, the average SEP can be interpreted as: i) the SEP averaged over all symbols for a randomly selected user, or ii) the SEP averaged over all symbols transmitted within the network. Similar interpretation applies for other performance metrics. }


\section{Baseline System Model \& Aggregate Interference Characterization } \label{sec:model}

\subsection{System Model}

As a starting point, we consider a baseline bi-dimensional single-tier downlink cellular network to introduce the basic SG analysis. Uplink and more advanced downlink scenarios are presented in Section~\ref{sec:advanced}. We assume that all BSs are equipped with single antennas and transmit with the same power $P$. {We also assume that each user equipment (UE) is equipped by a single antenna and is associated to its nearest BS}.\footnote{{Nearest BS association captures the traditional average radio signal strength (RSS) based association for single tier cellular networks when shadowing is ignored.}} Then the service area of each BS can be geometrically represented by a Voronoi-tessellation\cite{ voronoi, SG_applications} as shown in Fig.~\ref{fig:system_model}. We assume that BSs have saturated buffers and that every BS has a user to serve (saturation condition). Each BS maps its user data using a general bi-dimensional unit-power constellation $\mathbf{S}$ with $M$ symbols, denoted by $s_m=a_m \exp\{\jmath \theta_m \}$, where $m=1,2,3,....M$, such that, $\EX{\vert s_m\vert^2}=1$. All symbols from all BSs are modulated on the same carrier frequency and are transmitted to the corresponding users. The transmitted signal amplitude attenuates with the distance $r$ according to the power-law $r^{-\frac{\eta}{2}}$, where $\eta$ is an environmental dependent path-loss exponent. Multi-path fading is modeled via i.i.d. unit-variance circularly symmetric complex Gaussian random variables, denoted by $h$. We are interested in modeling the baseband signal received at an arbitrary user which is located $r$ meters away from his serving BS. The baseband signal (after proper down-conversion and low-pass filtering) can be represented as

\small
{\begin{align}
y &= \sqrt{P} s h r^{-\frac{\eta}{2}} + {i}_{agg} + n,
\label{basic_baseband0} 
\end{align}}
\normalsize

\noindent{where $s \in \bf{S}$ is intended symbol, $h\sim \mathcal{CN}(0,1)$, ${i}_{agg}$ is the aggregate interference amplitude experienced  from all interfering BSs and $n\sim\mathcal{CN}(0,N_o)$ is the noise. } 

\subsection{Network Abstraction}

The first step in the analysis is to choose a convenient PP to abstract the network elements (i.e., BSs and users). Then, the performance metrics of interest are expressed as functions of the selected PP. Last but not least, these functions can be evaluated using results from SG. Note that the term {\em ``convenient PP"} is used to denote a PP that balances a tradeoff between tractability and practicality.  As will be discussed later, a PP that is perfectly practical may obstruct the model tractability, and hence, approximations are usually sought. For the sake of complete presentation, we first shed light on the tractability issue of general PPs. Then we discuss the Poisson point process (PPP) approximation, which is usually used in the literature to retain tractability.

\begin{figure}[t]
	\begin{center}
	 \includegraphics[width=2.7in,height=2.7in]{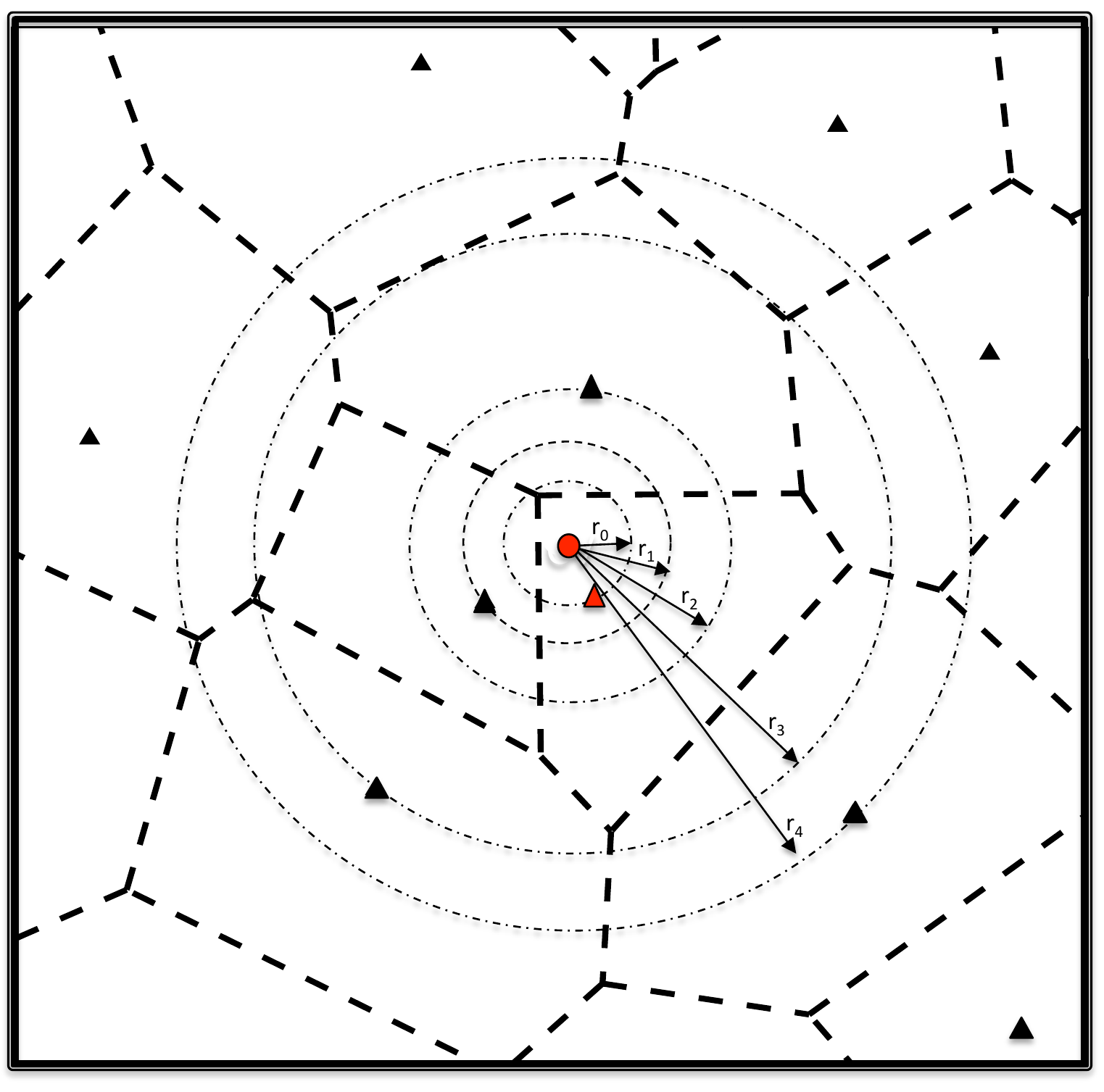}

	\end{center}
	\caption{ A realization from a cellular network in which the triangles represent the BSs, the red triangle/circle represents the test BS/user, and the bold dotted lines discriminate BSs footprints. The locations of other users are omitted for clarity. }
	\label{fig:system_model}
\end{figure}

Consider that the BS locations are abstracted by a general infinite two-dimensional PP $\Psi=\{x_i; i \in \mathbb{Z}\}$, where $x_i \in \mathbb{R}^2$ represents the coordinates of the $i^{th}$ BS.\footnote{Each BS is denoted by its location and the terms ``point" and ``BS" are used interchangeably. Infinite networks are considered for simplicity and due to the negligible contribution from far-away BSs to the aggregate interference. Also, the analysis can be easily modified to finite networks.} {At the moment, assume that the selected PP $\Psi$ perfectly reflects the correlation between the BSs belonging to the same service provider.} Repulsion (i.e., a minimum distance between BSs) is an important form of correlation that exists in cellular networks due to the network planning process. 

Without loss of generality, it is assumed that the points in the set $\Psi$ are ordered with respect to (w.r.t.) their distance from the test user and that the test user is located at an arbitrary origin, see Fig.~\ref{fig:system_model}.\footnote{The origin is an arbitrary reference point in $\mathbb{R}^2$ which is selected for the analysis. Usually the origin is selected to be the test user's location at which we evaluate the performance. {Note that the notion of {\em arbitrary origin} holds for stationary PPs only, otherwise, the analysis is location dependent.}} In this case, the distance from the $n^{th}$ BS to the test user is given by $r_n = \left\| x_n \right\|$, and the inequalities ($ r_{n-1} < r_n < r_{n+1}$) are satisfied with probability one. For the sake of simple presentation, we define the set $\tilde{\Psi} = \{ \left\|x_i \right\|; i \in \mathbb{Z}\} = \{ r_i; i \in \mathbb{Z}\}$, which contains the ordered BSs distances to the test user. Due to the RSS-based association, the test user is associated with the BS located at $x_0$ and the baseband received signal by the test user  can be expressed as

\small
{\begin{align}
y_0 &= \sqrt{P} s_{0} h_{0} r_0^{-\frac{\eta}{2}} + \underset{{i}_{agg}}{\underbrace{ \sum_{r_k \in \tilde{\Psi} \setminus r_0} \!\!\sqrt{P} s_{k} h_{k} r_k^{-\frac{\eta}{2}}}} +n, 
\label{basic_baseband} 
\end{align}}
\normalsize

\noindent{where $s_0 \in \bf{S}$ is the intended symbol, $s_k \in \bf{S}$ is the interfering symbol from the $k^{th}$ BS, $h_o \sim \mathcal{CN}(0,1)$ is the intended channel fading parameter, $h_k \sim \mathcal{CN}(0,1)$ is the interfering channel fading parameter. The random variables $s_k$ are independent and identically distributed (i.i.d). Ditto for the random variabes $h_k$. Moreover, the symbols and fading parameters are independent of one another. Note that $r_0$ is excluded from ${\tilde{\Psi}}$ as the serving BS does not contribute to the interference. For simplicity, we conduct the analysis for a given $r_0$ (i.e., assuming constant $r_0$). Then the condition on $r_0$ is relaxed in Section~\ref{sec:advanced}.  {It is worth noting that the received signal in the form of \eqref{basic_baseband} also applies to other types of wireless networks that impose an interference protection of $r_0$ around receivers.}}

By visual inspection of \eqref{basic_baseband} it is clear that the aggregate interference ${i}_{agg}$ {involves numerous sources of uncertainties}. Neither the number nor locations of the interfering BSs are known $\{x_k\}_{k\in \mathbb{Z}}$. In other words, the set of interfering BSs ${\Psi}\setminus x_0$ is a random set with infinite cardinality (or random cardinality for finite networks). In the following subsections, we show how to handle this randomness and statistically characterize the aggregate interference in \eqref{basic_baseband}. Before getting into the details, we need to emphasize that we do not aim to calculate an instantaneous value for   ${i}_{agg}$. Instead, we aim to characterize  ${i}_{agg}$ via its probability density function (PDF), characteristic function (CF), and/or moments.  As will be shown later, and also discussed in \cite{survey_h, moe_win}, the distribution of  ${i}_{agg}$ is not Gaussian. This is because the central limit theorem does not hold for  ${i}_{agg}$ as the sum in \eqref{basic_baseband} is dominated by the interference from nearby BSs.

\subsection{SG Analysis for ${i}_{agg}$ for General Point Process}

{Due to the many sources of involved uncertainties,} it is not feasible to characterize   ${i}_{agg}$  in an elementary manner (i.e., by evaluating the distribution for sum and product of random variables). Instead, we express the characterization parameter of interest (e.g., the moments of ${i}_{agg}$) as a function of the PP ($\tilde{\Psi}$), {then apply SG results to seek a solution}. As shown in Fig.~\ref{fig:gateways}, SG provides two main techniques that transform a function that involves all points in a PP to {an integral} over the PP domain, namely, Campbell's theorem and the probability generating functional (PGFL).\footnote{The PP domain is the smallest region in the Euclidean space that contains the PP.} However, as shown in the figure, a certain representation for the parameter of interest is mandatory to exploit these techniques. Since Campbell's theorem requires an expectation over a random sum, it can be directly used to calculate moments. On the other hand, the PGFL requires an expectation over a random product, which makes it suitable to calculate the characteristic function of ${i}_{agg}$. Campbell's theorem states that:

\begin{figure}[t]
	\begin{center}
	 \includegraphics[width=3in,height=2in]{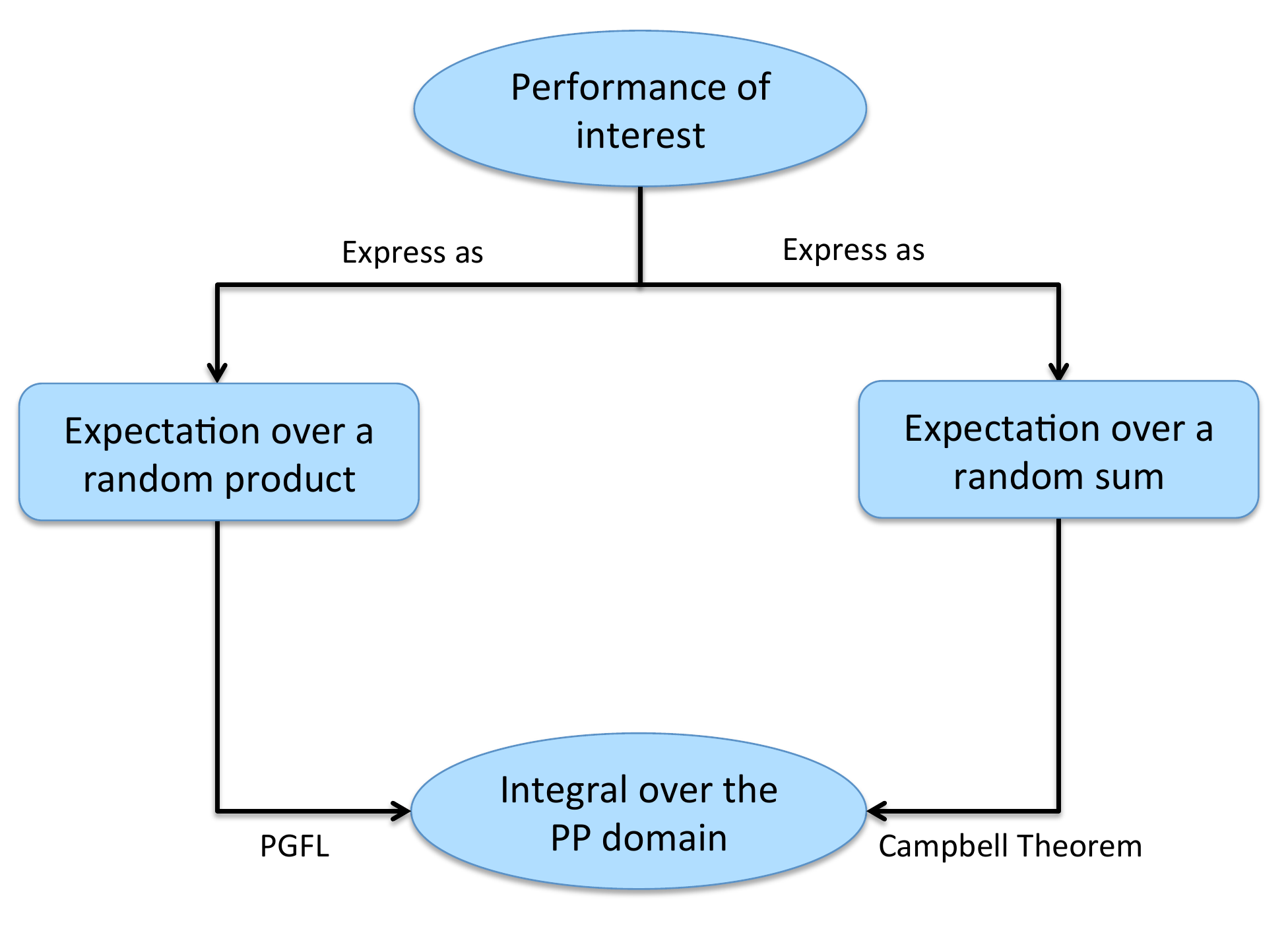}

	\end{center}
	\caption{ The two main techniques to handle the PP randomness. In the figure, PGFL denotes the probability  generating functional, which is defined in the text }
	\label{fig:gateways}
\end{figure}

\begin{theorem}[Campbell Theorem]

Let $\Phi$ be a PP in $\mathbb{R}^n$ and $f: \mathbb{R}^n \rightarrow \mathbb{R}$ be a measurable function, then 

{\begin{align}
\EX{ \underset{x_i \in \Phi} \sum f(x_i) } =& \int_{\mathbb{R}^n} f(x) \Lambda(\mathrm{d}x) 
\label{eq:campbell}
\end{align}}

\noindent{where $\Lambda(\mathrm{d}x)$ is the intensity measure of the PP $\Phi$ and $x \in \mathbb{R}^n$~\cite[Chapter 1.9]{SG_applications}}. {In case of PPs in $\mathbb{R}^2$, \eqref{eq:campbell} reduces to 
{\begin{align}
\EX{ \underset{x_i \in \Phi} \sum f(x_i)} {=}&  \int_{\mathbb{R}^2} f(x) \lambda(x) \mathrm{d}x
\label{eq:campbell_red}
\end{align}}

\noindent{where, $\lambda(x)$ is the two dimensional intensity function.}}
\end{theorem}

As shown in \eqref{eq:campbell}, Campbell's theorem transforms an expectation of a random sum over the PP to {an integral involving the PP intensity function}. Note that the integration boundaries represent the boundaries of the region where the PP exists. For example, in the case of the depicted cellular networks, the RSS association implies that no interfering BS can exist within the distance $r_0$. Applying Campbell's theorem, the mean value of the aggregate interference in \eqref{basic_baseband} can be expressed as

\small
{\begin{align}
\EX{ {i}_{agg}} &= \EX{ \sum_{r_k \in \tilde{\Psi} \setminus r_0} \!\!\sqrt{P} s_{k} h_{k} r_k^{-\frac{\eta}{2}} } \notag \\
&\overset{(a)}{=} \EXs{{\tilde{\Psi}\setminus{r_0}}}{\sum_{r_k \in \tilde{\Psi} \setminus r_0} \EXs{s_k, h_k}{\sqrt{P} s_{k} h_{k} r_k^{-\frac{\eta}{2}}} }\notag \\
&\overset{(b)}{=}  \int_0^{2\pi} \lambda \int_{r_0}^\infty \EXs{s, h} {\sqrt{P} s h  r^{-\frac{\eta}{2}}} r \mathrm{d}r \mathrm{d}\theta \notag \\
&= 2\pi \lambda \sqrt{P} \EX{s}  \EX{h}   \int_{r_0}^\infty  r^{-\frac{\eta}{2}+1} \mathrm{d}r \notag \\
&\overset{(c)}{=}  0
\end{align}}
\normalsize
\noindent{where $(a)$ follows from the linearity of the expectation operator and the independence between the BS locations, the transmitted symbols, and the fading gains; $(b)$ follows from Campbell's theorem in which the integration is computed in the polar coordinates ($\mathrm{d}x = r \mathrm{d}r \mathrm{d}\theta$) {with a constant intensity function $\lambda(x) = \lambda$;} $(c)$ follows from $\EX{s}=0$, due to the symmetry of the symbols' constellation and the equal probability of the interfering symbols; $(c)$ can also follow from $\EX{h}=0$, due to the zero mean complex Gaussian assumption of the channel fading.  }

Campbell's theorem can also be used to find the second moment of interference:

\begin{align}
&\EX{ \left(\underset{r_i \in \Phi} \sum f(x_i) \right)^2} = \EX{ \underset{x_i \in \Phi}{\sum} f^2(x_i)  +   \overset{x_i \neq y_i}{\underset{x_i,y_i \in \Phi}{\sum}} f(x_i)f(y_i) } \notag \\
&= \EX{ \underset{x_i \in \Phi}{\sum} f^2(x_i) } +  \EX{ \overset{x_i \neq y_i}{\underset{x_i,y_i \in \Phi}{\sum}} f(x_i)f(y_i)} \notag \\
&= \int_{\mathbb{R}^n} f^2(x) \Lambda(dx) + \int_{\mathbb{R}^n} \int_{\mathbb{R}^n} f(x) f(y)  \mu^{(2)}(dx,dy) 
\end{align}
\normalsize

\noindent{where $\mu^{(2)}(\mathrm{d}x,\mathrm{d}y)$ is the second factorial moment \cite{martin_book} of the PP $\tilde{\Psi}$, which is not always straightforward to compute.\footnote{For a homogeneous PPP with intensity $\lambda$, the second factorial moment is given by $\mu^{(2)}(\mathrm{d}x,\mathrm{d}y) = \lambda^2 \mathrm{d}x \mathrm{d}y$. } From the above discussion, it seems that Campbell's theorem is restricted to compute the first moment of the interference and can be extended to derive the second moment when $\mu^{(2)}$ can be obtained. Therefore, Campbell's theorem is not sufficient to fully characterize ${i}_{agg}$.}

The second technique to characterize ${i}_{agg}$ is through the PGFL \cite[Definition 4.3]{martin_book}. The PGFL converts random multiplication of functions over PP, in the form of $\mathbb{E}\left\{  \Pi_{x_i \in \Phi} f(x_i)  \right\}$, to an integral over the PP domain. Random multiplication is useful to obtain the CF of the aggregate interference ${i}_{agg}$  as

\small
{\begin{align}
&\varphi_{{i}_{agg}}(\boldsymbol{\omega}) \notag \\
& \overset{(a)}{=} \mathbb{E}\left\{ e^{\jmath \text{Re}\left\{\boldsymbol{\omega}^H{i}_{agg} \right\} } \right\} \notag \\ 
& = \mathbb{E}\left\{ e^{\jmath \omega_1 \mathrm{Re}({{i}_{agg}}) + \jmath \omega_2 \mathrm{Im}({{i}_{agg}})} \right\} \notag \\ 
&= \mathbb{E}\left\{ e^{\jmath \omega_1 \underset{r_k \in \tilde{\Psi} \setminus r_0 }{\sum} \!\! \mathrm{Re}(s_{k} h_{k}) \sqrt{P}  r_k^{-\frac{\eta}{2}} + \jmath {\omega_2} \underset{r_k \in \tilde{\Psi} \setminus r_0 }{\sum} \!\! \mathrm{Im}(s_{k} h_{k}) \sqrt{P}  r_k^{-\frac{\eta}{2}}} \right\} \notag \\
&= \mathbb{E}\left\{ e^{ \;\;  \underset{r_k \in \tilde{\Psi} \setminus r_0 }{\sum} \jmath \sqrt{\frac{{P}}{ r_k^{{\eta}}}}   (\omega_1 \mathrm{Re}(s_{k} h_{k}) + {\omega_2} \mathrm{Im}(s_{k} h_{k}) )} \right\} \notag \\
&= \EXs{\tilde{\Psi} \setminus r_0 }{ \prod_{r_k \in \tilde{\Psi} \setminus r_0} \EXs{s_k, h_k} { e^{\jmath \sqrt{\frac{{P}}{ r_k^{{\eta}}}}   (\omega_1 \mathrm{Re}(s_{k} h_{k}) + {\omega_2} \mathrm{Im}(s_{k} h_{k}) )} }} \notag \\
\label{eq:pgf1}
\end{align}}\normalsize

\noindent{where $(a)$ in \eqref{eq:pgf1} follows form the difinition of the CF for complex random variables \cite[Definition 10.1]{andersen1995}, and $\boldsymbol{\omega} =\omega_1 +  \jmath \omega_2$.}

In \eqref{eq:pgf1} we have the CF of the ${i}_{agg}$ represented as an expectation over a random product of a function of the process $\Psi$. Hence, we can use the PGFL of $\Psi$ to compute $\varphi_{{i}_{agg}}(\cdot)$. Unfortunately, expressions for the PGFL only exist for a limited number of PPs. Hence, in order to use the PGFL and characterize the aggregate interference via its CF, we should approximate the PP $\Psi$ via one of the PPs with known PGFL.

In conclusion, characterizing the aggregate interference from a general PP $\{{\tilde{\Psi}} \setminus r_0\}$ is not trivial and may not be analytically tractable. While the PP intensity is sufficient to obtain the mean of ${i}_{agg}$ associated with a general PP $\{{\tilde{\Psi}} \setminus r_0\}$ via Campbell's theorem, tractable expressions for higher order moments cannot be generally obtained. Furthermore. the PGFL does not exist for all PPs to characterize the aggregate interference via its CF.      Therefore, we have to resort to some approximation to maintain tractability. The most common and widely accepted approximation for $\{{\tilde{\Psi}} \setminus r_0\}$ is the PPP, which is discussed in the next section.  

\section{Poisson Point Process Approximation} \label{sec:PPP}

{Due to its simple PGFL expression, which leads to simple evaluation of \eqref{eq:pgf1}, the PPP is an appealing approximation for the interfering BSs locations. Furthermore, the PPP is stationary and spatially ergodic, which further simplifies the analysis. The validity of the PPP approximation is reinforced by the studies in \cite{tractable_app, martin_ppp, valid,  marco_fitting}, which show the close match between the SINR obtained from PPP analysis and the SINR obtained via simulations with actual cellular network topology.} The PPP is formally defined as

\begin{definition}[Definition 1] (Poisson point process (PPP)): A PP $\Phi = \{ x_i; i= 1,2, 3, \ldots \} \subset \mathbb{R}^d$ is a PPP if and only if the number of points inside any compact set $\mathcal{B} \subset \mathbb{R}^d$ is a Poisson random variable,  and the numbers of points in disjoint sets are independent.
\label{def:PPP}
\end{definition}

\begin{figure}[t]
	\begin{center}
	  \subfigure[The interference from spatially correlated interferers]{\label{fig:PPP_aprox1}\includegraphics[width=1.8in]{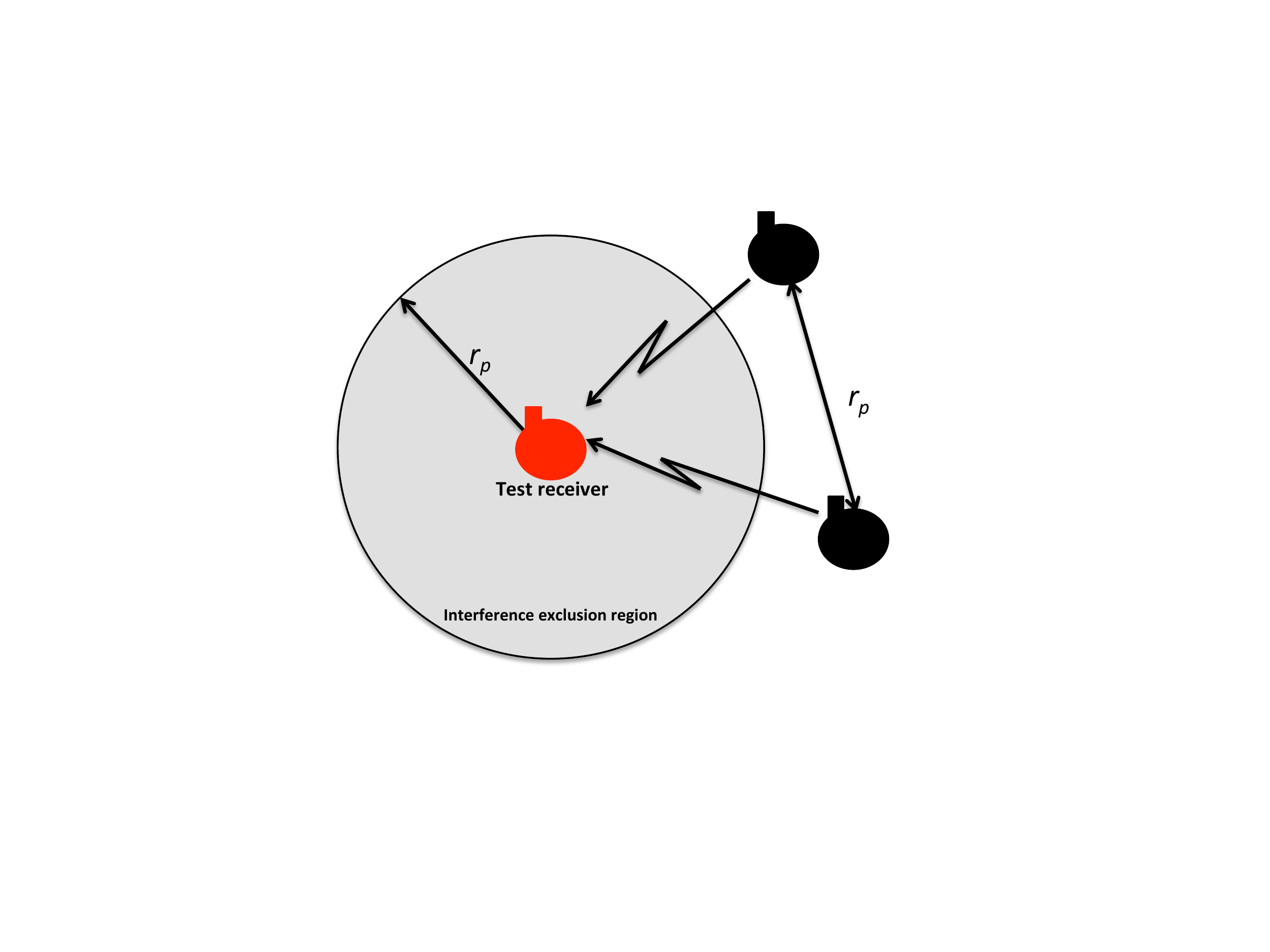}}
		\subfigure[the interference from an equi-dense spatially non-correlated interferers]{\label{fig:PPP_aprox2}\includegraphics[width=1.8in]{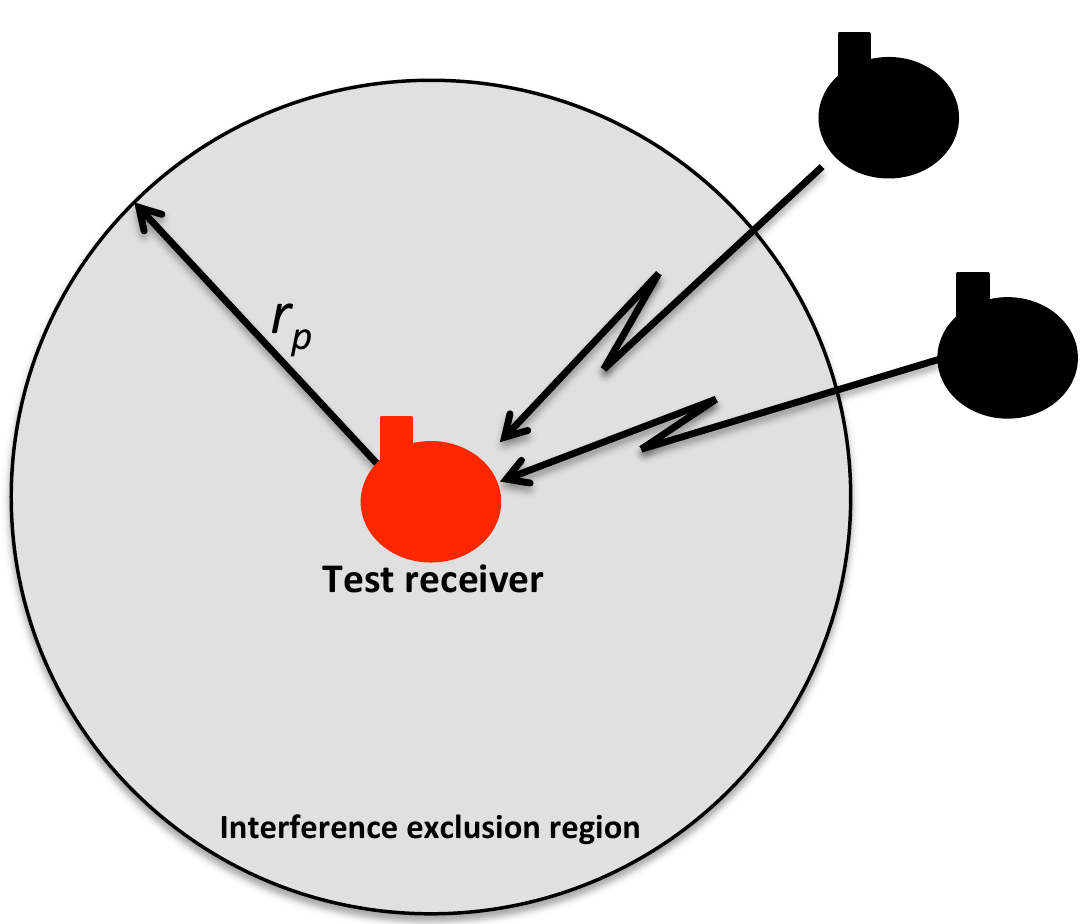}}
	\end{center}
	\caption{ Simple network model with one receiver (red circle in the center) and three trasmitters (black cricles).}
	\label{fig:PPP_aprox}
\end{figure}

Useful characteristics/expressions for the PPP are listed in \textbf{Appendix~\ref{PPP}}. From the PPP definition, one can see that the PPP does not impose any correlation between its points. Nevertheless, it gives an accurate estimate for the interference at a reference location in cellular networks, which impose repulsion {(i.e., a form of correlation)} between BS locations. Rigorous statistical studies for the PPP approximation {for BS locations from operators in UK} are conducted in \cite{martin_ppp, marco_fitting}. The authors show that 
the PPP accurately captures the SINR statistical behavior at any reference point. However, the authors in \cite{Correlated_Interferers} emphasize that the PPP accuracy depends on the magnitude of correlation between points with large spatial separations. Let $r_p$ be the minimum distance between two points in a repulsive PP, then the PPP approximation is accurate if and only if points separated by distances much larger than $r_p$ satisfies an asymptotic independence property~\cite{Correlated_Interferers}.\footnote{A stochastic PP where the points are prohibited to coexist within a certain distance $r_p$ from each other is denoted as a repulsive PP.} By virtue of the PPP approximation, tractable analysis is conducted for  several types of repulsive PPs in cellular networks domain \cite{Correlated_Interferers, uplink_H, uplink_harpreet, Jeffery_D2D, D2D_h, cog_h, Bennis} and ad hoc network domain \cite{New, nihar, Spec-eff, Mod-HCPP, letter, matern-paper2, dense_ppp}. Otherwise, only simulation studies are possible \cite{simulation_matern}. {The takeaway message from these studies is that the PPP approximation gives an accurate estimate for the interference associated with a repulsive PP. However, the intensity of the PPP used for approximation and the interference exclusion region around the test receiver should be carefully  chosen.} An intuitive explanation for the accuracy of the PPP approximation is as follows: approximating a repulsive PP with a PPP mainly neglects the mutual correlation (i.e., repulsion) among the points. However, the correlation with the test receiver is captured by the interference exclusion region.\footnote{{In the context of PPP, the interference exclusion region creates a void probability for the PPP points to exist in a certain region around the receiver. 
}} {For instance, Fig.~\ref{fig:PPP_aprox1} shows a PP that exhibits repulsion among its points as well as repulsion w.r.t. the test receiver, and Fig.~\ref{fig:PPP_aprox2} shows an approximation of the PPs in Fig.~\ref{fig:PPP_aprox1} by relaxing the  mutual repulsion between the interfering points. The PP in  Fig.~\ref{fig:PPP_aprox2} mimics the interference in  Fig.~\ref{fig:PPP_aprox1} on the test receiver because both have similar interference exclusion regions and same number of interferers. }


{From the above discussion, we emphasize that the PPP used to approximate a repulsive PP should be parameterized with two parameters, namely, the intensity $\lambda(x)$ and the  interference boundaries, as shown in Fig.~\ref{fig:int_bound}.\footnote{ Note that the intensity function $\lambda(x)$ is parameterized  by the location $x$ as the PPP is not necessarily homogeneous over the spatial domain.} Usually, the interference outer boundary is considered infinite due to the large-scale nature of the cellular network size and the negligible contribution from faraway BSs to the aggregate interference.} Hence, as long as the PPP approximation is considered, the intensity $\lambda(x)$ and inner interference boundary should be carefully estimated. In each of the presented case studies in Section~\ref{sec:advanced}, we will highlight how to estimate the intensity and the interference exclusion region.

\subsection{Interference Characterization}

In this subsection, we characterize the interference in a Poisson  field of interferers with exclusion region around the test receiver.  Let $z_k = h_k s_k$, then using the PGFL of a homogeneous PPP (see \eqref{PPP_PGFL} in \textbf{Appendix~\ref{PPP}}) the CF of the aggregate interference in \eqref{eq:pgf1} can be written as

\small
{\begin{align}
&\varphi_{{i}_{agg}}(\boldsymbol{\omega}) {=} \EXs{\tilde{\Psi}\setminus r_0}{ \prod_{r_k \in \tilde{\Psi}} \EXs{z_k} {e^{\jmath \sqrt{\frac{{P}}{ r_k^{{\eta}}}}   ( {\omega_1} \mathrm{Re}\{z_k\} + {\omega_2} \mathrm{Im}\{z_k\})} } } \notag \\
 &\overset{(a)}{=}  \exp \left\{-\int_{\mathbb{R}^2} \left(1- \EXs{z} { e^{ \jmath \sqrt{\frac{{P}}{ r^{{\eta}}}}   ( {\omega_1} \mathrm{Re}\{z_k\} + {\omega_2} \mathrm{Im}\{z_k\}) } } \right)  \Lambda(dr) \right\} \notag \\
 &\overset{(b)}{=}  \exp \left\{-2 \pi \lambda \int\limits_{r_0}^{\infty} \left(1- \mathbb{E}_{s}\EXs{z|s}{ e^{\jmath \sqrt{\frac{{P}}{ r^{{\eta}}}}  ( {\omega_1} \mathrm{Re}\{z\} + {\omega_2} \mathrm{Im}\{z\}) }}\right) r dr \right\} \notag \\
  & \overset{(c)}{=}  \exp \left\{-2 \pi \lambda  \int\limits_{r_0}^{\infty} \left(1-\EXs{s}{ e^{ - \frac{\left|\boldsymbol{\omega}\right|^2 P |s|^2 }{4 r^{{\eta}}}}}\right) r dr \right\}  \notag \\
& {\overset{(d)}{=}  \exp \Bigg\{ \frac{{ \pi \lambda}}{M}   \sum_m  \Bigg[ r_0^2\Bigg(1-  e^{-\frac{ \left|\boldsymbol{\omega}\right|^2 P |s_m|^2 }{4 r_0^{{\eta}}}} \Bigg) -   \Bigg. \Bigg. } \notag \\ 
& \quad \quad \quad \quad \quad { \Bigg. \Bigg.  \Bigg( \frac{ \left|\boldsymbol{\omega}\right|^2 P |s_m|^2}{4} \Bigg)^\frac{2}{\eta}  \gamma\Bigg(1-\frac{2}{\eta}, \frac{ \left|\boldsymbol{\omega}\right|^2 P |s_m|^2 }{4 r_0^{\eta}}\Bigg) \Bigg] \Bigg\}  }     
\label{eq:pgf2}
\end{align}}
\normalsize

\noindent{where $(a)$ follows from  the PGFL of the PPP (cf. \eqref{PPP_PGFL} in \textbf{Appendix~A}), $(b)$ follows from the RSS association (i.e., inner interference boundary is $r_0$) and substituting $z$ with $h s$, $(c)$ follows from the circularly symmetric Gaussian distribution of $h$, and $(d)$ is obtained by change of variables }

\begin{equation}
y=\frac{ \left|\boldsymbol{\omega}\right|^2 P |s_m|^2 }{4 r_0^{{\eta}}}, \notag
\end{equation}

\noindent{integration by parts, and the equi-probable symbol generation. The steps from $(a)$ to $(d)$ in \eqref{eq:pgf2} are the SG common steps to derive the characteristic function of the aggregate interference. Note that the CF in \eqref{eq:pgf2} is only valid for $\eta>2$. Otherwise (i.e., $\eta \leq 2$), the interference power is infinite almost surely \cite{now_martin}. Putting $r_0 = 0$ in \eqref{eq:pgf2}, we have }

\small
\begin{equation}
\left. \varphi_{{i}_{agg}}(\boldsymbol{\omega})\right|_{\tiny{r_0=0}}{=} \exp\left\{- \frac{\pi \lambda |\boldsymbol{\omega}|^\frac{4}{\eta} P^\frac{2}{\eta} \mathbb{E}\{|s|^\frac{4}{\eta}\} \Gamma\left(1-\frac{2}{\eta}\right)}{2^\frac{4}{\eta}} \right\}
\end{equation}
\normalsize 

\noindent{which is equivalent to \cite[equation (9)]{moe_win} given for ad hoc network. From \cite{moe_win}, it can be noted that with no exclusion region around the test receiver the aggregate interference (${i}_{agg}$) has an $\alpha$-stable distribution with infinite moments. The interference protection of $r_0$, provided by the basic cellular association, diminishes the interference distribution's heavy tail and results in finite interference moments. To study the moments of the interference, we manipulate \eqref{eq:pgf2} to express the CF of the aggregate interference in the following forms\footnote{ {The expression in \eqref{eq:pgf3} is obtained from \eqref{eq:pgf2} using the power series expansion of the incomplete Gamma function $\gamma(s,x)=x^s \Gamma(s)e^{-x}\sum_{k=0}^\infty \frac{x^k}{\Gamma(s+k+1)}$ and some mathematical manipulations.}}}

\begin{figure}[t]
	\begin{center}
\includegraphics[width=2.150in,height=2.15in]{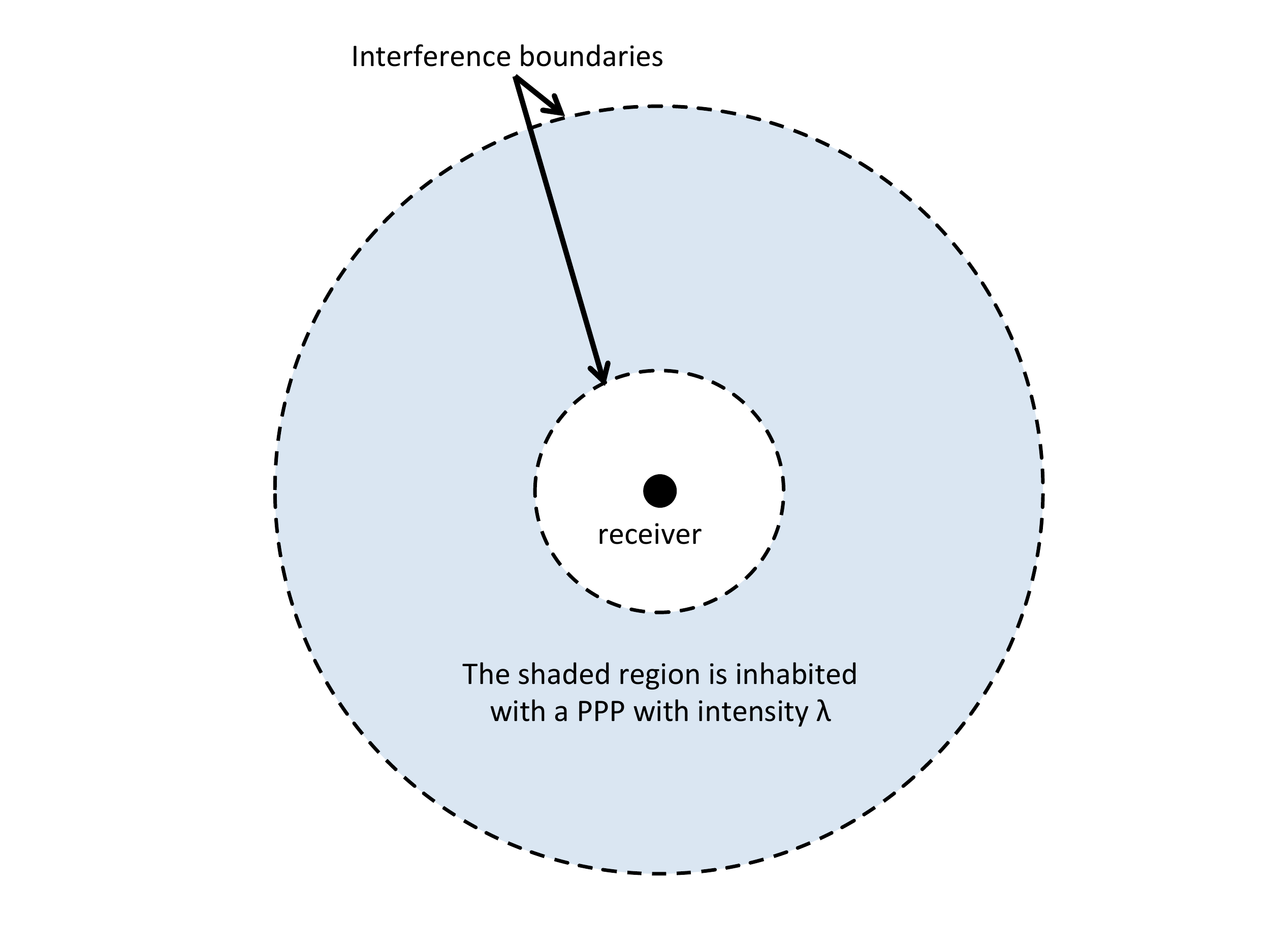}
	\end{center}
	\caption{ Intensity and boundaries are the main two parameters for the interference associated with PPP. }
	\label{fig:int_bound}
\end{figure}

\scriptsize
{\begin{align}
\varphi_{{i}_{agg}}(\boldsymbol{\omega}) &= \exp \left\{ \frac{{ \pi \lambda}}{M}   \sum_m  r_0^2 \left[ 1-  {_1}F_1\left(-\frac{2}{\eta},1-\frac{2}{\eta} ,- \frac{\left|\boldsymbol{\omega}\right|^2 P |s_m|^2 }{4 r_0^\eta}\right)  \right] \right\}  \notag \\
&= \exp \left\{{{2 \pi \lambda}}  r_0^2 \sum_{k=1}^{\infty} \frac{ (-1)^k \left|\boldsymbol{\omega}\right|^{2k} \mathbb{E}\{|s|^{2k}\} }{ (\eta k -2) k!}  \left( \frac{ P}{4 r_0^\eta}\right)^k  \right\} 
\label{eq:pgf3}
\end{align}}
\normalsize

\noindent{While the first form for $\varphi_{{i}_{agg}}(\boldsymbol{\omega})$ in \eqref{eq:pgf3} is compact and can be used to obtain the PDF of ${i}_{agg}$ via numerical inversion (e.g., Gil-Pelaez inversion theorem), the second form for $\varphi_{{i}_{agg}}(\boldsymbol{\omega})$ in \eqref{eq:pgf3} is easy to differentiate and obtain the moments of the ${i}_{agg}$. For the sake of simple presentation, we get the moments through cumulants. Following \cite{CFs_book}, the $n^{th}$ cumulant per dimension for the complex interference signal is defined as

\begin{align}
\kappa_n  &= \kappa_n\left(\mathrm{Re}\{{i}_{agg}\}\right) \notag \\
 &= \kappa_n\left(\mathrm{Im}\{{i}_{agg}\}\right) \notag \\
 &=\left.\frac{\partial^{n} \ln\left(\varphi_{{i}_{agg}}(|\boldsymbol{\omega}|)\right)}{j^n \partial{\omega_1}^n}\right|_{|\boldsymbol{\omega}|=0} \notag \\
 &=    \left.\frac{\partial^{n} \ln\left(\varphi_{{i}_{agg}}(|\boldsymbol{\omega}|)\right)}{j^n \partial{\omega_2}^n}\right|_{|\boldsymbol{\omega}|=0} 
\end{align}

\noindent{Note} that $\kappa_n\left(\mathrm{Re}\{{i}_{agg}\}\right)  = \kappa_n\left(\mathrm{Im}\{{i}_{agg}\}\right) $ because the interference signal is circularly symmetric as the CF in \eqref{eq:pgf3} is a function of $|\boldsymbol{\omega}|$ only. For notational convenience, we drop the real and imaginary parts and denote the per dimension $n^{th}$ cumulant as $\kappa_n \left({i}_{agg}\right)  = \kappa_n\left(\mathrm{Re}\{{i}_{agg}\}\right)  = \kappa_n\left(\mathrm{Im}\{{i}_{agg}\}\right) $. Using this notation, the per dimension cumulants are 

\begin{align}
\kappa_n\left({i}_{agg}\right)  &=\left\{\begin{matrix}
0, &  \text{$n$ is odd} \\ 
& \\
   \frac{{ \pi \lambda P^{\frac{n}{2}}} (n)_{n/2}}{ 2^{n-1}  (\frac{\eta n}{2}-2) }  r_0^{2-\frac{\eta n}{2}}  \mathbb{E}\left\{|s|^n \right\},  &   \text{$n$ is even} 
\end{matrix}\right.
\label{eq:moments1}
\end{align}

\noindent{From} the cumulants, the per-dimension moments can be obtained as

\begin{align}
\mathbb{E}\left\{\text{Re}\left\{{i}_{agg}^n\right\}\right\}  &=\left\{\begin{matrix}
0, &  \text{$n$ is odd} \\ 
& \\
   \kappa_2,  &   \text{$n$ = 2 } \\ & \\
   \kappa_4 + 3 (\kappa_2)^2,  &   \text{$n$ = 4 }  \\ 
   & \\
      \kappa_6 +15 \kappa_2 (\kappa_4 + (\kappa_2)^2),  &   \text{$n$ = 6 }  \\
   \vdots  &   \vdots 
\end{matrix}\right.
\label{eq:moments_k}
\end{align}

\noindent{The} expected aggregate interference power can be expressed as

\begin{align}
\mathbb{E}\left\{ {i}_{agg}  {i}_{agg}^H \right\} = \frac{{ 2 \pi \lambda P r_0^{2-{\eta}} }}{ {\eta}-2}.   
\label{eq:power}
\end{align}

\noindent{The} per dimension kurtosis, defined as $\frac{\kappa_4}{\kappa_2}$, is 

\begin{equation} \label{kurt}
\text{kur} = \frac{3 (\eta-2)^2 \mathbb{E}\left\{|s|^4\right\}}{4 \pi \lambda (\eta-1) r_0^2}.
\end{equation}

The characteristic function in \eqref{eq:pgf3}, the cumulants in \eqref{eq:moments1}, the expected interference power in \eqref{eq:power}, and the kurtosis in \eqref{kurt} show several interesting facts about the aggregate interference in the depicted system model:

\begin{itemize} 
\item The interference is circularly symmetric  complex random variable.
\item The interference is not Gaussian and the central limit theorem does not apply.
\item The interference power is infinite at $\eta = 2$ or $r_0 = 0$.
\item All interference cumulants, and hence moments, are finite for $\eta > 2$ and $r_0 >0$.
\item The interference distribution has a positive and finite kurtosis for $\eta > 2$ and $r_0 >0$, which indicates that it has a heavier tail than the Gaussian distribution.
\item The interference power decays with the interference exclusion radius at the rate of $r_0^{2-\eta}$ for $\eta>2$.
\item The interference increases linearly with the intensity $\lambda$ and power $P$.
\end{itemize}


\subsection{Numerical Results for ${i}_{agg}$}

In this section, we provide numerical results to visualize some properties of the aggregate interference in PPP networks with interference exclusion region. Also,  the numerical results show how the aggregate interference in PPP networks is related to the Gaussian and $\alpha$-stable distributions. Fig.~\ref{i_agg_PDF} plots the PDF of $\text{Re}\{{ i}_{agg}\}$ against the $\alpha$-stable and Gaussian PDFs with the same parameters.\footnote{Due to the circular symmetry of  ${i}_{agg}$ the PDF of $\text{Im}\{{i}_{agg}\}$ is similar to that of $\text{Re}\{{i}_{agg}\}$ given in Fig.~\ref{i_agg_PDF}. } The figure confirms the heavy (fast-decaying) tail of the  ${ i}_{agg}$ when compared to the Gaussian ($\alpha$-stable) PDF. With smaller exclusion distance $r_0$, the interference  ${ i}_{agg}$ approaches the $\alpha$-stable distribution. As $r_0$ increases,  ${i}_{agg}$ approaches the Gaussian distribution.

\begin{figure}[t]
	\begin{center}
	  \subfigure[]{\label{r150}\includegraphics[width=2.70in]{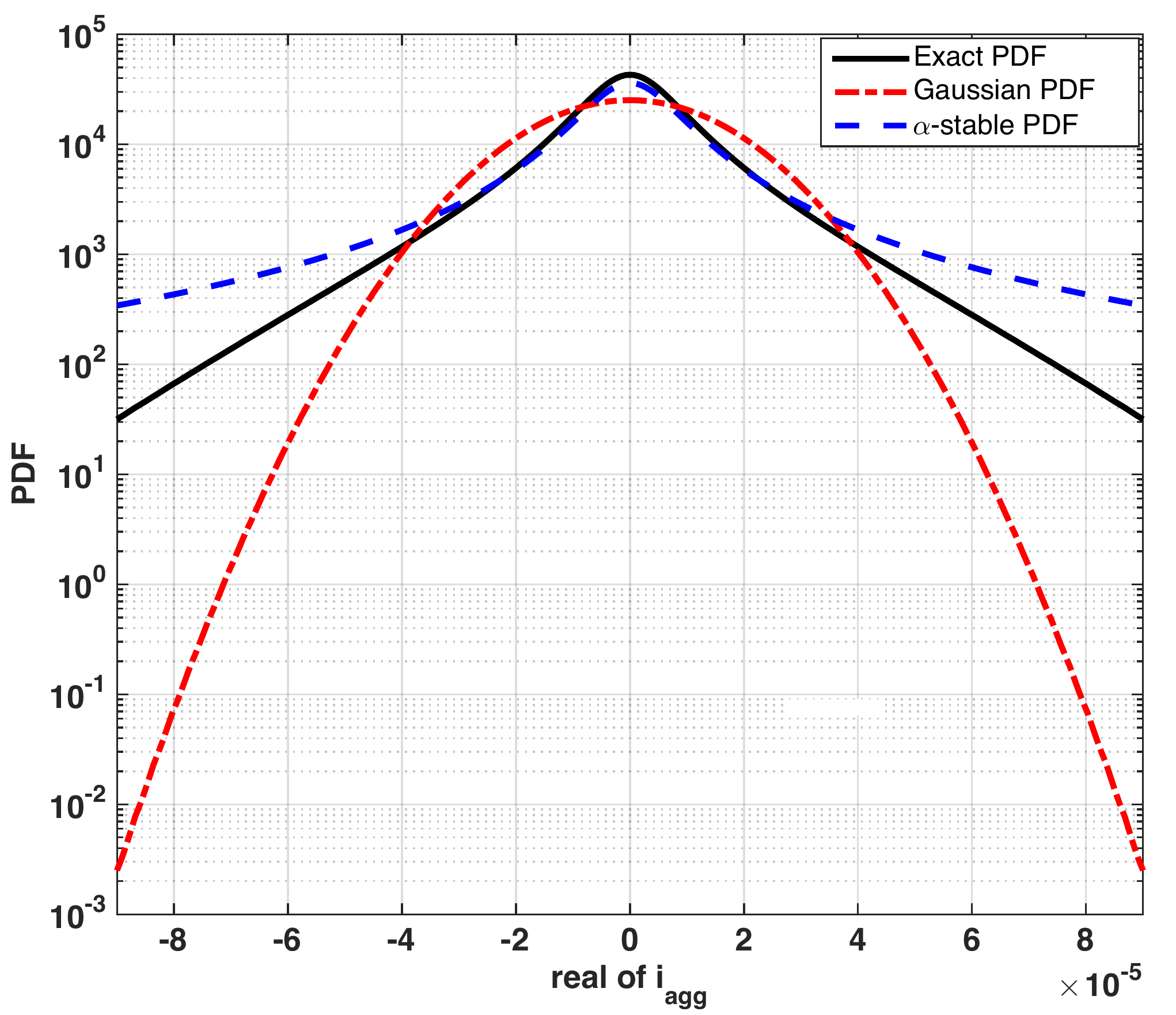}}
		\subfigure[]{\label{r250}\includegraphics[width=2.70in]{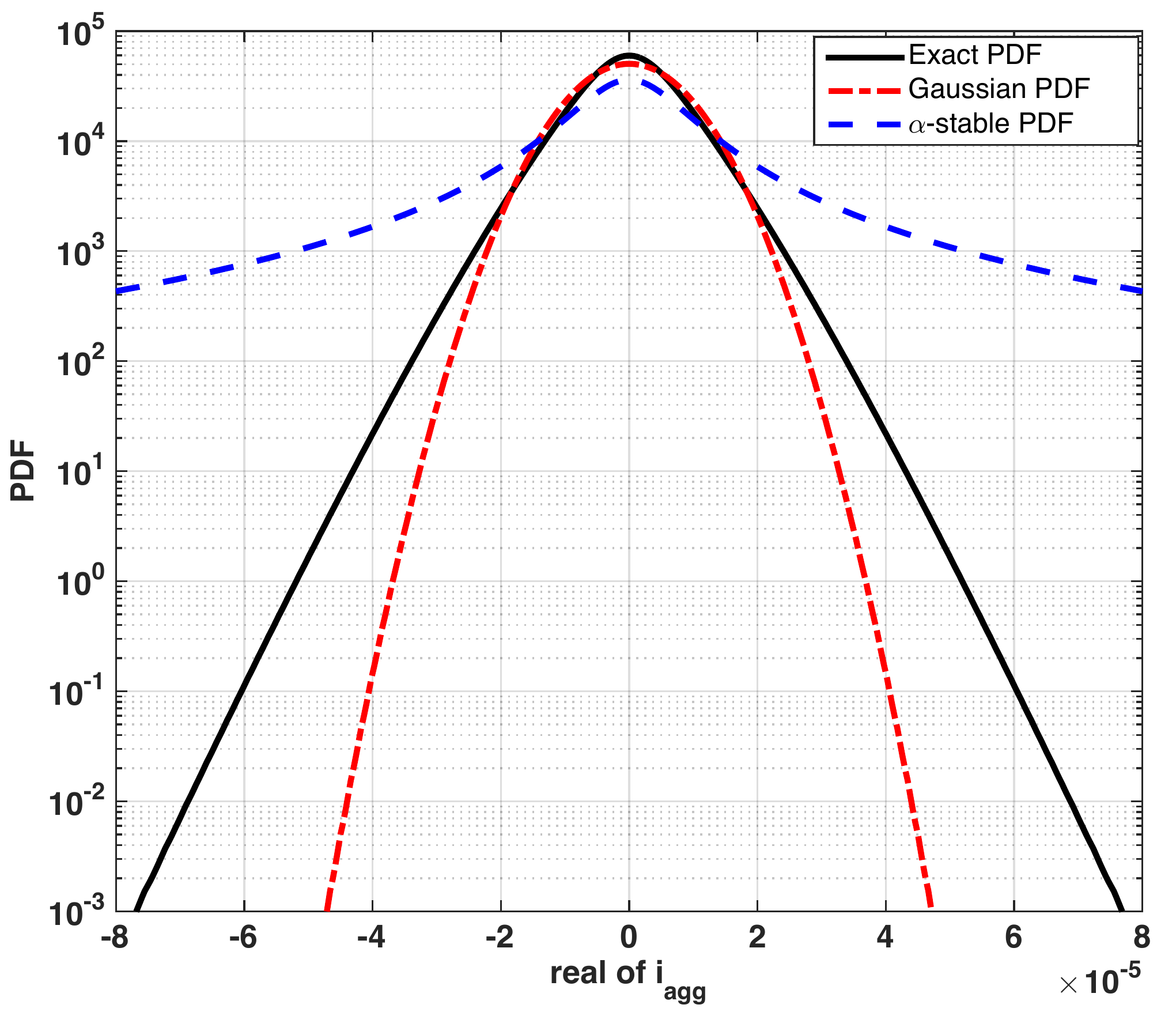}}
	\end{center}
	\caption{ The PDF of  $\text{Re}\{{ i}_{agg}\}$ obtained by numerically inverting \eqref{eq:pgf3} at $\lambda=1$ BS/km$^2$, $P=10$ W and $\eta=4$ for a) $r_0=250$ m and b) $r_0=500$ m .}
	\label{i_agg_PDF}
\end{figure}

To see the relation between  ${ i}_{agg}$, Gaussian, and $\alpha$-stable distributions more clearly, we plot the relative Kolmogorov--Smirnov (KS) distance in Fig.~\ref{KS_statistic}.\footnote{The KS distance measures the maximum distance between two CDFs $F_1(\cdot)$ and $F_2(\cdot)$, and is defined as ${\rm KS}=\underset{x}{\sup}|F_1(x)-F_2(x)|$.} Note that the KS statistic compares the entire CDFs and does not capture deviations in the tail probabilities. The figure shows that ${ i}_{agg}$ can neither be classified as Gaussian nor as $\alpha$-stable distributed. However, as $r_0$ increases, ${ i}_{agg}$ deviates from the $\alpha$-stable distribution and approaches the Gaussian distribution. Furthermore, as the intensity increases, the rate at which ${i}_{agg}$ deviates from the $\alpha$-stable distribution and approaches the Gaussian distribution increases. This is because increasing the intensity of interferers populates the interference boundary with more interferers, and hence, the central limit theorem becomes more applicable. On the contrary, at low intensity the interference is dominated by a small number of interferers, which renders the limit theorem inapplicable.

\begin{figure}[t]
	\begin{center}
	  \includegraphics[width=2.80in]{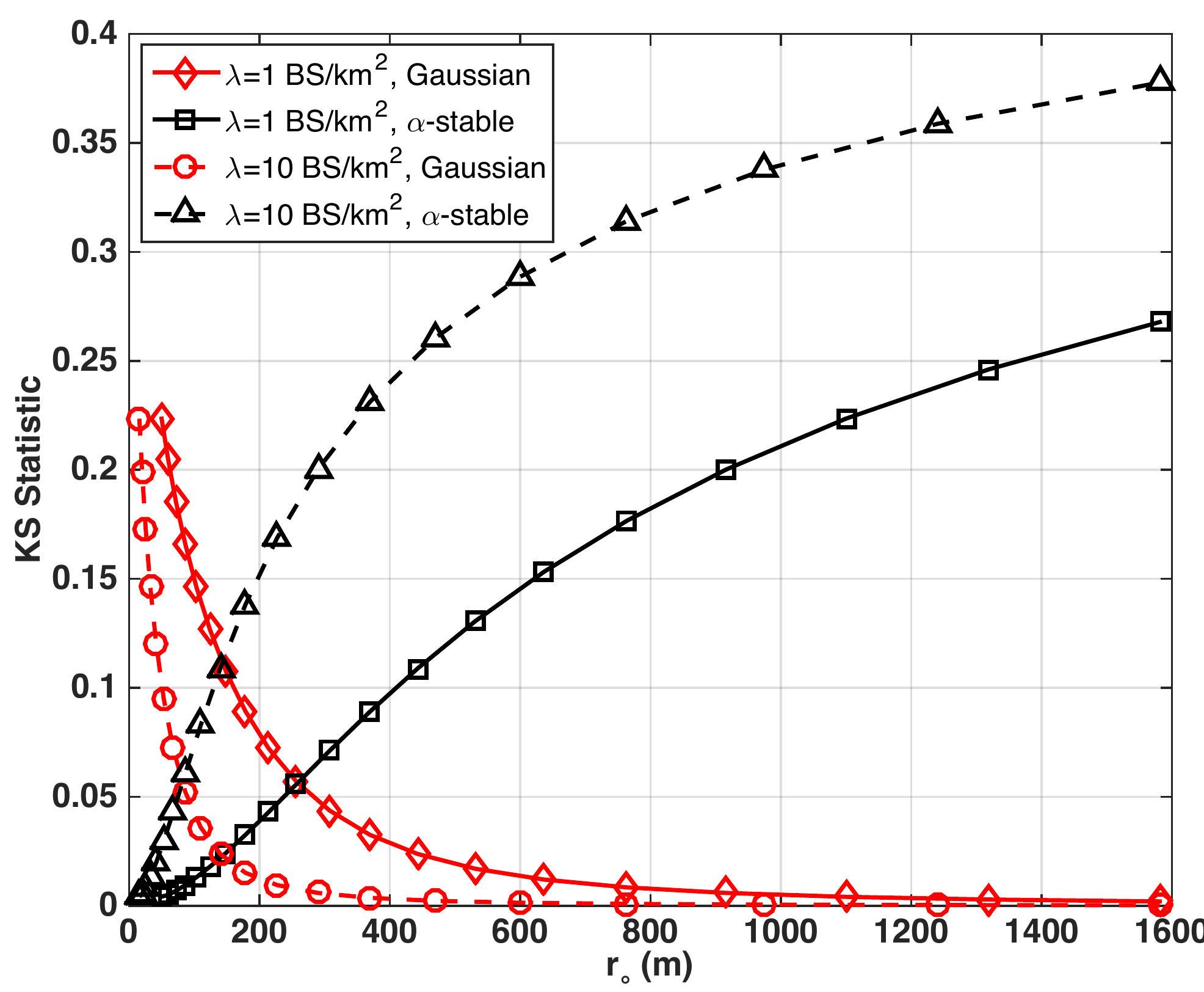}
	\end{center}
	\caption{ The KS statistic for ${i}_{agg}$ when compared to a Gaussian and $\alpha$-stable distributions.}
	\label{KS_statistic}
\end{figure}

\subsection{Section Summary}

In this section we motivate the use of PPP for network abstraction in order to obtain tractable results. We derive the CF of the aggregate baseband interference and compute its moments. We show that the aggregate interference in PPP networks with exclusion region around the receiver is neither Gaussian nor $\alpha$-stable distributed. Then, we highlight some characteristics of the baseband aggregate interference. In the next section, we will turn our focus to error probability performance. 


\section{Exact Error Probability Analysis} \label{sec:exact}

Error probability performance metrics are tangible measures used to fairly judge the performance of communication systems. Error probability includes bit error probability (BEP), symbol error probability (SEP), and pairwise error probability. In the context of wireless networks, error probability  performance has mainly been studied and conducted for additive white Gaussian noise (AWGN) or Gaussian interference channels \cite{slim_book}. In this section, we illustrate how to generalize the error probability analysis to the cellular networks domain.  Without loss of generality, we focus on the SEP, denoted by $\mathcal{S}$, for coherent maximum likelihood detector with $M$-QAM modulation scheme given by \cite[chapter 8]{slim_book},
{\begin{align} \label{SNR_ASEP}
\mathcal{S}= \displaystyle {w_{1} } \text{erfc}\left( \sqrt{\beta_1 \Upsilon}\right) + \displaystyle {w_{2} } \text{erfc}^2\left( \sqrt{\beta_2 \Upsilon}\right),
\end{align}}

\noindent{where , $w_{1} = 2 \frac{\sqrt{M}-1}{\sqrt{M}}$, $w_{2} = - \left(\frac{\sqrt{M}-1}{\sqrt{M}}\right)^2$, $\beta_1=\beta_2=\frac{3}{2(M-1)}$ are modulation-dependent weighting factors, and $\Upsilon$ is the  the signal-to-noise-ratio (SNR). It is worth noting that changing the factors $w_{1}$, $w_{2}$, $\beta_1$, and $\beta_2$, the SEP and BEP can be calculated for different modulation schemes and constellation sizes as shown in Table~\ref{table:modulations}.}

{All parameters in the SEP expression in \eqref{SNR_ASEP} are deterministic and the expression is derived based on the Gaussian distribution of the noise, in which the SNR $\Upsilon$ is the signal power divided by the variance of the Gaussian noise.} 
As shown in the previous section, the aggregate interference in the depicted  system model is not Gaussian, and hence, the cellular network does not maintain the same assumptions that are used to derive \eqref{SNR_ASEP}. Therefore, \eqref{SNR_ASEP} is not legitimate to calculate the SEP in cellular networks.

 One elegant solution to apply \eqref{SNR_ASEP} to study the error performance in the depicted large-scale cellular network is to represent the interference as a conditional Gaussian random variable \cite{moe_win_error,eid_app, Laila_Uplink}. Hence, treating interference as noise, \eqref{SNR_ASEP} is legitimate to calculate the conditional error probability. Then, an averaging step is required to obtain the unconditional error probability. This is known in the literature by the {\em Equivalent-in-Distribution} (EiD) approach, as it relies on the equivalence in distribution between the interference and the sum of randomly scaled Gaussian random variables. The rest of this section is devoted to illustrate the exact error performance characterization via the EiD approach. We first show how to represent the interference as a conditional Gaussian random variable, then we exploit this representation to calculate the average SEP (ASEP) in cellular networks. Note that SG provides the spatial average SEP, denoted as $\bar{\mathcal{S}}$.

\begin{center}
\begin{table}
\setlength\extrarowheight{7pt}
\caption{SEP Modulation-Specific Parameters for binary phase shift keying (BPSK), binary frequency shift keying ({BFSK}), quadrature phase shift keying (QPSK), M-quadrature amplitude modulation (M-QAM),  M-pulse amplitude modulation (M-PAM), differential encoded BPSK (DE-BPSK), minimum shift keying (MSK) }
\centering
\begin{tabular}{|c|c|c|c|}
\cline{1-4} 
& \multicolumn{3}{ |c| }{ \multirow{2}{*}{\textbf{Modulation Specific Parameters}}} \\ 
\textbf{Modulation}&   \multicolumn{3}{ |c| }{ } \\ \cline{2-4}
\textbf{Scheme}  & \multirow{2}{*}{$c$} & \multirow{2}{*}{$w_c$} & \multirow{2}{*}{$\beta_c$}  \\ 
 &  &  &   \\ \cline{1-4}
{\multirow{2}{*}{BPSK} } & 1 & {$ \frac{1}{2}$} &{$\frac{1}{2}$}  \\  \cline{2-4} 
                       & 2
   &$0$ & {-} \\  \cline{1-4}
   {\multirow{2}{*}{BFSK} } & 1 & {$ \frac{1}{2}$} &{$1$}  \\  \cline{2-4} 
                       & 2
   &$0$ & {-} \\  \cline{1-4}
   {\multirow{2}{*}{QPSK} } & 1 & {$ 1 $} &{$\frac{1}{2}$}  \\  \cline{2-4} 
                       & 2
   &{$- \frac{1}{4} $} & {$\frac{1}{2}$} \\  \cline{1-4}
{\multirow{2}{*}{M-QAM} } & 1 & {$ 2\frac{\sqrt{M}-1}{\sqrt{M}}$} &{$\frac{3}{2(M-1)}$}  \\  \cline{2-4} 
                       & 2
   &{$ - \left(\frac{\sqrt{M}-1}{\sqrt{M}}\right)^2$} & {$\frac{3}{2(M-1)}$} \\  \cline{1-4}
  {\multirow{2}{*}{M-PAM} } &
 1 & {$  \frac{{M}-1}{{M}}$} &{$\frac{3}{(M^2-1)}$}   \\  \cline{2-4} 
 &
  2 & 0 & -    \\  \cline{1-4} 
  {\multirow{2}{*}{M-PSK } } & 1  & 1 & ${\sin(\frac{\pi}{M})}$  \\ \cline{2-4}
             Upper-bound        &
2 & 0 &  -  \\ \cline{1-4}
{\multirow{2}{*}{DE-BPSK} } & 1  & 1 & 1  \\ \cline{2-4}
                      &
2 & 1 &  1 \\ \cline{1-4}
{\multirow{2}{*}{MSK} } & 1  & $\frac{1}{2}$ & 1  \\ \cline{2-4}
                      &
2 & 0 &  - \\ \cline{1-4}
\end{tabular}
\label{table:modulations}
\end{table}
\end{center}

%

\subsection{Conditional Gaussian Representation for Interference}

The conditional Gaussian representation of the interference is obtained by exploiting the fact that matching characteristic functions implies equivalent distributions. The authors in \cite{eid_app} show that an equivalent-in-distribution representation for ${{i}_{agg}}$ can be expressed as $ {{i}_{eq}}= {{\sum_{q=1}^{\infty} \sqrt{\mathcal{B}_{q} }G_{q}}}$, where $\mathcal{B}_{q}$ is a real random variable with Laplace transform (LT) $\mathcal{L}_{\mathcal{B}_{q}}\left( s\right)=  e^{(-s)^q}$ and $G_{q}$ is a {zero-mean} circularly symmetric complex Gaussian random variable with variance $\sigma_q^2$. Note that the selected LT of $\mathcal{B}_{q}$ here is different than that in \cite{eid_app} to avoid negative variances.} The equivalence in distribution is proved by showing that ${{i}}_{eq}$ has a matching characteristic function to \eqref{eq:pgf3}. The CF of ${{i}}_{eq}$ is obtained as

\small
\begin{align}
\mathbb{E}\left\{e^{\jmath \boldsymbol{\omega} {i}_{eq} }\right\} &= \EXs{B_{q} }{ \EXs{G_{q}}{  e^{\jmath \boldsymbol{\omega} \sum_{q=1}^\infty \sqrt{B_{q}} G_{q}} }} \notag \\
&=    \exp\left\{ \sum_{q=1}^\infty \left(- \frac{\sigma^2_{q} |\boldsymbol{\omega}|^2}{4}\right)^q \right\}. 
  \label{eq:equi}
\end{align}
\normalsize

\noindent{Comparing} \eqref{eq:equi} with \eqref{eq:pgf3}, one can see that if $\sigma^2_q$ is selected as 

\begin{align}
\sigma^2_{q} = \left( \frac{2 \pi\lambda r_0^{2-\eta q} P^q \mathbb{E}\left\{ |s|^{2q}\right\}}{ (\eta q-2) q!} \right)^{\frac{1}{q}},
\label{eq:sigma_q}
\end{align}

\noindent{then} \eqref{eq:equi} and \eqref{eq:pgf3} have {equivalent} CFs, and hence, equivalent distributions. Exploiting the Gaussian representation for ${i}_{agg}$, the baseband received signal at the test UE can be rewritten as

\begin{align} \label{eq_dist}
y_0 &\overset{D}{=} \sqrt{P} s_0 h_0 r_0^{-\frac{\eta}{2}} +  \underset{{{i}}_{eq}}{\underbrace{\sum_{q=1}^{\infty} \sqrt{{B}_{q} }G_{q}}} + n,  \notag \\
&= \sqrt{P} s_0 h_0 r_0^{-\frac{\eta}{2}} + \tilde{n},
\end{align}

\noindent{where $\tilde{n}={i}_{eq} + n$. Since $\left\{G_{q}\right\}_{q=1}^{\infty}$ are independent circularly symmetric Gaussian random variables, conditioning on $\left\{B_{q}\right\}_{q=1}^{\infty}$,  the lumped interference plus noise term ($\tilde{n}$) is a circularly symmetric complex Gaussian random variable with a total variance of $(\sum_{q=1}^\infty B_q \sigma_q^2 + N_0)$. This representation is the key that merges SG analysis and the rich literature available on AWGN based performance analysis. Since $\tilde{n}$ in \eqref{eq_dist}  is conditional Gaussian, the SNR formulas in AWGN channels can be extended to model error performance in cellular networks, as shown in the next subsection. }

\subsection{ASEP with Non-Gaussian Cellular Interference}

Let $\Xi=h_0 \cup \left\{\mathcal{B}_{q}\right\}_{q=1}^{\infty}$, then following \cite{slim_book}, the conditional average {\rm SINR}, when treating interference as noise, is given by
\begin{align}
\!\! \bar{\Upsilon}(r_0 \vert \Xi) & =  \frac{\EXs{s_0}{ \mathbb{E} \left\{ y_0 \right\} \mathbb{E} \left\{ y_0^*\right\} } }{ \mathbb{E} \left\{ y_0 y_0^*\right\}  - \mathbb{E} \left\{ y_0\right\} \mathbb{E} \left\{ y_0^*\right\}}  \notag \\
&= \frac{P |h_0|^2 \mathbb{E}\{|s_0|^2\} r_0^{-\eta}}{N_0+ \sum_{q=1}^{\infty} {\mathcal{B}_{q} \sigma_{q}^2}}\notag \\
&= \frac{P |h_0|^2 r_0^{-\eta}}{N_0+ \sum_{q=1}^{\infty} {\mathcal{B}_{q} \sigma_{q}^2}}.
\label{eq:cond_SINR}
\end{align}

\noindent{Conditioning} on $\Xi$, the SINR in \eqref{eq:cond_SINR} is similar to the legacy SNR in \eqref{SNR_ASEP} but with increased noise variance of $N_0+ \sum_{q=1}^{\infty} {\mathcal{B}_{q} \sigma_{q}^2}$. Hence, rewriting \eqref{SNR_ASEP}, the ASEP with interference can be expressed as

\small
\begin{align} \label{ASEP_cond10}
&\bar{\mathcal{S}}(r_0 \vert \Xi)= \displaystyle {w_{1} } \text{ erfc}\left( \sqrt{\beta_1  \bar{\Upsilon}(r_0 \vert \Xi)}\right) + \displaystyle {w_{2} } \text{ erfc}^2\left( \sqrt{\beta_2  \bar{\Upsilon}(r_0 \vert \Xi)}\right).
\end{align}
\normalsize

Let $\zeta=\frac{\sum_{q=1}^{\infty} {\mathcal{B}_{q} \sigma_{q}^2}}{P r_0^{-\eta}}$. Then, the unconditional ASEP can be obtained by an additional averaging step as\footnote{Unconditional with respect to the elements of $\Xi$, however, we are still conditioning on $r_0$.}

\scriptsize
\begin{align} \label{ASEP_cond11}
&\bar{\mathcal{S}}(r_0)= \displaystyle {w_{1} }\mathbb{E}\left\{ \text{erfc}\left( \sqrt{\beta_1  \bar{\Upsilon}(r_0 \vert \Xi)}\right) \right\} + \displaystyle {w_{2} } \mathbb{E}\left\{\text{erfc}^2\left( \sqrt{\beta_2  \bar{\Upsilon}(r_0 \vert \Xi)}\right)\right\} \notag \\
&= \displaystyle {w_{1} }\mathbb{E}\left\{ \text{erfc}\left( \sqrt{\frac{|h_0|^2}{\frac{N_0 r_0^\eta}{P \beta_1 }+ \frac{\zeta}{\beta_1 }}} \right) \right\} + \displaystyle {w_{2} } \mathbb{E}\left\{\text{erfc}^2\left( \sqrt{\frac{|h_0|^2}{\frac{N_0 r_0^\eta}{P \beta_2 }+ \frac{\zeta}{\beta_2 }}} \right)\right\} \notag \\
&\overset{(a)}{=} \sum_{c=1}^2 w_c \left(1 - \frac{c}{\sqrt{\pi}} \int\limits_0^{\infty} \frac{\text{erfc}(\sqrt{z} \mathbbm{1}_{\{c=2\}})}{\sqrt{z}} e^{-z\left( 1+ \frac{N_0 r_0^\eta}{P \beta_c}\right)} \mathcal{L}_{\zeta}\left(  \frac{z}{\beta_c}\right) \mathrm{d} z \right).
\end{align}
\normalsize

\noindent{The equality $(a)$ follows from the lemma proposed in \cite{hamdi_useful_tech}, which is also given in \textbf{Appendix~\ref{app:useful}}.\footnote{Let $Y$ be a Gamma random variable, \cite{hamdi_useful_tech} shows that the expectation in the form of $\mathbb{E}\left[ \text{erfc} \left(\sqrt{{Y}/{X}} \right) \right]$ and $\mathbb{E}\left[ \text{erfc}^2 \left(\sqrt{{Y}/{X}} \right) \right]$ can be computed in terms of the LT of $X$. In our case, $Y=|h_0|^2$ is an exponential distribution which is a special case of gamma distribution.} The LT of $\zeta$ is given in \textbf{Appendix~\ref{app:LT_zeta}}. Substituting the LT of $\zeta$ into \eqref{ASEP_cond11}, the ASEP is characterized via Theorem~\ref{theo:ASEP_EiD} given at the top of the next page.}

\begin{figure*}

\begin{theorem}
\label{theo:ASEP_EiD}
Consider cellular network  modeled via a PPP with intensity $\lambda$ in Rayleigh fading environment with universal frequency reuse and no intra-cell interference. Then, the downlink ASEP with $M$-QAM modulated signals for a user located at the distance $r_0$ away from his serving BS, is expressed as

\small
\begin{align} \label{eq:ASEP_exact_r0}
\bar{\mathcal{S}}(r_0)
&= \sum_{c=1}^2 w_c \left( 1 - \frac{c}{\sqrt{\pi}}  \int\limits_0^{\infty} \frac{ \text{erfc}(\sqrt{z} \mathbbm{1}_{c=2}) }{\sqrt{z}} \exp\left\{-z\left( 1+ \frac{N_0 r_0^\eta}{P \beta_c}\right) - \pi\lambda r_0^{2}  \left( \frac{1}{M} \sum_{m=1}^{M}  {}_{1}F_1\left( -\frac{2}{\eta}; 1-\frac{2}{\eta}; - \frac{z |s_m|^2}{\beta_c} \right)-1   \right) \right\}  \mathrm{d} z \right).
\end{align}
\normalsize
\hrulefill

\end{theorem}

\end{figure*}

\subsection{Section Summary}

This section explains the steps for exact ASEP calculation via the EiD approach. The EiD approach is used to express the aggregate interference as a conditional Gaussian random variable and use the available AWGN based ASEP expressions. The EiD approach proceeds as follows: 
\begin{enumerate}
\item \textbf{Interference Characterization:} Use SG to obtain the characteristic function of the aggregate complex interference signal ${i}_{agg}$ in the form of \eqref{eq:pgf3}.
\item \textbf{Gaussian Representation:} Express the interference via the infinite sum $ {{i}_{eq}}= {{\sum_{q=1}^{\infty} \sqrt{\mathcal{B}_{q} }G_{q}}}$ and calculate the variances $\{\sigma_q\}_{q=1}^{\infty}$ that matches the CFs in \eqref{eq:pgf3} and \eqref{eq:equi} to ensure equivalence in distribution.
\item \textbf{Conditional Analysis:} Condition on $\{B_q\}_{q=1}^{\infty}$ and obtain the conditional ASEP via AWGN based expression with the conditional SINR as in \eqref{ASEP_cond10}.
\item \textbf{Deconditioning:} Decondition over the non-Gaussian random variables to obtain ASEP as in \eqref{ASEP_cond11}. 
\end{enumerate}

Although exact, the ASEP expression given in \eqref{eq:ASEP_exact_r0} is quite complex and computationally intensive due to the integral over an exponential function with a sum of hypergeometric functions in the exponent. Furthermore, the complexity of the EiD approach increases for advanced system models with Nakagami-m fading and/or multiple antennas~\cite{eid_Mimo}. Therefore, approximations and more abstract analysis are conducted in the literature to seek simpler and more insightful performance expressions, as will be shown in the next sections.

\section{Gaussian Signaling Approximation} \label{sec:GCB}

The complexity of the EiD approach is due to the fact that it statistically accounts for the transmitted symbol by each interfering source. Abstracting such information highly facilitates the analysis. Instead of assuming that each interfering transmitter maps its data using a distinct constellation,  it can be assumed that each transmitter randomly selects its transmitted symbol from a Gaussian constellation with unit variance.\footnote{Note that if the interfering BSs are coded and operating close to capacity, then the signal transmitted by each is Gaussian~\cite{moe_win_error_2}. However, we are interested in the Gaussian signaling as an approximation for the interfering symbols which are drawn from the distinct constellation $\bf{S}$.} As shown in this section, the Gaussian signaling approximation directly achieves the conditional Gaussian representation for aggregate interference. Hence, the ASEP expressions for AWGN channels are legitimate to be used. Furthermore, the Gaussian signaling approximation circumvents the complexity of the EiD approach without compromising the modeling accuracy. 

In this section, we first validate the Gaussian signaling approximation and show that it does not change the distribution of the aggregate interference. We also show its effect on the interference moments. Then, we show the approximate error probability performance with the Gaussian signaling approximation. 

\subsection{Validation} \label{sec:GCB_valid}


The Gaussian signaling approximation  does not approximate the aggregate interference by a Gaussian random variable. Instead, it assumes that each interferer chooses a symbol $s$ from complex Gaussian distribution such that $\mathbb{E}\{|s|^2\}=1$. Then, the transmitted symbol by each interfering BS $x_i$ experiences the location dependent path-loss $r_i^{-\eta/2}$ and encounters independent random fading $h_i$ before reaching the test receiver. The main idea in the Gaussian signaling approximation is to abstract the {information} carried in the aggregate interference to facilitate the error rate analysis. The baseband signal representation in the Gaussian signaling approximation is similar to \eqref{basic_baseband}, except that $s_k$ has a complex Gaussian distribution with a unit variance. Following the same steps as in \eqref{eq:pgf2} and \eqref{eq:pgf3}, the CF of the approximate aggregate interference $\hat{{i}}_{agg}$ is obtained as

\scriptsize
{\begin{align}
\varphi_{\hat{{i}}_{agg}}(\boldsymbol{\omega})   &= \exp \left\{-\frac{\pi \lambda P |\boldsymbol{\omega}|^2}{2(\eta-2) r_0^{\eta-2}}  {_2}F_1 \left(1, 1-\frac{2}{\eta}; 2-\frac{2}{\eta}; - \frac{P \left|\boldsymbol{\omega}\right|^2}{4 r_0^\eta} \right)  \right\}  \notag \\
  &= \exp\left\{{2 \pi \lambda  r_0^{2}} \sum_{k=1}^\infty \frac{(-1)^k \left|\boldsymbol{\omega}\right|^{2k}}{\eta k-2}  \left(\frac{P}{4 r_0^\eta} \right)^{k}  \right\} 
\label{eq:GCB}
\end{align}}
\normalsize

\noindent{Equation} \eqref{eq:GCB} shows that the aggregate interference signal is circularly symmetric, which implies that the distribution and moments of the real and imaginary parts of $\hat{{i}}_{agg}$ are identical. Following the same notation in \eqref{eq:moments1}, we drop the real and imaginary parts, and denote the per dimension $n^{th}$ cumulant as $\kappa_n \left(\hat{{i}}_{agg}\right)= \kappa_n\left(\mathrm{Re}\{\hat{{i}}_{agg}\}\right)  = \kappa_n\left(\mathrm{Im}\{\hat{{i}}_{agg}\}\right) $. Using this notation, the cumulants of $\hat{{i}}_{agg}$ are given by

\begin{align} \label{eq:moments2}
\kappa_n(\hat{i}_{agg}) &=\left\{\begin{matrix}
0, &  \text{$n$ is odd} \\ 
{}&\\
 \frac{ \pi \lambda P^{\frac{n}{2}} n! }{2^{n-1} (\frac{\eta n}{2}-2)} r_0^{2-\frac{n \eta}{2}},    &   \text{$n$ is even.} 
\end{matrix}\right.
\end{align}
Further, the moments can be obtained as in \eqref{eq:moments_k} and the aggregate interference power can be expressed as
\begin{align}
\mathbb{E}\left\{ \hat{{i}}_{agg}  \hat{{i}}_{agg}^* \right\} = \frac{{ 2 \pi \lambda P r_0^{2-{\eta}} }}{ {\eta}-2}.     
\label{eq:power_GCB}
\end{align}

Comparing \eqref{eq:GCB} with \eqref{eq:pgf3}, it can be observed that both CFs have equivalent forms but with slightly different parameters, which confirms that the Gaussian signaling approximation maintains the same distribution for the aggregate interference.\footnote{Note that \eqref{eq:GCB} is related to \eqref{eq:pgf3} by substituting for $ \mathbb{E}\left\{\right|s|^{2n}\}=n!$, which is the case when $s\sim\mathcal{C}(0,1)$.} Also, comparing \eqref{eq:power_GCB} with \eqref{eq:power}, it can be observed that both $\hat{{i}}_{agg}$ and ${{ i}}_{agg}$ have equivalent powers. Hence, all the characteristics described for  ${{i}}_{agg}$ in Section~\ref{sec:PPP} hold for $\hat{{i}}_{agg}$. Fig.~\ref{i_agg_Approx} compares the PDF of $\hat{{ i}}_{agg}$ with the PDF of ${{ i}_{agg}}$. The figure shows that the PDF $\hat{{i}}_{agg}$ matches that of  ${{i}_{agg}}$ with high accuracy. Comparing \eqref{eq:moments1} with \eqref{eq:moments2}, it can be observed that difference between ${ i}_{agg}$ and $\hat{{i}}_{agg}$ exists only in even cumulants with orders higher than two, as highlighted in Table~\ref{tab:moments_comp}. {Our numerical results in Section~\ref{num_res_ref} (e.g., see Fig. \ref{fig:ASER_DL_r0}) show that such differences have minor effect on the SINR-dependent performance metrics such as the ASEP.} 

\begin{table}[!ht]
\caption{\: Per-dimension Cumulant Comparison for $\eta=4$.}
\center
\begin{tabular}{|l| l| l| l|}
\hline
\rowcolor[HTML]{C0C0C0}
\textbf{Cumulants} & \textbf{4-QAM}  &  \textbf{16-QAM}  & \textbf{Gaussian }  \\ \hline
& & & \\
$\kappa_2$             &  $\frac{\pi \lambda P}{2 r_0^2}$     &  $\frac{\pi \lambda P}{2 r_0^2}$      &   $\frac{\pi \lambda P}{2 r_0^2}$         \\        
& & & \\ \hline
& & & \\
$\kappa_4$             &  $ \frac{0.25 \pi \lambda P^2}{  r_0^6}$  &    $  \frac{0.33 \pi \lambda P^2 }{  r_0^6}$               &   $\frac{0.5 \pi \lambda P^2}{  r_0^6}$          \\       
& & & \\ \hline
& & & \\
$\kappa_6$                &  $ \frac{ 0.375 \pi \lambda P^3}{  r_0^{10}}$   & $ \frac{0.735 \pi \lambda P^3}{ r_0^{10}}$         &  $\frac{ 2.25 \pi \lambda P^3}{ r_0^{10}}$       \\          
& & & \\ \hline
& & & \\
Kurtosis  $\left(\frac{\kappa_4}{\kappa_2^2}\right)$            &  $\frac{1}{\pi \lambda r_0^2}$      & $\frac{1.32}{ \pi \lambda r_0^2}$      &   $\frac{2}{\pi \lambda r_0^2}$       \\          
& & & \\ \hline
\end{tabular}
\vspace{3mm}
\label{tab:moments_comp}
\end{table}

\begin{figure}[t]
	\begin{center}
	  \label{A_r150}\includegraphics[width=2.50in]{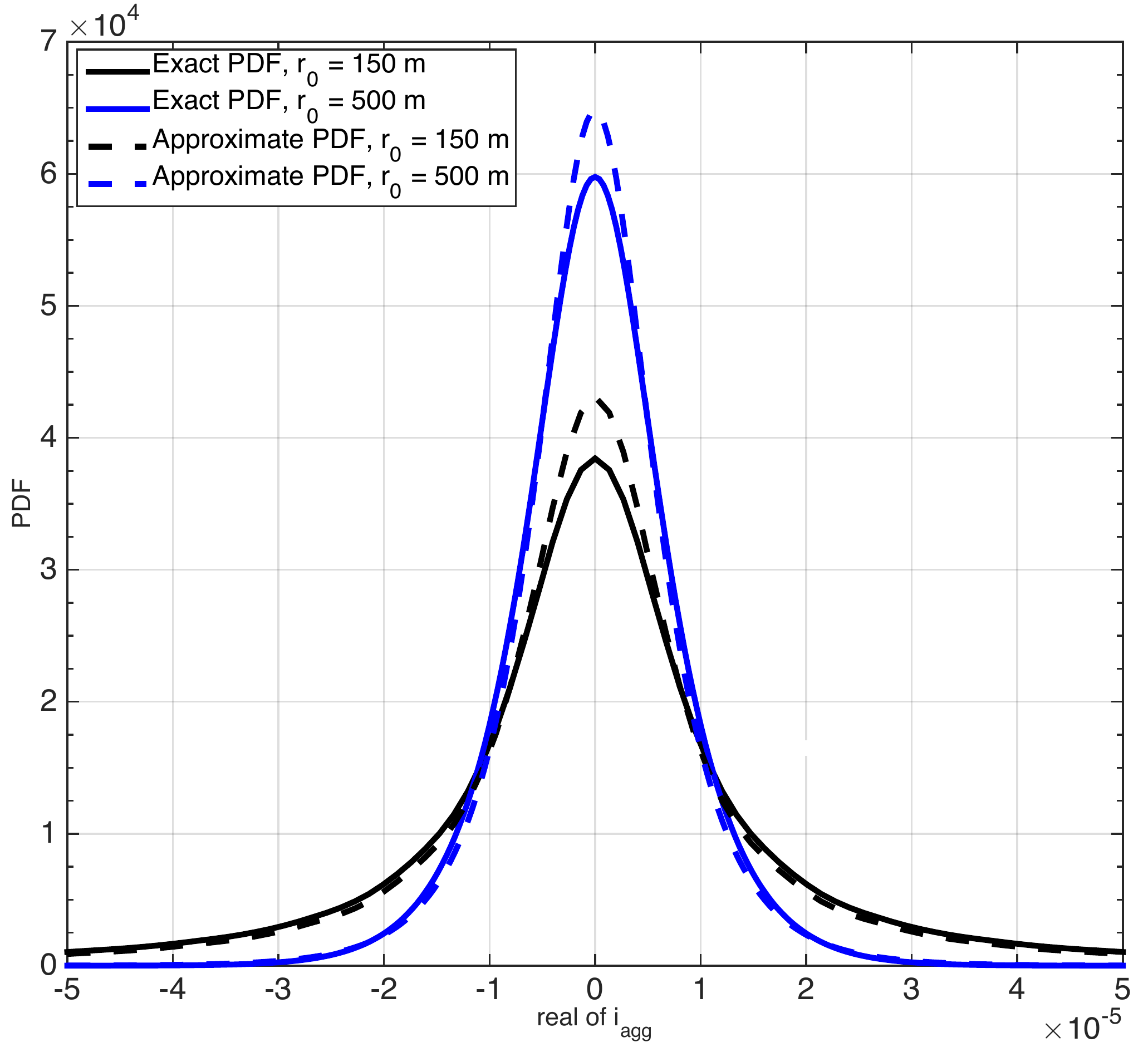}
	\end{center}
	\caption{ The PDF of  ${i}_{agg}$ obtained by numerically inverting \eqref{eq:pgf3} and \eqref{eq:GCB} at $\lambda=1$ BS/km$^2$, $P=10$ W and $\eta=4$ for $r_0=150$ m and $r_0=500$ m .}
	\label{i_agg_Approx}
\end{figure}

\subsection{Approximate Error Probability Analysis} \label{num_res_ref}

\normalsize
The Gaussian signaling assumption highly simplifies the analysis steps and reduces the computational complexity for the error probability expression. The main idea is to circumvent the complexity of the EiD approach by abstracting unnecessary system details (i.e., the interferers' transmitted symbols)~\cite{Laila_letter}. To visualize the conditional Gaussian representation of the aggregate interference, we rewrite the baseband signal at the test receiver \eqref{basic_baseband} with the Gaussian signaling as

\small
{\begin{align}
y_0 &\approx \sqrt{P} s_{0} h_{0} r_0^{-\frac{\eta}{2}} + \underset{\hat{{i}}_{agg}}{\underbrace{ \sum_{r_k \in \tilde{\Psi} \setminus r_0} \!\!\sqrt{P} \tilde{s}_{k} h_{k} r_k^{-\frac{\eta}{2}}}} +n, 
\label{basic_baseband_gauss} 
\end{align}}
\normalsize

\noindent{where} $s_0$ is the useful symbol that is randomly drawn form the constellation $\bf{S}$, and $\tilde{s}_k$ is an interfering symbol randomly drawn from a Gaussian constellation. Due to the Gaussian signaling assumption, conditioning on the network geometry (i.e., $r_k \in \tilde{\Psi}$, $\forall k$), channel gains (i.e., $h_0$ and $h_k$, $\forall k$), the received signal $y_0$ is conditional Gaussian. Particularly, the  conditional aggregate interference $\hat{{i}}_{agg} \sim \mathcal{CN}(0, \mathcal{I}_{agg})$ has a complex Gaussian distribution with total variance of ($\mathcal{I}_{agg} = {\sum}_{r_k \in \tilde{\Psi}\setminus r_0} P |h_k|^2 r_k^{-\eta}$). Hence, approximating the interfering symbols with Gaussian signals directly achieves the conditional Gaussian representation of the aggregate interference and renders the AWGN based ASEP expressions legitimate to be used. The SINR in \eqref{eq:cond_SINR}, with the Gaussian signaling approximation, can be expressed as

\small
{\begin{align}
\!\! \bar{\Upsilon}(r_0 \vert h_0, \mathcal{I}_{agg}) & =  \frac{\EXs{s_0}{\mathbb{E} \left\{ y_0 \right\} \mathbb{E} \left\{ y_0^*\right\} } }{ \mathbb{E} \left\{ y_0 y_0^*\right\}  - \mathbb{E} \left\{ y_0\right\} \mathbb{E} \left\{ y_0^*\right\}}  \notag \\
&= \frac{P |h_0|^2 \mathbb{E}\{|s_0|^2\} r_0^{-\eta}}{ \underset{r_k \in \tilde{\Psi}\setminus r_0}{\sum} P |h_k|^2 r_k^{-\eta} + N_0} \notag \\
&= \frac{P |h_0|^2 r_0^{-\eta}}{ \mathcal{I}_{agg} + N_0},
\label{eq:sinr_GCB}
\end{align}}
\normalsize 

\noindent{Similar} to the EiD case in \eqref{ASEP_cond11}, the unconditional ASEP in the Gaussian signaling approximation is expressed as\footnote{The ASEP is unconditional because  the expectation in \eqref{eq:ASEP_GCB1} w.r.t. $\mathcal{I}_{agg}$ and $h_0$, however, the expressions is still for a given $r_0$.}

\scriptsize
\begin{align} \label{eq:ASEP_GCB1}
&\bar{\mathcal{S}}(r_0) = \displaystyle {w_{1} }\mathbb{E}\left\{ \text{erfc}\left( \sqrt{\beta_1  \bar{\Upsilon}(r_0 \vert h_0, \mathcal{I}_{agg})}\right) \right\} \notag \\
& \quad \quad \quad \quad \quad \quad \quad \quad \quad \quad \quad  \quad \quad \quad \quad \quad \quad + \displaystyle {w_{2} } \mathbb{E}\left\{\text{erfc}^2\left( \sqrt{\beta_2  \bar{\Upsilon}(r_0 \vert h_0, \mathcal{I}_{agg})}\right)\right\} \notag \\
&\overset{(a)}{=} \sum_{c=1}^2 w_c \left(1 - \frac{2}{\sqrt{\pi}} \int\limits_0^{\infty} \frac{\text{erfc}(\sqrt{z} \mathbbm{1}_{\{c=2\}})}{\sqrt{z}} e^{-z\left( 1+ \frac{N_0 r_0^\eta}{P \beta_c}\right)} \mathcal{L}_{\mathcal{I}_{agg} }\left(  \frac{r_0^\eta}{P \beta_c}\right) \mathrm{d} z \right) 
\end{align}
\normalsize

\noindent{where}, similar to \eqref{ASEP_cond11}, $(a)$ follows from the lemma proposed in \cite{hamdi_useful_tech}, which is given in \textbf{Appendix~\ref{app:useful}}. The ASEP in \eqref{eq:ASEP_GCB1} requires the LT of $\mathcal{I}_{agg}$, which is characterized in the following lemma

\begin{lemma} \label{LT_I_power}
The  LT of the aggregate inter-cell interference in one-tier cellular network modeled via a PPP with constant transmit power $P$, intensity $\lambda$, Rayleigh fading, and nearest BS association is given by
\small
\begin{align}
\mathcal{L}_{\mathcal{I}_{agg}}(z)   &{=}  \exp \left\{-\frac{2 \pi \lambda z P r_0^{2-\eta}}{\eta-2} {}_{2}F_1\left(1, 1-\frac{2}{\eta}; 2- \frac{2}{\eta}; - \frac{zP}{r_0^{\eta}}\right)  \right\} 
\label{eq:LT_I_power}
\end{align}
\normalsize
\end{lemma}
\begin{proof}
See \textbf{Appendix~\ref{LT_I_power_proof}}.
\end{proof}
\begin{remark}
 A special case of Lemma~\ref{LT_I_power} is for $\eta=4$, which is a common practical value for path-loss exponent in outdoor urban environments. In this case, unlike the EiD approach ASEP in \eqref{eq:ASEP_exact_r0}, the LT of $\mathcal{I}_{agg}$ expression reduces from the Gauss hypergeometric  function $_{2}F_1(.,.;.;.)$  to the elementary inverse tangent function as
\begin{equation}
\mathcal{L}_{\mathcal{I}_{agg}}(z) \overset{(\eta=4)}{=}  \exp \left\{-\pi \lambda \sqrt{z P} \arctan\left( \frac{\sqrt{zP}}{r_0^{2}}\right)  \right\}
\label{eq:LT_I_power_sc}
\end{equation}
\end{remark}

Accordingly, the ASEP for the downlink communication links is provided by Theorem \ref{theo:ASEP_GCB}, given at the top of the next page, which is obtained by plugging \eqref{eq:LT_I_power} and \eqref{eq:LT_I_power_sc} into \eqref{eq:ASEP_GCB1}.

\begin{figure*}
\begin{theorem}
\label{theo:ASEP_GCB}
Consider cellular network modeled via a PPP with intensity $\lambda$ in Rayleigh fading environment with universal frequency reuse and no intra-cell interference. Then, the downlink ASEP, with $M$-QAM modulated useful signal and Gaussian interfering signals, for a user located at the distance $r_0$ away from his serving BS, is expressed as
\small
\begin{align}
\bar{\mathcal{S}}(r_0) &= \sum_{c=1}^2 w_c \left( 1 - \frac{c}{\sqrt{\pi}}  \int\limits_0^{\infty} \frac{ \text{erfc}(\sqrt{z} \mathbbm{1}_{c=2}) }{\sqrt{z}} \exp\left\{-z\left( 1+ \frac{N_0 r_0^\eta}{P \beta_c}\right) -\frac{2 \pi \lambda z r_0^{2}}{\beta_c (\eta-2)} {}_{2}F_1\left(1, 1-\frac{2}{\eta}; 2- \frac{2}{\eta}; - \frac{z}{\beta_c} \right)  \right\}  \mathrm{d} z \right)
\label{eq:ASEP_GCB0}  \\
&\overset{\eta=4}{=}\sum_{c=1}^2 w_c \left( 1 - \frac{c}{\sqrt{\pi}}  \int\limits_0^{\infty} \frac{ \text{erfc}(\sqrt{z} \mathbbm{1}_{c=2}) }{\sqrt{z}} \exp\left\{-z\left( 1+ \frac{N_0 r_0^4}{P \beta_c}\right) -\pi \lambda  r_0^2 \sqrt{\frac{z}{\beta_c}} \arctan\left( \sqrt{\frac{z}{\beta_c}}\right) \right\}  \mathrm{d} z \right).
\label{eq:ASEP_GCB2}
\end{align} 
\normalsize
\hrulefill
\end{theorem}

\end{figure*}

Fig. \ref{fig:ASER_DL_r0} compares the ASEP obtained via the EiD approach \eqref{eq:ASEP_exact_r0}, the Gaussian signaling approximation \eqref{eq:ASEP_GCB0}, the Gaussian aggregate interference with variance in \eqref{eq:power}, and Monte Carlo simulation for different BSs intensities. {The gap between the Gaussian aggregate interference approximation and the exact analysis (i.e., EiD) confirms that the central limit theorem for the aggregate interference does not apply. Hence, assuming Gaussian aggregate interference results in a loose estimate for the ASEP. On the other hand, the close match between the Gaussian signalling apperoximation and the exact analysis validates the Gaussian signaling approximation and shows that it accurately captures the ASEP.} The figure also shows that the gap between the Gaussian signaling approximation and the exact analysis  diminishes for higher constellations as discussed before in Section~\ref{sec:GCB_valid} (cf. Table~\ref{tab:moments_comp}). As discussed in Section~\ref{sec:PPP}, the interferers' intensity and interference boundary (see Fig.~\ref{fig:int_bound}) are the two parameters that characterize the interfering PPP. In this regard, the figure manifests the prominent effect of the interferers' intensity and interference boundary on the network performance.

\begin{figure}[t]
	\begin{center}
	 \includegraphics[width=3in]{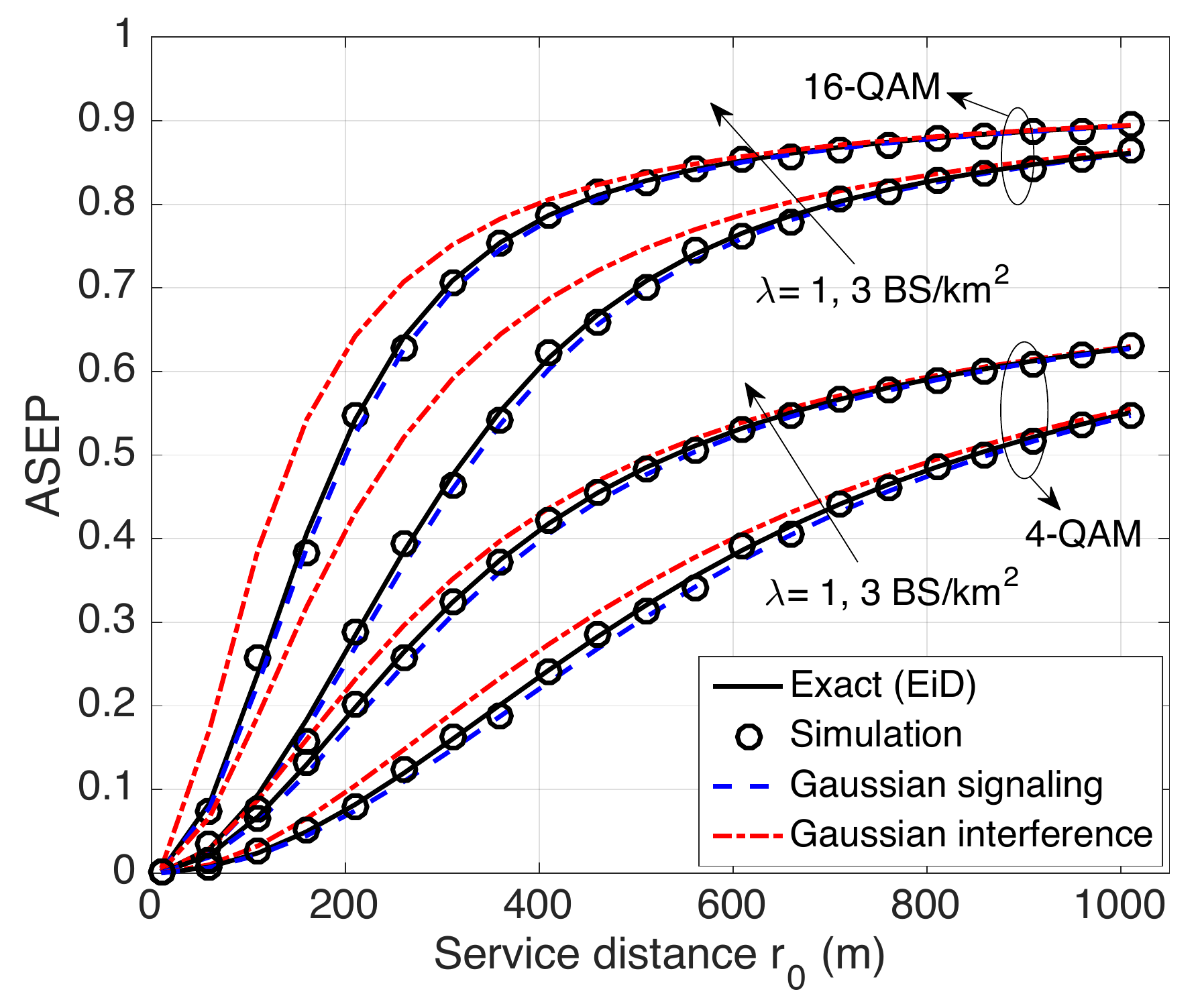}

	\end{center}
	\caption{ ASEP vs the service distance $r_0$ for 4-QAM and 16-QAM constellations. }
	\label{fig:ASER_DL_r0}
\end{figure}

\subsection{Section Summary}

This section motivates the Gaussian signaling approximation for interfering symbols to facilitate the ASEP analysis in cellular networks. We first validate the Gaussian signaling approximation by showing that it preserves the distribution of the aggregate interference signal, provides matching odd and second cumulants, as well as matching interference power for any constellation size. Difference between the exact interference and the interference based on Gaussian signaling only exists for even cumulants with orders higher than two. 

The effect of the Gaussian signaling approximation on the ASEP expression can be observed by comparing \eqref{eq:ASEP_exact_r0} with \eqref{eq:ASEP_GCB0}. One can see that the Gaussian signaling approximation reduces the sum of $M$ hypergeometric functions in the exponent for the constellation size $M$, to a single hypergeometric function exponent. This highly reduces the computational complexity to evaluate the ASEP without sacrificing the ASEP accuracy. Furthermore, for the special case of $\eta=4$ the expression for the ASEP reduces to {a computationaly simple inverse tangent function, which is not the case for the exact EiD. }

{The Gaussian signaling  approximation also facilitates the derivation steps to obtain the ASEP. Particularly, the analysis requires the LT of the aggregate interference power ($\mathcal{I}_{agg}$), which is easier to derive and simpler to evaluate than the CF of the baseband aggregate interference required by the EiD approach.}  Furthermore, the LT of $\mathcal{I}_{agg}$ can be used to compute several other performance metrics. As will be shown in the next section, the Gaussian signaling unifies the computation of the ASEP, outage probability, and ergodic capacity. 


\section{Outage Probability and Ergodic Rate} \label{sec:out}

Error probability expressions provide a tangible characterization of network performance and capture the effect of several system factors. However, as shown in the Section~\ref{sec:exact} and Section~\ref{sec:GCB}, the ASEP expressions are quite involved, even with the Gaussian signaling approximation. Such complicated expression do not directly reveal system design insights. Therefore, to simplify the analysis, several researchers resort to more conceptual analysis relying on quantities such as outage probability and ergodic rate. Such abstracted analysis leads to simple expressions that characterize the network performance, highlight the tradeoffs, and facilitate network design. 

\subsection{Definition of Outage Probability and Ergodic Rate}

For AWGN channels, the maximum rate per unit bandwidth (BW) that can be reliably transmitted, also known as the spectral efficiency, is defined by Shannon's capacity expression given by:
\begin{equation}
C=\log \left(1+ {\rm SNR} \right)
\label{SNR_cap}
\end{equation} 

\noindent where the SNR  in \eqref{SNR_cap} is the instantaneous signal-to-noise ratio. Shannon's capacity formula assumes that the additive noise is Gaussian and that coded trasmission is employed with codewords drawn from a Gaussian codebook. If this expression is extended to include interference, then the interference signal should also be Gaussian. This is the case when the interfering BSs also employ Gaussian codebooks, which is equivalent to the use of Gaussian signaling in Section \ref{sec:GCB}. Similar to \eqref{basic_baseband_gauss}, the baseband aggregate interference signal is Gaussian conditioned on the PPP, which validates lumping the aggregate interference with the noise term. That is, treating interference as noise, the instantaneous SINR ($\Upsilon$) in \eqref{eq:sinr_GCB} is analogous to the SNR in \eqref{SNR_cap} for Gaussian interfering symbols $\tilde{s}_k$ when conditioning on the interfering BSs locations $r_k \in \tilde{\Psi} \setminus r_0$. Therefore, \eqref{SNR_cap} is legitimate to asses the link capacity in the depicted large-scale cellular network. However, an additional averaging step over $\Upsilon$ is required, which leads to the following ergodic rate per unit BW definition

\small
{\begin{align} \label{eq:capacity1}
\mathcal{C} &= \mathbb{E}\left\{ \log(1+ \Upsilon ) \right\} \notag \\
&\overset{a}{=} \int_0^\infty \mathbb{P}\left\{ \log(1+ \Upsilon ) > t\right\} \mathrm{d}t\notag \\
&= \int_0^\infty \mathbb{P}\left\{ \Upsilon > e^t -1 \right\} \mathrm{d}t\notag \\
&= \int_0^\infty \left(1-F_{\Upsilon}\left( e^t -1 \right) \right) \mathrm{d}t\notag \\
&\overset{b}{=} \int_0^\infty \frac{\left(1-F_{\Upsilon}\left(y \right)\right)}{y+1} \mathrm{d}y 
\end{align}}
\normalsize

\noindent where $(a)$ follows because $\log(1+ \Upsilon )$ is a positive random variable, $(b)$ is obtained by change of variables, and $F_\Upsilon(\cdot)$ is the CDF of the SINR ($\Upsilon$). Shannon's capacity expression in \eqref{SNR_cap} can also be used to define the outage probability. Let $R$ be the transmission rate, then the outage probability is defined as the probability that the transmission rate is greater than the channel capacity, given by

\begin{align}
\mathcal{O}(R) &=  \mathbb{P} \left\{ \log \left(1+ \Upsilon \right) < R\right\}  \notag \\
& = \mathbb{P} \left\{ \Upsilon < e^R -1 \right\} 
\label{SNR_out1}
\end{align}

\noindent where $\Upsilon$ denotes the instantaneous SINR (i.e., as in \eqref{eq:sinr_GCB} without conditioning on either $h_0$ or $\mathcal{I}_{agg}$). Hence, the rate outage probability depends on interference and/or fading.

Bit error rate $({\rm BER})$ is another technique to define outage probability. In this case, the outage probability is defined as the probability that the ${\rm BER}$ exceeds a certain threshold $\epsilon$. Exploiting the Gaussian signaling approximation, the BER based outage probability is given by

\begin{align}
\mathcal{O}(\epsilon) &= \mathbb{P} \left\{ {\rm BER} > \epsilon \ \right\} \notag \\
& = \mathbb{P} \left\{ w_1 \text{erfc}\left(\beta_1 \Upsilon \right) > \epsilon \right\}  \notag \\
& = \mathbb{P} \left\{ \Upsilon < \frac{1}{\beta_1} \text{erfc}^{-1}\left(\frac{\epsilon}{w_1}\right) \right\}
\label{SNR_out2}
\end{align}

\noindent where \eqref{SNR_out2} ignores the $\text{erfc}^2(\cdot)$ term of \eqref{SNR_ASEP}.

Most of the SG literature does not discriminate between the two forms of outage probabilities in \eqref{SNR_out1} and \eqref{SNR_out2}. Instead, the outage probability is treated in an abstract manner with a unified abstracted threshold value ($T$), as follows:

\small
{\begin{align} \label{eq:out1}
\mathcal{O}(T) &= \mathbb{P}\left\{ \Upsilon < T \right\} \notag \\
&=F_{\Upsilon}(T)
\end{align}}
\normalsize

\noindent{Equations} \eqref{eq:capacity1} and \eqref{eq:out1} show that the SINR CDF is sufficient to characterize both the outage probability and ergodic rate. The SINR distribution is obtained in the next section. 

\subsection{SINR Distribution}

The SINR CDF is given by

\small
{\begin{align} \label{out_cdf}
F_{\Upsilon}(T) &= \mathbb{P}\left\{ {\Upsilon} < T \right\} \notag \\
 &= \mathbb{P}\left\{  \frac{P |h_0|^2 r_0^{-\eta}}{ \mathcal{I}_{agg} + N_0} < T \right\} \notag \\
 &= \mathbb{P}\left\{  |h_0|^2 < \frac{T ( \mathcal{I}_{agg} + N_0)}{P r_0^{-\eta}} \right\} \notag \\
  &\overset{(a)}{=} \EXs{\mathcal{I}_{agg}}{ F_{ |h_0|^2}\left( \frac{T ( \mathcal{I}_{agg} + N_0)}{P r_0^{-\eta}} \right) }\notag \\
  &\overset{(b)}{=} 1- e^{\frac{-T N_0 r_0^{\eta}}{P }}  \mathcal{L}_{\mathcal{I}_{agg}}\left(   \frac{T  r_0^{\eta}}{P }  \right)  
\end{align}}
\normalsize

\noindent where $(b)$ follows from the exponential distribution of $|h_0|^2$ and the definition of the LT. It is worth highlighting that $(a)$ in \eqref{out_cdf} cannot be always computed. This is because the PDF of the interference power $\mathcal{I}_{agg}$ is not available in closed-form, except for {very special cases which are not of  practical interest for cellular networks} \cite{martin_book, baccelli_vol1, baccelli_vol2, now_martin, now_jeff}.\footnote{{The interference distribution can be only found for special cases of PPP networks in which the interference boundaries (cf. Fig.~\ref{fig:int_bound}) go from $0$ to $\infty$~\cite{sousa1990}, which is not suitable to model cellular networks that enforce an inner interference boundary of $r_0$.}} However, the exponential distribution 
 of $|h_0|^2$ enables expressing the CDF of the SINR in terms of the LT of $\mathcal{I}_{agg}$. The LT of $\mathcal{I}_{agg}$ is given in Lemma~\ref{LT_I_power}, which is used to characterize the ergodic rate and  outage probability in the following theorem

\begin{theorem} \label{outage and rate}
Consider a cellular network modeled via a PPP with intensity $\lambda$ in Rayleigh fading environment with universal frequency reuse and no intra-cell interference. The downlink 
ergodic rate for a user located at the distance $r_0$ away from his serving BS can expressed as

\scriptsize
 {\begin{align}
\mathcal{C}(r_0) &= \int_0^{\infty} \frac{\exp \left\{-\frac{t N_0 r_0^{\eta}}{P}  - \frac{2 \pi \lambda t r_0^{2}}{\eta-2} {}_{2}F_1\left(1, 1-\frac{2}{\eta}; 2- \frac{2}{\eta}; - t \right)  \right\}}{t+1} \mathrm{d}t \notag \\
 &\overset{\eta=4}{=} \int_0^{\infty} \frac{\exp \left\{-\frac{t N_0 r_0^{4}}{P}  -  \pi \lambda \sqrt{t} r_0^{2} \arctan\left(\sqrt{t} \right) \right\}}{t+1} \mathrm{d}t,
\label{eq:abstract_cap}
\end{align}}
\normalsize

\noindent{and} outage probability for a user located at the distance $r_0$ away from his serving BS can expressed as

\scriptsize
 {\begin{align}
\mathcal{O}(r_0,T) &= 1-\exp \left\{-\frac{T N_0 r_0^{\eta}}{P}  - \frac{2 \pi \lambda T r_0^{2}}{\eta-2} {}_{2}F_1\left(1, 1-\frac{2}{\eta}; 2- \frac{2}{\eta}; - T \right)  \right\} \notag \\
 &\overset{\eta=4}{=} 1-\exp \left\{-\frac{T N_0 r_0^{4}}{P}  -  \pi \lambda \sqrt{T} r_0^{2} \arctan\left(\sqrt{T} \right) \right\}.
\label{eq:abstract_out}
\end{align}}
\normalsize
 
\end{theorem}
\begin{proof}
The theorem is obtained by plugging  the LT expressions \eqref{eq:LT_I_power} and \eqref{eq:LT_I_power_sc} into \eqref{out_cdf} to get the SINR CDF, which is then used to compute the ergodic rate and the outage probability as in  \eqref{eq:capacity1} and \eqref{eq:out1}, respectively. 
\end{proof}

\begin{figure*}

	\begin{center}
	  \subfigure[Outage probability for different BSs intensities at $T=1$ and $\eta= 4$.]{\label{fig:out1}\includegraphics[width=2.3in]{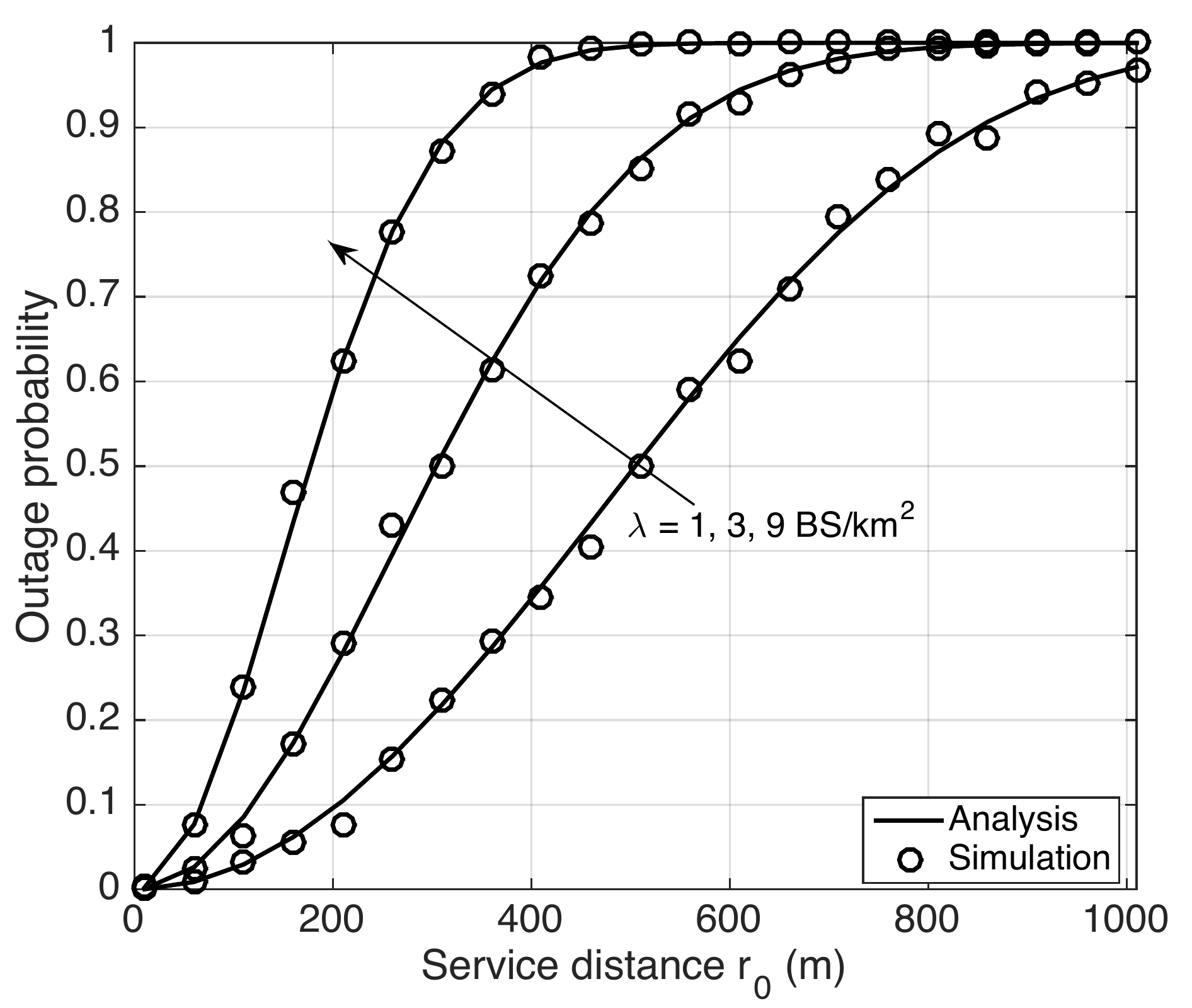}}
	  \subfigure[Outage probability for different SINR thresholds at $\lambda = 3 $ BS/km$^2$ and  $\eta= 4$.]{\label{fig:out2}\includegraphics[width=2.3in]{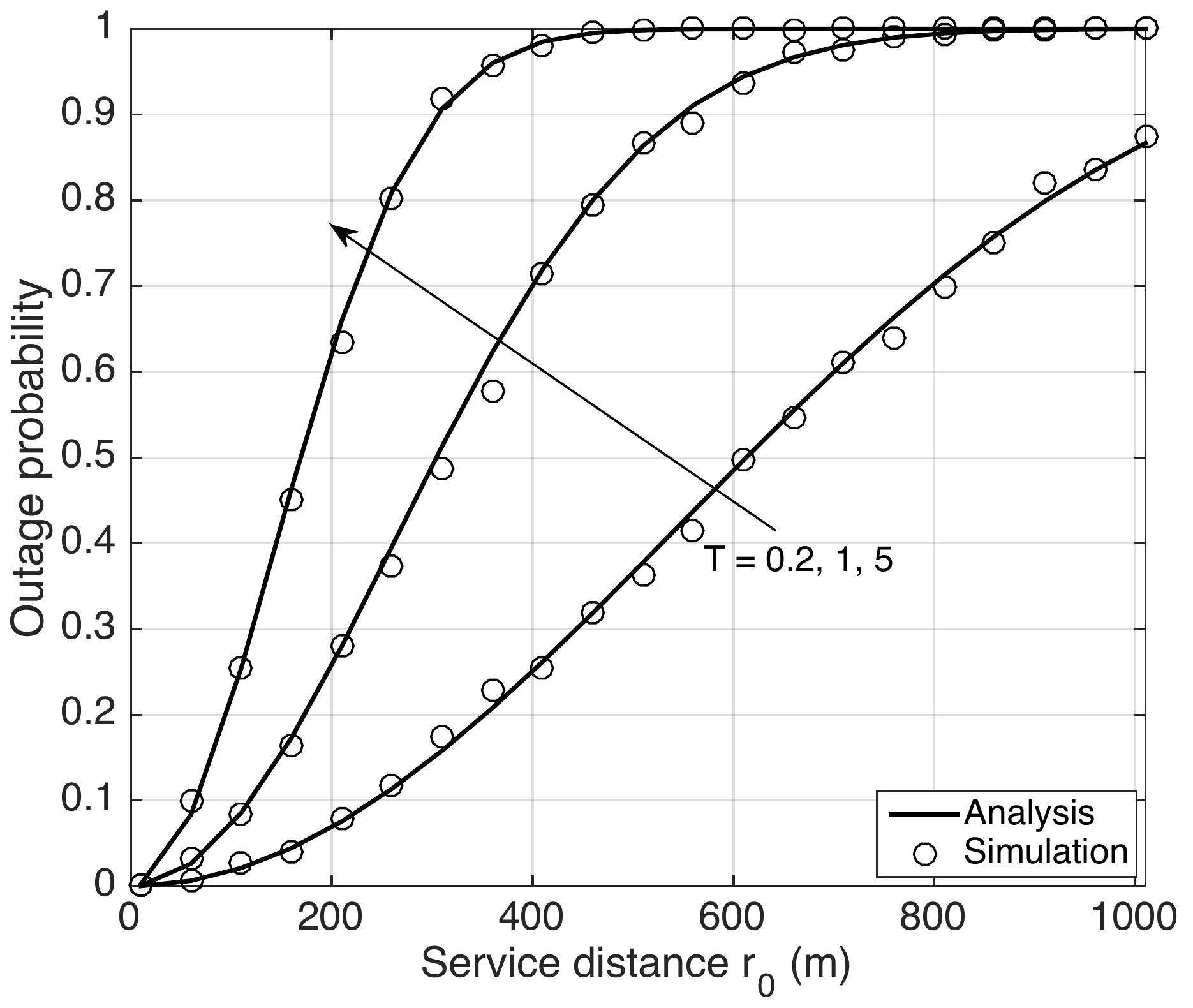}}
		\subfigure[Ergodic rate per unit BW for different BSs intensities at $\eta=4$.]{\label{fig:out3}\includegraphics[width=2.3in]{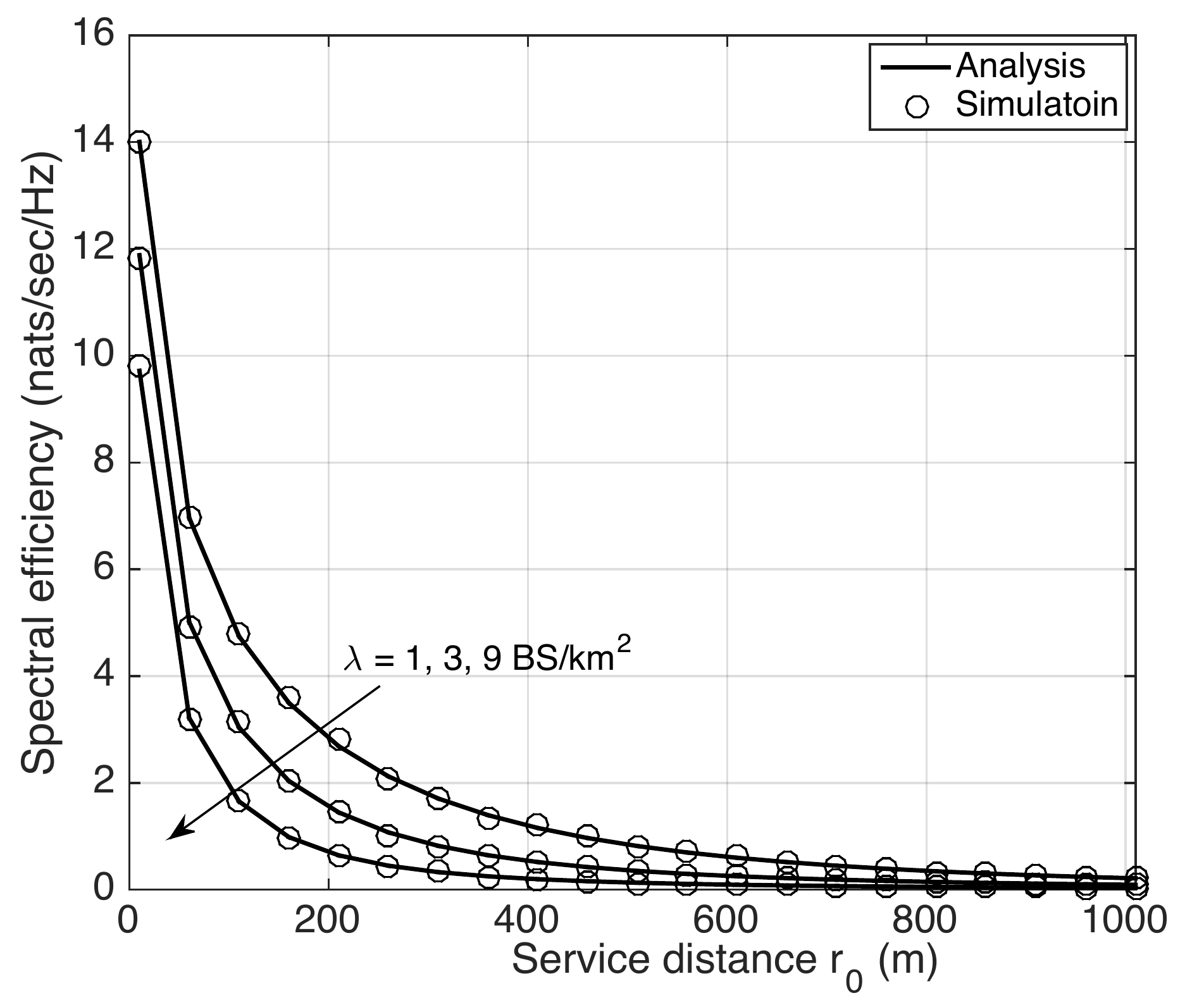}}
	\end{center}
	\caption{ Outage probability and ergodic rate vs the service distance $r_0$.}
	\label{fig:Out_Spec}
\end{figure*}

Fig.~\ref{fig:Out_Spec} validates \eqref{eq:abstract_cap} and \eqref{eq:abstract_out} against Monte Carlo simulation. Similar to Fig.~\ref{fig:ASER_DL_r0}, the results in Fig.~\ref{fig:Out_Spec} show the effect of interferers' intensity and interference boundary on the network performance. Hence, the outage probability and ergodic rate can be used as an alternative and simpler way to characterize the network behavior.\footnote{(i.e., The outage and ergodic rate expressions \eqref{eq:abstract_out} and \eqref{eq:abstract_cap} are simpler than the ASEP expressions \eqref{eq:ASEP_exact_r0} and \eqref{eq:ASEP_GCB0})} However, such simplicity comes at the expense of abstractions that may hide the true network behavior. {As shown in Fig.~\ref{fig:out2} the network performance is a function of the abstracted SINR threshold value, which gives a constellation oblivious performance measure. On the other hand, Fig.~\ref{fig:ASER_DL_r0} clearly shows the true ASEP for each modulation scheme.}

\subsection{Section Summary}

The outage probability and ergodic rate can be defined in terms of the SINR CDF. This may lead to closed-form simple expressions which help to characterize the network performance. 
It is worth mentioning that the Gaussian signaling approximation provides a unified approach to characterize SINR related performance metrics. Thas is, the outage probability, ergodic capacity, and also ASEP under Gaussian signaling approximation require obtaining the LT of the aggregate interference power as in \eqref{eq:LT_I_power}. Then, these quantities are computed by plugging the LT of $\mathcal{I}_{agg}$ into \eqref{eq:abstract_out}, \eqref{eq:abstract_cap}, and \eqref{eq:ASEP_GCB0}, respectively.



\section{Advanced Network Models} \label{sec:advanced}

In this section, we focus on analysis based on Gaussian signaling approximation. Hence, we only show $\mathcal{L}_{\mathcal{I}_{agg}}(\cdot)$ and we neither calculate  $\{\sigma_q^2\}_{q=1}^{\infty}$ nor $\mathcal{L}_{\zeta}(\cdot)$. As shown in the previous sections, the ASEP, outage probability, and ergodic rate expressions are all functions of the LT of the aggregate interference $\mathcal{I}_{agg}$. Therefore, throughout this section, we will show how the LT of the aggregate interference would change for each network model. For the sake of unified and simple presentation for the ASEP, outage probability, and ergodic rate, we focus on the case of $\eta=4$ and negligible noise variance. Note that the methodology of analysis that we apply to this special case can be directly extended for more general cases (i.e., general $\eta$ and with noise), but at the expense of slightly more involved expressions. We will also point out the references that conduct the general analysis in each of our case studies. In case of $\eta=4$ and negligible noise,  \eqref{eq:ASEP_GCB1} and \eqref{out_cdf} reduce to

\footnotesize
\begin{align} \label{eq:ASEP_unified}
&\bar{\mathcal{S}}(r_0) {=} \sum_{c=1}^2 w_c \left(1 - \frac{c}{\sqrt{\pi}} \int\limits_0^{\infty} \frac{e^{-z}  \text{erfc}(z \mathbbm{1}_{\{c=2\}})}{\sqrt{z}} \mathcal{L}_{\mathcal{I}_{agg} }\left(  \frac{r_0^4}{P \beta_c}\right) \mathrm{d} z \right) 
\end{align}
\normalsize
and
\small
{\begin{align} \label{eq:out_unified}
F_{\Upsilon}(T) {=} 1-  \mathcal{L}_{\mathcal{I}_{agg}}\left(   \frac{T  r_0^{4}}{P }  \right). 
\end{align}}
\normalsize

\noindent Hence, we focus on the LT of the aggregate interference evaluated at $\frac{a r_0^4}{P}$, where $a = \beta_c^{-1}$ for ASEP evaluation, and $a = T$ for outage probability and ergodic rate evaluation. 

As discussed in Section~\ref{PPP}, as far as the PPP is considered, the interference exclusion region (denoted hereafter as $r_\mathcal{I}$) and the intensity $\lambda$ are the two main parameters that discriminate LT of the interference in different network models. Note that the baseline network model used in the previous sections assumed a single tier cellular network with no interference coordination. Hence, the interference exclusion distance is equivalent to the service distance (i.e., $r_\mathcal{I} = r_0$) and the interferers' transmit powers are equivalent. However, this might not always be the case. In the next sections, we discriminate between the interference exclusion distance $r_\mathcal{I}$ and the service distance $r_0$. We will also discriminate between the interferers' transmit power $P_\mathcal{I}$ and the serving BS transmit power $P_0$. Then, the LT of the interference in \eqref{eq:LT_I_power_sc} can be generalized to

\small
\begin{equation} \label{LT_start1}
 \mathcal{L}_{\mathcal{I}_{agg}}(z,\lambda,r_\mathcal{I}) = \exp\left\{-\pi \lambda \sqrt{zP_\mathcal{I}} \arctan\left({\frac{\sqrt{zP_{\mathcal{I}}}}{r_\mathcal{I}^2}}\right) \right\}.
\end{equation}
\normalsize
Then, substituting $z= \frac{a r_0^4}{P_0}$ into \eqref{LT_start1}, we have

\small
\begin{align} \label{LT_start}
 \mathcal{L}_{\mathcal{I}_{agg}}(a,\lambda, r_0, r_\mathcal{I}) &= \notag \\ 
& \!\!\!\!\!\!\!\!\!\!\!\!\!\!\!\!\!\!\!\!\!\! \exp \left\{-\pi \lambda \sqrt{\frac{a P_\mathcal{I}}{P_0}} r_0^2 \arctan\left( \left(\frac{r_0}{r_\mathcal{I}}\right)^2\sqrt{\frac{aP_{\mathcal{I}}}{P_0}}\right) \right\}.
\end{align}
\normalsize

\noindent Equation \eqref{LT_start} serves as a basis for the analysis in the sequel.

\subsection{Random Link Distance $r_0$} \label{Relaxing_r_0}

A random link distance is an intrinsic property of the baseline  cellular network model. Hence, averaging over the link distance is required to obtain the spatially average performance. Note that the random service distance $r_0$ does not change any of the previous analysis and only adds an additional averaging step over $r_0$. This is because both the aggregate interference and the useful signal power in \eqref{eq:sinr_GCB} depend on the service distance $r_0$. Hence, we first obtain the conditional (i.e., on $r_0$) LT of the aggregate interference as in \eqref{eq:LT_I_power} and then conduct the averaging step over $r_0$. Note that the service distance $r_0$ in \eqref{eq:ASEP_unified} and \eqref{eq:out_unified} appears within the LT of $\mathcal{I}_{agg}$ only, and hence, the averaging step over $r_0$ only affects the LT expression.\footnote{This is not the case if noise is taken into consideration. In the prominent noise case, the averaging step should include the noise term as well as the LT of the interference.} That is, the ASEP and the SINR CDF are given in terms of the spatially averaged LT (i.e., after averaging over $r_0$). It is worth mentioning that in the subsequent case studies, random service distance is always considered and the spatially averaged LT is calculated. 

For cellular networks modeled via PPPs and employing nearest BS association, the distribution of $r_0$ is given in \eqref{pdf_r_0}. By averaging over $r_0$, the LT is given by

\small
 {\begin{align}
\mathcal{L}_{\mathcal{I}_{agg}}\left(a, \lambda \right) & = \int\limits_0^\infty 2 \pi \lambda r  e^{-\pi \lambda r^2 } \mathcal{L}_{\mathcal{I}_{agg}}\left(a,\lambda,r,r\right) \mathrm{d}r \notag \\
& = \int\limits_0^\infty 2 \pi \lambda r\exp\left\{-\pi \lambda r^2 \left(\sqrt{a}  \arctan\left({\sqrt{a}}\right)+1 \right) \right\}\mathrm{d}r  \notag \\
& = \frac{1}{\sqrt{a} \arctan\left(\sqrt{a} \right) + 1}
\label{eq:LT_random_ro}
\end{align}}
\normalsize

The ASEP and the SINR CDF are obtained by substituting \eqref{eq:LT_random_ro} into \eqref{eq:ASEP_unified} with $a=\beta^{-1}$ and into \eqref{eq:out_unified} with $a=T$, respectively. Fig.~\ref{fig:Out_downlink} validates \eqref{eq:LT_random_ro} via Monte Carlo simulation for the outage probability (i.e., $a=T$). The simplicity of \eqref{eq:LT_random_ro} reveals several insights into the performance of the cellular network. For instance, under a noise-limited operation, the ASEP depends only on the modulation scheme parameters $w_c$ and $\beta_c$. Hence, the outage probability is only a function of the threshold value $T$ and the ergodic rate is constant. This is different from the results shown in Fig.~\ref{fig:ASER_DL_r0} and Fig.~\ref{fig:Out_Spec} which show that the downlink performance depends on service distance $r_0$ and the BS intensity $\lambda$. This is because in Fig.~\ref{fig:ASER_DL_r0} and Fig.~\ref{fig:Out_Spec} we conditioned on $r_0$, and hence, the service distance $r_0$ does not adapt to the intensity $\lambda$. In reality, increasing the BSs intensity  would imply shorter service distance. Hence, the effect of increasing the number of interferers is balanced via a shorter service distance and results in constant SINR-dependent performance metrics \cite{tractable_app} \cite{k_tier}. A thorough discussion for this case study can be found in \cite{tractable_app}.

\subsection{Load-Aware Modeling} \label{sec:load}

The previous sections assume universal frequency reuse for a single channel and $\lambda_u \gg \lambda$, such that each BS always has a user to serve. However, in practice, multiple channels are available per BS and some channels may be left idle (i.e., some BSs might not be fully loaded). The results  in \cite{cog_h, on_cog,  on_ch, load_aware_harpreet, Harvest2_sakr} show that assuming fully-loaded network leads to a pessimistic performance evaluation. Hence, load-awareness is essential for practical performance assessment. In a load-aware model, the SINR-dependent performance analysis is conducted for each channel and the per-channel access probability in each BS is taken into account. Let $\mathbf{N}$ be the set of available channels, and without loss of generality, we assume that each BS randomly and uniformly selects a channel to assign for each user request.\footnote{If each BS assigns the channels based on the channel quality index (CQI), to exploit multi-user diversity, and all the channel gains are identically distributed, then, for a generic user at a generic time instant, each of the channels has the same probability to be the channel with the highest CQI.}  
 Following \cite{Harvest2_sakr}, the probability that a generic channel is used by a randomly selected BS is given by
 
{\begin{align}
p &= \mathbb{P}\left\{ n_j \in \mathbf{N} \text{ is used} \right\} \notag \\
&= \underset{k}{\sum} \mathbb{P}\left\{ \mathcal{U} =k \right\} \frac{\binom{N-1}{k-1}}{\binom{N}{k}} \notag \\
&= \underset{k}{\sum} \mathbb{P}\left\{ \mathcal{U} =k \right\} \frac{k}{N}
\end{align}}

\noindent where $N$ is the number of channels in $\bf{N}$, $\mathbb{P}\left\{ \mathcal{U} =k \right\}$ is the probability mass function (PMF) of the number of users served by each BS, which is given by \eqref{eq:load} when the UEs follow a PPP which is independent from the BS locations.  

From the SINR perspective, the analysis in the load-aware case is similar to Section~\ref{Relaxing_r_0}. However, the intensity of interfering BSs is thinned by the per channel access probability $p$. Hence, the intensity $\lambda$ in the LT expression in Lemma~\ref{LT_I_power} is replaced by the intensity of active BSs per channel $p \lambda$. On the other hand, the distribution of the service distance $r_0$ remains the same (i.e., with intensity $\lambda$) as each user has the opportunity to be associated with the complete set of BSs. However, it only receives interference from the subset of active BSs (i.e., the BSs using the same channel). Also, the interference exclusion region is equal to the service distance (i.e., $r_\mathcal{I} = r_0$). Hence, the LT of the aggregate interference is given by

\small
 {\begin{align}
 \mathcal{L}_{\mathcal{I}_{agg}}\left(a,\lambda\right) &{=} \int\limits_0^\infty 2 \pi \lambda r e^{-\pi \lambda r^2 }\mathcal{L}_{\mathcal{I}_{agg}}\left(a,p\lambda,r,r\right) \mathrm{d}r \notag \\
&{=} \int\limits_0^\infty 2 \pi \lambda r \exp \left\{ - \pi p \lambda r^{2} \left( \sqrt{a}  \arctan\left(\sqrt{a} \right) + \frac{1}{p} \right) \right\} \mathrm{d}r \notag \\
&{=} \frac{1}{p \left(\sqrt{a} \arctan\left(\sqrt{a} \right) + \frac{1}{p} \right)}
\label{eq:LT_random_load}
\end{align}}
\normalsize


\noindent Equation \eqref{eq:LT_random_load} shows that load-awareness can be easily incorporated into the analysis via the activity factor $p$. The effect of the activity factor $p$ is shown in Fig.~\ref{fig:Out_downlink}.

\begin{figure} 
	\begin{center}
	  \includegraphics[width=3 in]{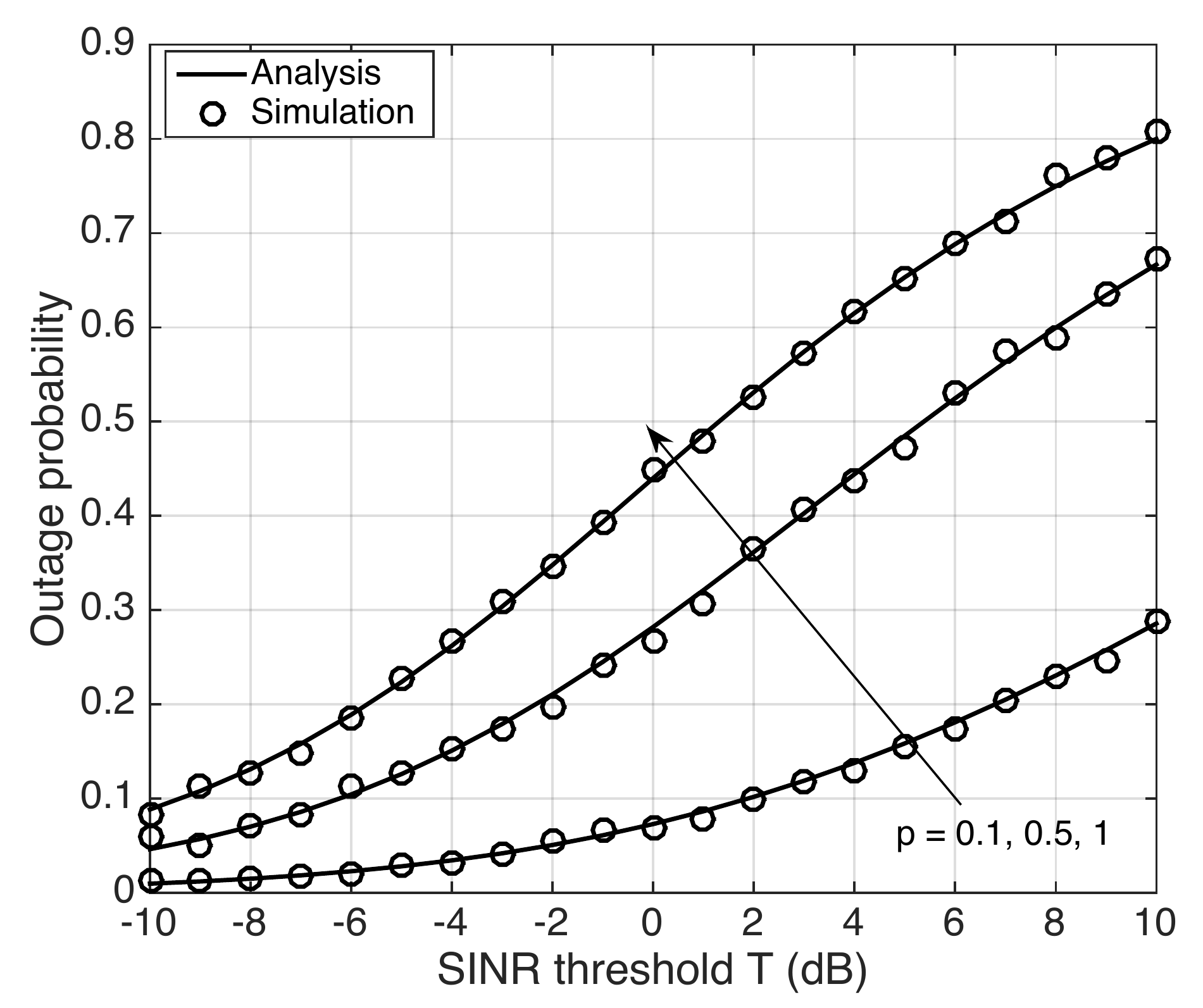}
	\end{center}
	\caption{Outage probability vs the SINR threshold for different BSs activity factor (p).}
	\label{fig:Out_downlink}
\end{figure}

\subsection{Multi-tier Cellular Networks}

Cellular networks are no longer single-tiered networks with operator's deployed macro BSs (MBSs) only. This is because MBSs are expensive to deploy in terms of time and money, which obstruct cellular operators to cope with the rapidly increasing capacity demand and device populations. Therefore, cellular operators tend to expand their networks via small BSs (SBSs) which are cheaper and faster to deploy. Some of these SBSs can be deployed directly by users in a plug and play fashion such as the LTE femto access points, which are installed by users at their homes and/or workplaces. Therefore,  modern cellular networks are multi-tiered networks that are composed of MBSs and several types of SBSs (e.g., micro, pico, femto).

The common assumption in SG analysis is to model multi-tier cellular networks via mutually independent teirs of BSs. On each tier, the BS locations follow an independent PPP which is characterized by its own transmission power $P_k$, intensity $\lambda_k$, and path-loss exponent $\eta_k$. It is usually assumed that UEs are associated to BSs according to a biased RSS strategy, which is controlled by a set of bias factors $\{B_1, B_2,\dots, B_k, \dots \}$. The bias factors are manipulated to control the load served by each network tier as shown in Fig.~\ref{fig:bias}. Let $\tilde{\Psi}_k=\{r_{0,k}, r_{1,k}, r_{2,k}, ..\}$ be the set of the ordered distances between a test user at the origin and the BSs in $\Psi_k$, in which $r_{i-1,k} < r_{i,k} < r_{i+1,k}$, for $\forall i \in \mathbb{Z}$. Then, assuming $K$ teirs of BSs, the test UE chooses to associate with tier $k \in \{1,2,...,K\}$ if 

\small
\begin{equation}
B_k P_k r_{0,k}^{-\eta_k} > B_i P_i r_{0,i}^{-\eta_i};  \quad 
{ i \in \{1,2,...,K\}}, 
\label{assoc_rule1}
\end{equation}
\normalsize

\noindent for all $i \neq k$. For simplicity, we focus on the case where all tiers have a common path-loss exponent $\eta_k=4$. The general case analysis can be found in \cite{load_aware_harpreet,k_tier}. Hence, the association rule becomes
 
\small
\begin{equation}
B_k P_k r_{0,k}^{-4} > B_i P_i r_{0,i}^{-4};  \quad 
{ i \in \{1,2,...,K\}},
 \label{assoc_rule}
\end{equation}
\normalsize

\begin{figure}[t]
	\begin{center}
	  \subfigure[With no biasing $(B_2=B_1)$]{\includegraphics[width=2.5in]{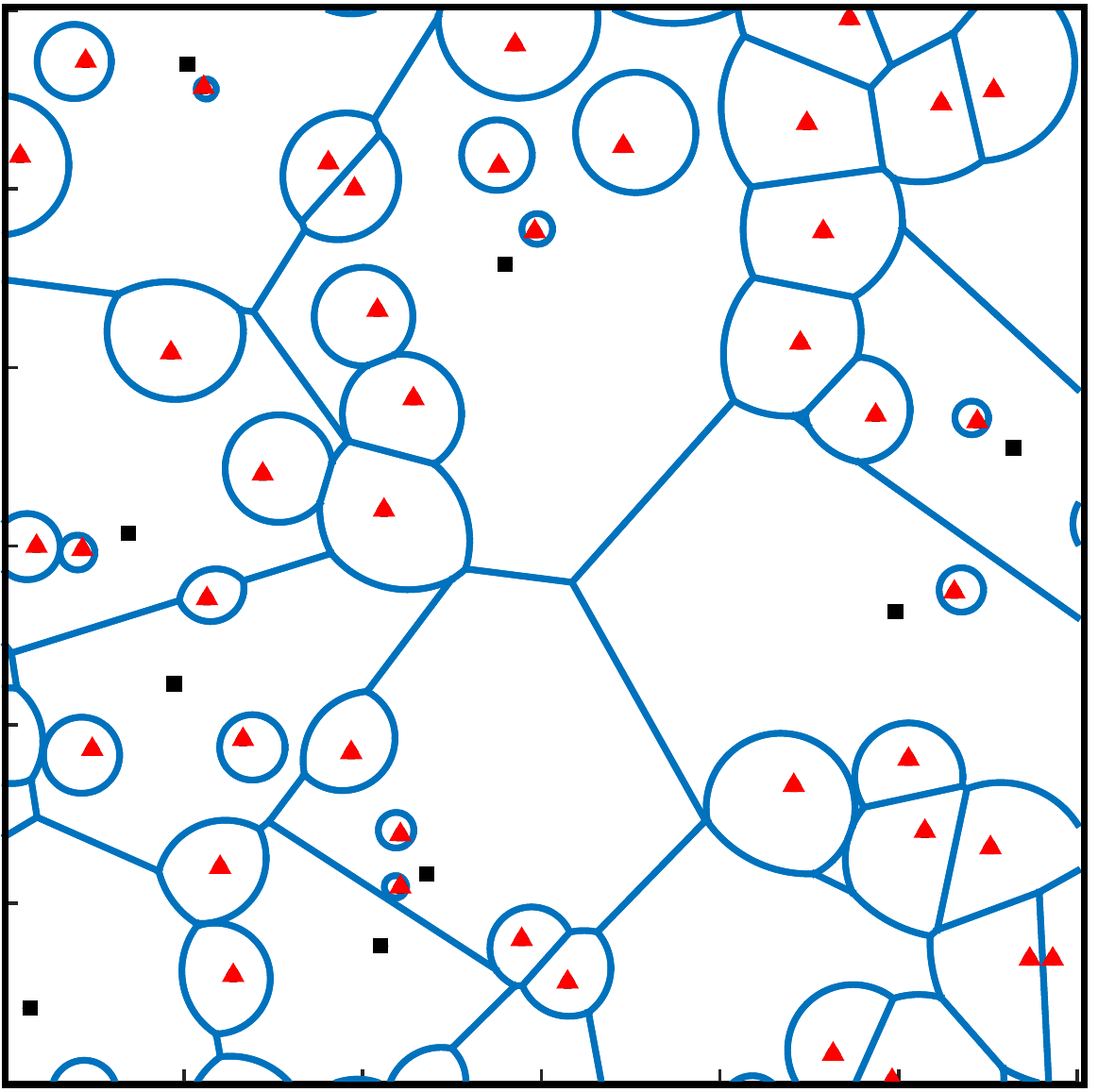}}
		\subfigure[With biasing $(B_2=10 B_1)$]{\includegraphics[width=2.5in]{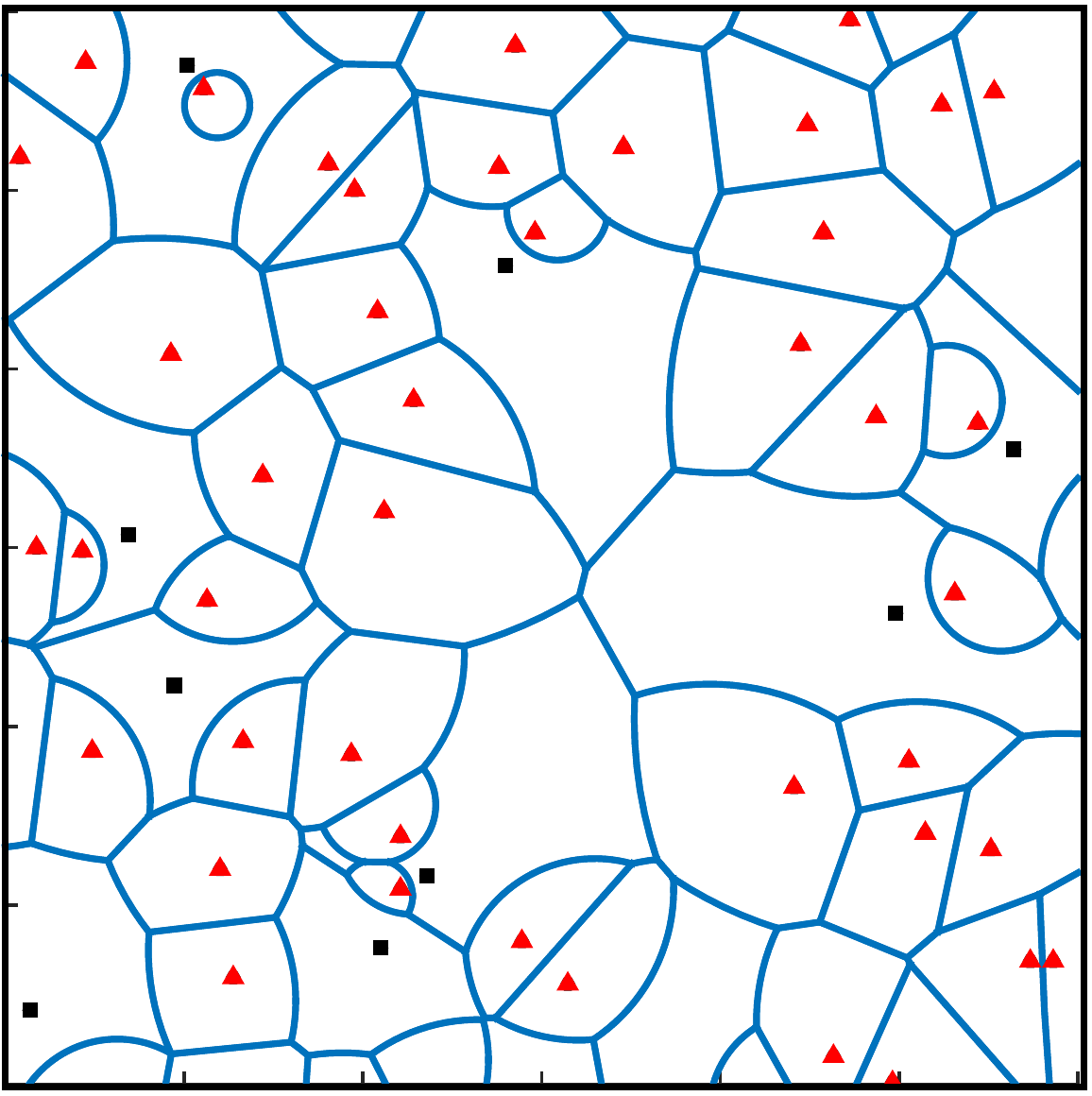}}
	\end{center}
	\caption{{ Two-tier cellular network with the same BS locations and different bias factors, in which the squares indicate the MBSs, the triangles indicate the SBSs for $P_1 = 50 P_2$, $\lambda_1= 0.2 \lambda 2$ and $\eta=4$. Biasing is used to increase the coverage of SBSs to offload users from the MBSs to the SBSs.}}
	\label{fig:bias}
\end{figure}



\noindent for all $i \neq k$. The performance in each tier may differ according to its parameters. Thus, per-tier performance is usually conducted. Let us focus on a generic tier $k$. Then looking into \eqref{eq:ASEP_unified} and \eqref{eq:out_unified}, one can see that the LT of the aggregate interference power should be evaluated at $ \frac{a r_{0,k}^4}{P_k}$ to conduct the performance analysis for tier $k$. The aggregate interference in this case is the cumulative interference coming from all tiers. Assuming universal frequency reuse across all tiers, the aggregate interference from all tiers can be calculated as

\begin{align} \label{eq:aggregate_Ktier_interference}
\mathcal{L}_{\mathcal{I}_{agg}}(a,\boldsymbol{\Lambda},r_{0,k},\boldsymbol{R}) &= \mathbb{E}\left\{e^{- z\sum\limits_{i=1}^K \mathcal{I}_i}\right\} \notag \\
&\overset{(a)}{=} \prod\limits_{i=1}^K\mathbb{E}\left\{e^{- z \mathcal{I}_i} \right\} \notag  \\
&= \prod\limits_{i=1}^K \mathcal{L}_{\mathcal{I}_{i}}(z,\lambda_i,r_{\mathcal{I}_i}).
\end{align}

\noindent where $\boldsymbol{\Lambda} = \{\lambda_i\}_{i=1}^K$, $\boldsymbol{R} = \{r_{\mathcal{I}_i}\}_{i=1}^K$, $(a)$ follows from the independence between the different tiers, and $r_{\mathcal{I}_i}$ and $\mathcal{I}_i$  are, respectively, the interference boundary for the $i^{th}$ tier  and the aggregate interference from the $i^{th}$ tier. 

The LT of the interference from each tier is similar to \eqref{LT_start}. The per-tier interference boundary is obtained from the association rule given in \eqref{assoc_rule}. For a user who is associated with tier $k$ with the association distance $r_{0,k}$, the $i^{th}$ tier interference should have the intensity $\lambda_i$ and interference boundary 

\begin{equation}
r_{\mathcal{I}_i} = r_{0,i} > \left(\frac{B_i P_i}{B_k P_k}\right)^\frac{1}{4} r_{0,k}.
\end{equation}  

\noindent From \eqref{LT_start} with $P_0=P_k$ and $P_\mathcal{I} = P_i$, the LT for the per-tier interference can be expressed as

\small
\begin{align} \label{eq:pt_LT1}
\mathcal{L}^{(k)}_{\mathcal{I}_i}\left(a,\lambda_i,r_{0,k},r_{\mathcal{I}_i}\right) & =\notag \\
& \!\!\!\!\!\!\!\!\!\!\!\!\! \exp \left\{-{ \pi \lambda_i \sqrt{\frac{a P_i}{P_k}} } r_{0,k}^2 \arctan\left( \sqrt{\frac{aB_k}{B_i }}   \right)  \right\}. 
\end{align} 
\normalsize

\noindent Combining \eqref{eq:aggregate_Ktier_interference} and \eqref{eq:pt_LT1}, the LT of the aggregate interference experienced by a user in tier $k$ is

\small
\begin{align} 
 \mathcal{L}^{(k)}_{\mathcal{I}_{agg}}\left(a,\boldsymbol{\Lambda},r_{0,k}, \boldsymbol{R}\right)&= \notag \\
& \!\!\!\!\!\!\!\!\!\!\!\!\!\!\!\!\!\!\!\!\!\! \exp \left\{- \sum_{i=1}^k{ \pi \lambda_i \sqrt{\frac{a P_i}{P_k}} } r_{0,k}^2  \arctan\left( \sqrt{\frac{aB_k}{B_i }}   \right)  \right\}. 
\end{align} \label{eq:pt_LT}
\normalsize

Similar to Section~\ref{Relaxing_r_0}, the service distance $r_{0,k}$ is random with the PDF shown in \eqref{mt_distance}, which is a function of the relative values of the tiers' powers, bias factors, and path-loss exponents. In our case (i.e., $\eta_k=4$, $\forall k$), the service distance distribution for a user in the $k^{th}$ tier reduces to

\small
\begin{equation}
f_{r_{0,k}}(x)=2 \pi \left(\sum_{i=1}^K \sqrt{\frac{B_i P_i}{B_k P_k}} \lambda_i\right) x \exp\left\{-\pi \sum_{l=1}^K \sqrt{\frac{B_l P_l}{B_k P_k}} \lambda_l x^{2} \right\}.
\end{equation}
\normalsize

\noindent The spatially averaged LT for users in the $k^{th}$ tier is then given by

\small
\begin{align} 
 & \mathcal{L}^{(k)}_{\mathcal{I}_{agg}}(a,\boldsymbol{\Lambda}) \notag \\
 &= \int\limits_0^\infty  \mathcal{L}_{\mathcal{I}_{agg}}\left(a,\boldsymbol{\Lambda},r,\left\{\frac{B_iP_i}{B_k P_k} r\right\}_{i=1}^K\right) f_{r_{0,k}}(r) \mathrm{d}r \notag \\
&= \frac{\sum_{i=1}^K \sqrt{{B_i P_i}} \lambda_i}{\sum_{l=1}^K \sqrt{{B_l P_l}} \lambda_l \left(1+ \sqrt{\frac{a B_k}{B_l}}  \arctan\left( \sqrt{\frac{a B_k}{B_l}} \right) \right)  }.
\end{align} 
\normalsize

\noindent For $\eta =4$, the tier association probability in \eqref{assoc_prob} reduces to

\begin{equation} \label{eq_red_association_prob}
\mathcal{A}_k = \frac{\lambda_k \sqrt{B_k P_k}}{\sum_{i=1}^K \lambda_i \sqrt{B_i P_i}},
\end{equation} 

\noindent Using \eqref{eq_red_association_prob} the averaged LT is given by

\small
\begin{align} 
&\!\!\!\!\! \mathcal{L}_{\mathcal{I}_{agg}}(a, \boldsymbol{\Lambda}) \notag \\
&\!\!\!\!\! = \sum_{k=1}^{K} \mathcal{A}_k \mathcal{L}^{(k)}_{\mathcal{I}_{agg}}(a,\boldsymbol{\Lambda}) \notag \\ 
& \!\!\!\!\!= \sum_{k=1}^k {\left(\sum_{i=1}^{k} \frac{\lambda_i}{\lambda_k} \sqrt{\frac{ B_i P_i}{ B_k P_k}} \left[1+ \sqrt{\frac{a B_k}{B_i}}  \arctan\left( \sqrt{\frac{a B_k}{B_i}} \right)\right]\right)^{-1}} \!\!\!\!\!\!.
\end{align} 
\normalsize

\noindent If unbiased RSS association is adopted (i.e., $B_k = 1$, $\forall k$), then the LT reduces to
\small
\begin{align} \label{mt_simple} 
\mathcal{L}_{\mathcal{I}_{agg}}(a,\boldsymbol{\Lambda})= \frac{1}{1+ \sqrt{a }  \arctan\left( \sqrt{a} \right)}.
\end{align} 
\normalsize

Despite the different transmission powers and intensities of BSs in multi-tier cellular networks, the simple expression in \eqref{mt_simple} shows that the unbiased RSS association reduces the SINR-dependent performance metrics to the single-tier case, which is independent from network parameters (i.e., numner of tiers, transmission powers, intensities of BSs, etc.).

\subsection{Interference Coordination and Frequency Reuse}

{For simplicity, we study a user-centric interference coordination with frequency reuse in a single-tier cellular network modeled via a PPP with intensity $\lambda$. Due to the randomized network structure modeled by the PPP, the traditional hexagonal grid tailored frequency reuse schemes cannot be employed. Therefore, we assume that the available spectrum is divided into $\Delta$ sub-bands and that frequency reuse is adopted via coordination among the BSs~\cite{Gaus_approx1}. As shown in Fig.~\ref{fig:reuse}, each BS uses a frequency sub-band which is not used by the $\Delta -1$ BSs closest to its serving user.} The main problem in frequency reuse is that the positions of interfering BSs are correlated (i.e., the BSs that are using the same sub-band), which violates the PPP assumption. For analytical tractability, the usual method that is used in such cases is to approximate the set of interfering BSs with a PPP with intensity $\frac{\lambda}{\Delta}$. It is well perceived that approximating a repulsive PP by a PPP that have equivalent intensity gives an accurate estimate for the interference if the exclusion distance $r_{\mathcal{I}}$ around the test receiver is accurately calculated \cite{survey_h, uplink_H, cog_h, hcpp, dense_ppp}. In our case, since each BS selects one of the $\Delta$ sub-bands, the intensity of the interfering BSs on each sub-band is $\lambda/\Delta$. Exploiting the equi-dense PPP approximation, the LT of the aggregate interference in the form of \eqref{LT_start} is legitimate to be used. 

\small
\begin{align} \label{LT_coordination1}
\mathcal{L}_{\mathcal{I}_{agg}}(a,\frac{\lambda}{\Delta}, r_0, r_\mathcal{I}) & = \exp\left\{-\pi \frac{\lambda}{\Delta} \sqrt{a} r_0^2 \arctan\left( \left(\frac{r_0}{r_\mathcal{I}}\right)^2\sqrt{a}\right) \right\}
\end{align}
\normalsize

{The adopted user-centric coordination imposes an increased geographical interference protection around UEs, and hence, $r_\mathcal{I} > r_0 $. Particularly, since each BS is using a frequency which is not used by the nearest $\Delta-1$ neighbors, the geographical interference protection is given by $r_\mathcal{I} = r_{\Delta-1}$.} Note that $r_{\Delta-1}$ and $r_{0}$ are correlated with the joint PDF in \eqref{eq:joint}. Averaging over the joint PDF of $r_{\Delta-1}$ and $r_{0}$, the spatially averaged LT of the aggregate interference is given by \eqref{LT_final}. 

\begin{figure}[t]
	\begin{center}
	 \includegraphics[width=2.5in]{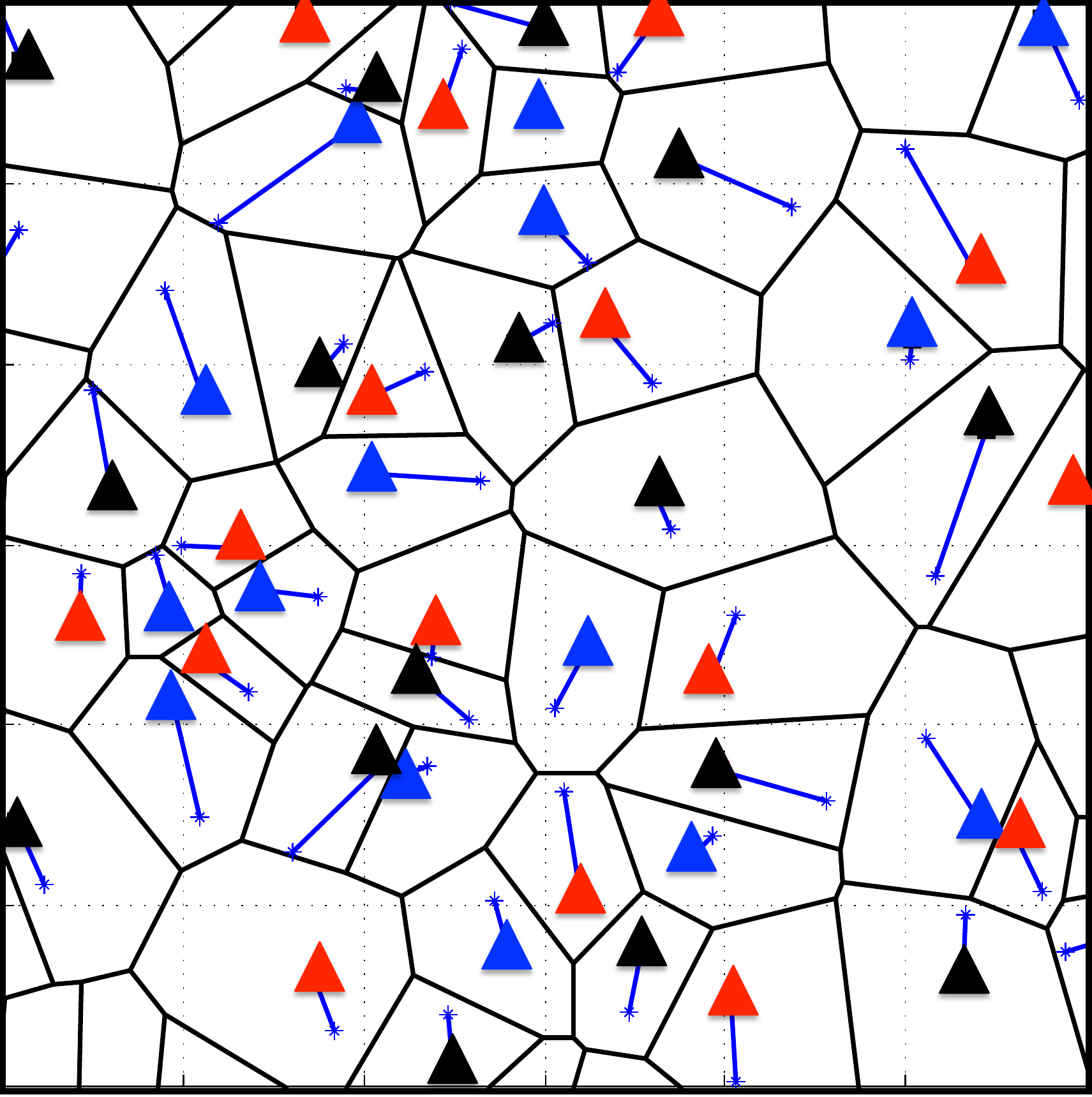}

	\end{center}
	\caption{{User-centric coordinated frequency reuse with $\Delta=3$, in which the BSs are represented by triangles and UEs are represented by stars. BSs using the same frequency are highlighted with similar color.} }
	\label{fig:reuse}
\end{figure}

\begin{figure*}
\
\begin{align}
\!\!\!\!\mathcal{L}_{\mathcal{I}_{agg}}(a,\frac{\lambda}{\Delta}) =  \int\limits_0^\infty \int\limits_{x}^\infty  \frac{4 (\pi \lambda)^{\Delta} y (y^2 - x^2)^{\Delta-2}}{\Gamma(\Delta-1)}  \exp\left\{-\pi \lambda \left( \frac{\sqrt{a} x^2}{\Delta}  \arctan\left( \left(\frac{x}{y}\right)^2\sqrt{a}\right) - y^2\right) \right\} \mathrm{d}y  \mathrm{d}x
\label{LT_final}
\end{align}
\hrulefill
\end{figure*}

It is important to highlight that the conditional PDF in \eqref{eq:joint} is based on the BSs intensity $\lambda$ not $\frac{\lambda}{\Delta}$. This is because the UEs have the opportunity to associate with the complete set of BSs with intensity $\lambda$. However, once associated, it communicates on one of the $\Delta$ sub-bands which interferes with a subset of the BSs with intensity $\frac{\lambda}{\Delta}$.

It is obvious that interference coordination and frequency reuse have complicated the analysis, resulting in a double integral expression for the spatially averaged LT of interference in \eqref{LT_final}. However, such expression is still valuable as it can be efficiently evaluated in terms of time and complexity when compared to Monte Carlo simulations.

Fig.~\ref{fig:out_reuse} validates \eqref{LT_final} and shows the effect of the coordinated frequency reuse on the network outage probability. As shown in \eqref{LT_coordination1} and \eqref{LT_final}, coordinated frequency reuse affects both the interference boundary and the interferers intensity. This explains the significant performance improvement shown in Fig.~\ref{fig:out_reuse} for increasing the reuse factor $\Delta$.

\begin{figure}[t]
	\begin{center}
	 \includegraphics[width=3in,height=3in]{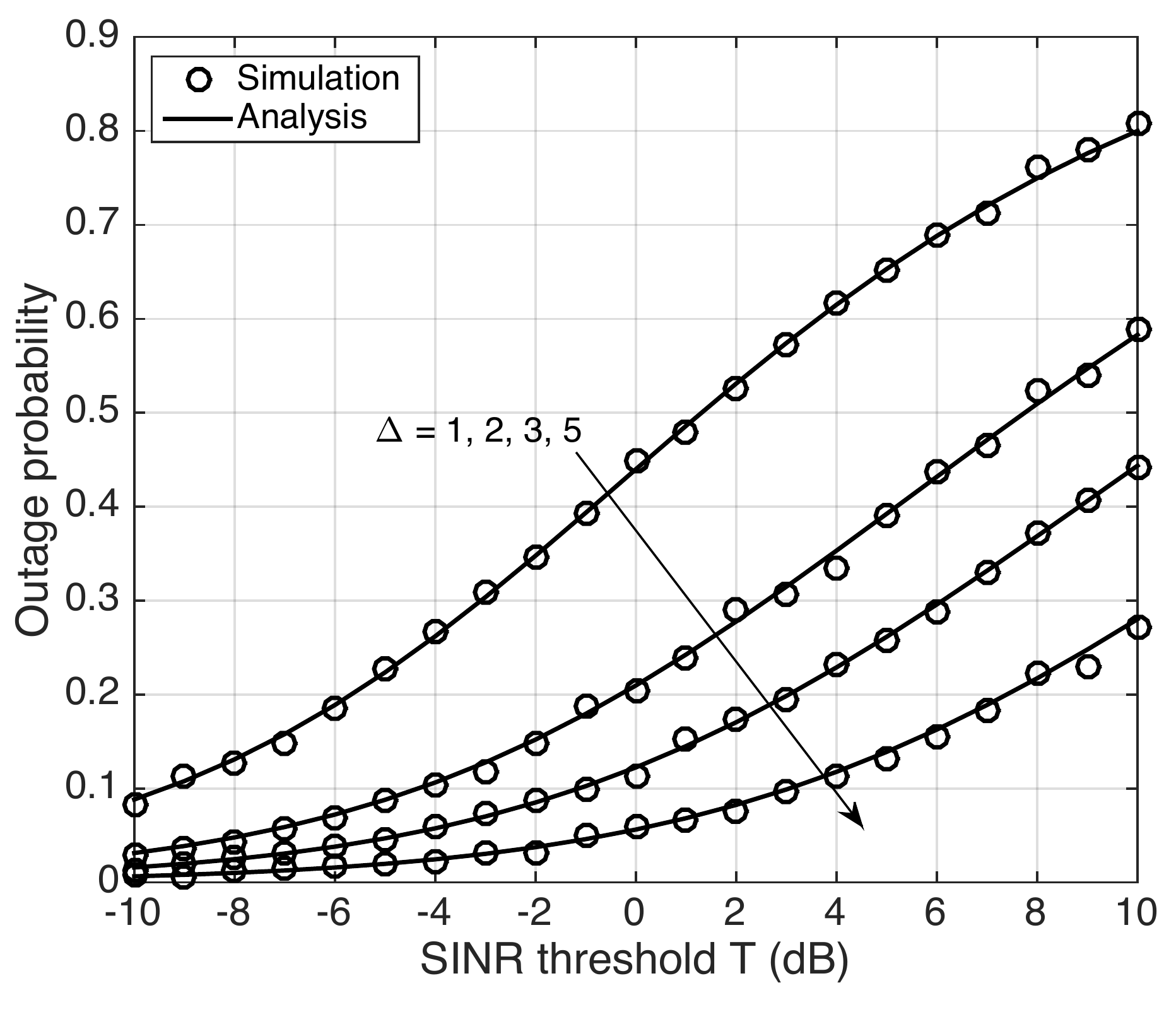}

	\end{center}
	\caption{ Outage probability for coordinated frequency reuse with $\Delta=1,\;2,\;,3,\;\text{and}\;5$.  }
	\label{fig:out_reuse}
\end{figure}

\subsection{Uplink Transmission}

For the uplink case study, we assume a single-tier cellular network, as shown in Fig.~\ref{fig:uplink}, with intensity $\lambda$ and universal frequency reuse and no intra-cell interference. That is, each BS assigns a unique channel per-user. It is also assumed that the UEs constitute an independent PPP with intensity $\lambda_u \gg \lambda $ such that each BS always has a user to serve on each channel. Per-UE power control is a crucial assumption in the uplink case  to limit interference between the users, as shown in \cite{uplink_harpreet, uplink_H}. For simplicity, we assume a full channel inversion power control in which each user invert its path-loss to maintain a constant average power level of $\rho$ at the serving BS. That is, if the UE is located $r$ meters away from its serving BS, the transmit power should be $\rho r^\eta \overset{\tiny (\eta=4)}{=} \rho r^4$ to have the signal power level of $\rho$ at the serving BS. 

In the uplink case, the interfering sources are the UEs and the receivers are the BSs. Without loss in generality, the test BS is assumed to be located at the origin. Although the complete set of UEs constitutes a PPP  with intensity $\lambda_u$, the interfering (i.e., simultaneously active) UEs on a certain channel do not constitute a PPP with intensity is $\lambda_u$. Due to the unique channel assignment per BS, only one active user per channel is allowed in each voronoi cell as shown in Fig.~\ref{fig:uplink}. This brings correlation, in the form of repulsion, among the set of interfering users. To facilitate the analysis and maintain tractability, the set of interfering UEs is approximated  with a PPP with the same intensity. Since there is only one active user in each voronoi cell, the intensity of the approximate PPP is selected to be equal to the BSs intensity $\lambda$. In this case, the PGFL of the PPP is legitimate to be used as an approximation to obtain the LT of the aggregate interference in uplink cellular networks. The accuracy of this approximation is verified in Fig.~\ref{fig:out_uplink} as well as in \cite{uplink_H, uplink2_jeff, uplink_harpreet, Uplink_lett, uplink_zolfa, uplink_alamouri, Laila_Uplink}. 

Although the set of interfering UEs is approximated via a PPP, the LT in \eqref{LT_start1} cannot be directly used.  This is because the employed power control imposes a constant received signal power $\rho$ at the test BS. As a result, the SINR expression for the uplink is different from that of the downlink case presented in \eqref{eq:sinr_GCB}. The SINR at the test BS in the uplink case is given by

\begin{equation} \label{SINR_uplink}
\Upsilon_u = \frac{\rho h}{\sigma^2 + \mathcal{I}}
\end{equation}

\noindent Ignoring noise and replacing $z$ by $\frac{a}{\rho}$ in \eqref{LT_start1}, the starting LT for the uplink case is given by  

\small
\begin{align} \label{LT_start_uplink}
\mathcal{L}_{\mathcal{I}_{agg}}(a,\lambda, r_\mathcal{I}) \notag \\
& \!\!\!\!\!\!\!\!\!\!\!\!\!\!\!\!\!\!\!\!\!\!\!\!\!\!\!\!\!\!\!\!\!= \exp\left\{-\pi \lambda \EXs{P_\mathcal{I}}{ \sqrt{\frac{a P_\mathcal{I}}{\rho}}  \arctan\left( \left(\frac{1}{r_\mathcal{I}}\right)^2\sqrt{\frac{aP_{\mathcal{I}}}{\rho}}\right) } \right\}
\end{align}
\normalsize
which is no longer a function of $r_0$. Nevertheless, the distributions of the service distances $r_0$ affect the interference power $P_{\mathcal{I}_i}$ from each UE due to the employed power control. In other words, the transmission power of each UE is a function of the random distance to his serving BS, which has the distribution in \eqref{pdf_r_0}. Assuming that all interfering UEs have i.i.d. transmission powers, \eqref{LT_start_uplink} should be averaged over the distribution of $P_\mathcal{I}$. Note that the averaging over $P_\mathcal{I}$ is done within the PGFL expression (i.e., within the exponential function of \eqref{LT_start_uplink}) because  $P_\mathcal{I}$ takes a different realization for each interfering user.

The interference boundary for the uplink is given by 
\begin{equation} \label{boundary}
r_\mathcal{I} > \left(\frac{P_\mathcal{I}}{\rho}\right)^\frac{1}{\eta} \overset{(\eta=4)}{=} \left(\frac{P_\mathcal{I}}{\rho}\right)^\frac{1}{4}
\end{equation}
which is calculated from the employed power control and the association rule. That is, each user adjusts its power to maintain the power level $\rho$ at his nearest BS. Hence, the interfering power from any other user at the test BS satisfies $P_\mathcal{I} r_\mathcal{I}^{-\eta}< \rho$, which leads to the boundary in \eqref{boundary}. Substituting $r_\mathcal{I}$ back into \eqref{LT_start_uplink}, we have

\small
\begin{align} \label{LT_start_uplink2}
\mathcal{L}_{\mathcal{I}_{agg}}(a,\lambda)= \exp\left\{-\pi \lambda \mathbb{E}\left\{\sqrt{P_\mathcal{I}}\right\} \sqrt{\frac{a}{\rho}}  \arctan\left( \sqrt{a} \right)  \right\}
\end{align}
\normalsize

The power $P_\mathcal{I} = \rho r^\eta$, where $r$ has the PDF in \eqref{pdf_r_0}. Hence, $ \mathbb{E}\left\{\sqrt{P_\mathcal{I}}\right\} = \frac{\sqrt{\rho}}{\pi \lambda}$. Substituting $\mathbb{E}\left\{\sqrt{P_\mathcal{I}}\right\}$ with $\frac{\sqrt{\rho}}{\pi \lambda}$ back into \eqref{LT_start_uplink2}, we have

\small
\begin{align} \label{LT_start_uplink3}
\mathcal{L}_{\mathcal{I}_{agg}}(a,\lambda)= \exp\left\{- \sqrt{{a}}  \arctan\left( \sqrt{a} \right)  \right\}
\end{align}
\normalsize
 which is independent of the power control threshold $\rho$ and the BS intensity $\lambda$.
More advanced uplink system models with fractional power control and/or maximum transmit power constraint can be found in \cite{uplink_H, uplink_harpreet, uplink2_jeff, uplink_alamouri, Laila_Uplink}.

\begin{figure}[t]
	\begin{center}
	 \includegraphics[width=2.5in]{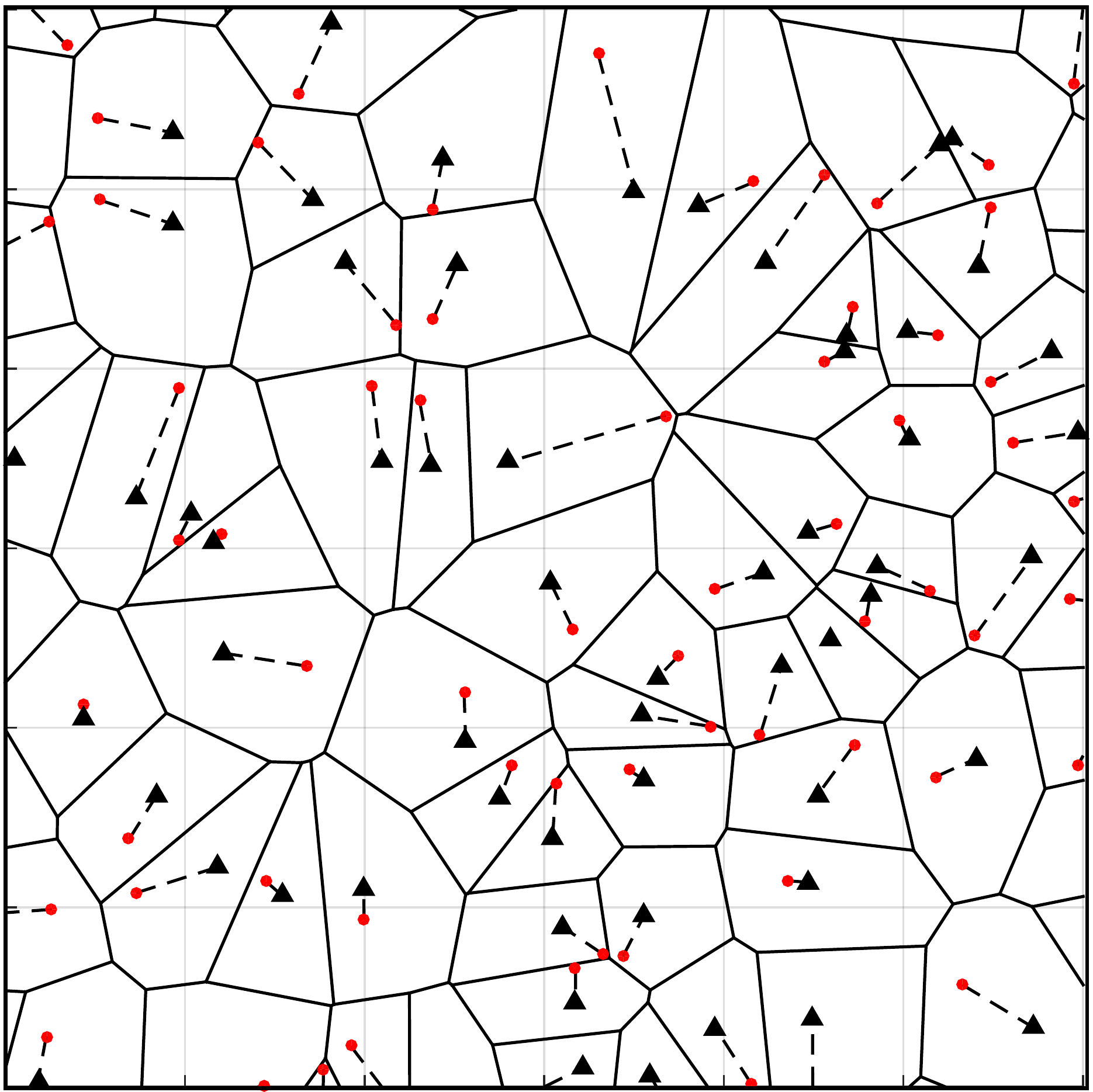}

	\end{center}
	\caption{ Single-tier cellular network in which the triangles indicate the MBSs, the circles indicate UEs, the dotted lines indicate association, and sold lines indicate BSs footprints. }
	\label{fig:uplink}
\end{figure}

\begin{figure}[t]
	\begin{center}
	  \includegraphics[width=3.0in]{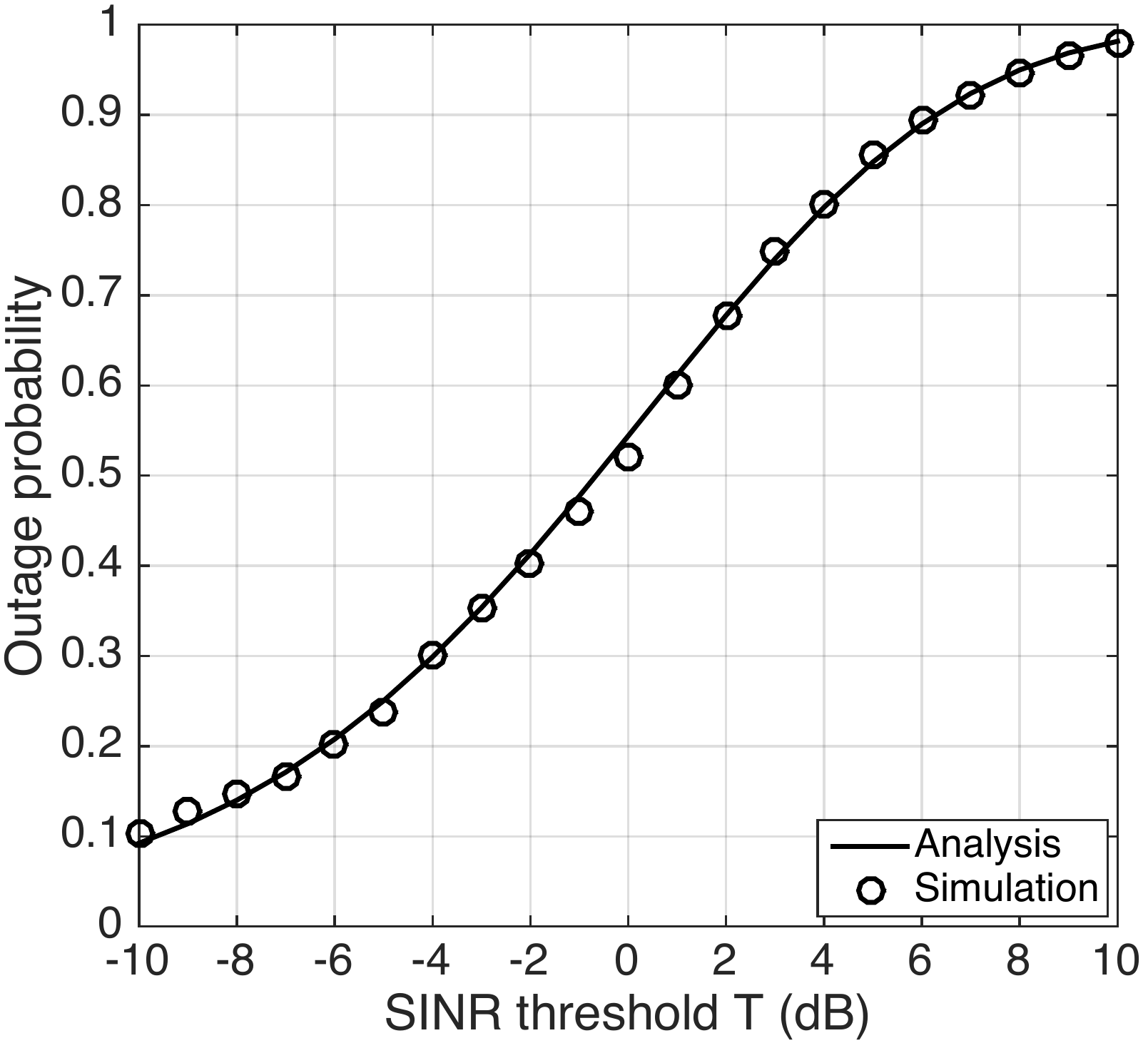}
	\end{center}
	\caption{ Outage probability vs SINR threshold for the uplink at different values of $\rho$.}
	\label{fig:out_uplink}
\end{figure}

Fig.~\ref{fig:out_uplink} verifies \eqref{LT_start_uplink3} and the PPP approximation for the interfering UEs. Comparing Fig.~\ref{fig:out_uplink} with Fig.~\ref{fig:Out_downlink}, it can be observed that the uplink transmission has higher outage probability than the downlink counterpart. This is because uplink transmissions have limited transmission power and the association does not impose geographical interference protection for the uplink transmission. Hence, the uplink is more vulnerable to outages than the downlink. Comparison between uplink and downlink performance can be found in \cite{uplink_H}.

\subsection{General Fading}

All of the above analysis is based on the exponential power fading (i.e., Rayleigh environment) assumption, which enables expressing the ASEP, outage probability, and ergodic rate using the LT of the aggregate interference. Assuming general fading on the interfering links, the analytical tractability is not affected as we can still express all performance metrics using the LT of the aggregate interference. Nevertheless, the expression of the LT of the aggregate interference may become more involved. Tractability issues occur when the  exponential power fading on the useful link is changed. In this case, the outage probability and ASEP can no longer be expressed in terms of the LT of the aggregate interference.\footnote{Unlike outage probability and ASEP, the ergodic rate can always be expressed in terms of the LT of the aggregate interference~(cf. \cite[Lemma~1]{hamdi2010useful}), and hence, can be evaluated for general fading environment~\cite{Rate_marco}.} In \cite{survey_h}, the authors discuss four techniques which are used in the literature to extend SG analysis to other fading environments. These techniques are to:
\begin{itemize}
\item approximate the interference using a certain PDF via moments fitting, in which the moments are obtained for the interference LT;
\item resort to bounds by considering dominant interferers only and/or statistical inequalities;
\item use Plancherel-Parseval theorem to obtain the aforementioned performance metrics via complex integrals in the Fourier transform domain;
\item inversion (e.g., Gil-Pelaez inversion theorem~\cite{Gil_marco}).
\end{itemize}
We will not delve into the details of these techniques as they are already discussed in \cite{survey_h}. However, we will highlight two important exceptions.

\subsubsection{Nakagami-m} \label{sec:Nakagami} The first scenario where the above analysis holds is the Nakagami-m fading with integer $m$. For the ASEP analysis, \cite{hamdi_useful_tech} obtains expressions for $\mathbb{E}\{\text{erfc}(h/x)\}$ and $\mathbb{E}\{\text{erfc}^2(h/x)\}$ using the LT of $X$, where $h$ is gamma distributed with integer shape  parameter as shown in Appendix~\ref{app:useful}. Note that the LT of the aggregate interference in Nakagami-m fading changes from \eqref{LT_start} to

\small
\begin{align} \label{LT_start_fading}
\!\!\!\!\!\!\mathcal{L}_{\mathcal{I}_{agg}}(a,\lambda, r_0, r_\mathcal{I}) & \notag \\
& \!\!\!\!\!\!\!\!\!\!\!\!\!\!\!\!\!\!\!\!\!\!= \exp\left\{-\pi \lambda r_0^2 {}_2F_1\left(-\frac{2}{\eta},m,1-\frac{2}{\eta},- \left(\frac{r_0}{r_\mathcal{I}}\right)^\eta {\frac{aP_{\mathcal{I}}}{P_0}}\right) \right\}
\end{align}
\normalsize

The outage probability and ergodic rate can be computed from the CDF of the SINR as shown in Section~\ref{sec:out}. In the Nakagami-m case, the authors in \cite{MIMO_Hetnet_harpreet} show that if $m$ is an integer, the CDF of the SINR can be expressed in terms of the LT of the aggregate interference using the following identity 
\begin{equation} \label{LT_identity}
t^n f(t) \stackrel{LT}{\longrightarrow} (-1)^k \frac{d^k\mathcal{L}_{f(t)}\left(s\right)}{ds^k}.
\end{equation}

Let $h$ be a gamma random variable with shape parameter $U$ and scale parameter $1$. From  \eqref{out_cdf} we have

\small
{\begin{align}  \label{nakagami}
F_{\Upsilon}(T) &= \int_ x F_{h}\left( \frac{T \mathcal{I}_{agg}}{P r_0^{-\eta}} \right) f_{\mathcal{I}_{agg}}(x) \mathrm{d}x \notag \\
  &\overset{(a)}{=} 1- \int_x \sum_{u=0}^{U-1}  \frac{1}{u!} \left(\frac{T \mathcal{I}_{agg} r_0^{\eta}}{P }\right)^u e^{-\frac{T \mathcal{I}_{agg} r_0^{\eta}}{P }}  f_{\mathcal{I}_{agg}}(x) \mathrm{d}x  \notag \\
    &\overset{(b)}{=} 1- \sum_{u=0}^{U-1} \frac{(-1)^u}{u!} \left(   \frac{T  r_0^{\eta}}{P }  \right)^u  \left.\frac{\mathrm{d}^u\mathcal{L}_{I_{\rm agg}}\left(z\right)}{\mathrm{d} z^u}\right|_{z=\frac{T  r_0^{\eta}}{P }}  
\end{align}}
\normalsize

\noindent{where $(a)$ follows from the CDF of the gamma distribution with integer shape parameter, and $(b)$ follows from switching the integral and summation order, the LT definition, and the identity in \eqref{LT_identity}.}

\subsubsection{Additional Slow Fading} When an additional slow fading is incorporated into the analysis on top of the exponential or Nakagami-m fading, the analysis remains tractable if the RSS association adapts to the slow fading. That is, the users are always associated to the BS that provides the highest received signal strength. Applying the displacement theorem \cite{martin_book}, the effect of shadowing is captured by scaling the intensity of the PPP with the shadowing fractional moment $\mathbb{E}\left\{x^\frac{2}{\eta}\right\}$, where $x$ is the shadowing random variable~\cite{shadowing_letter}.

\subsection{Multiple Input Multiple Output (MIMO) Antenna Systems}

Due to the vast diversity of available MIMO techniques and the significant differences between their operations, it is difficult to present a unified analytical framework for all MIMO case studies. Further, we do not want to lose the tutorial flavor and delve into MIMO systems details, which already exist elsewhere in the literature. Therefore, we choose to present a simple receive diversity MIMO case study just to convey the idea of extending SG analysis to MIMO systems. MIMO with transmit diversity is discussed in the next section in the context of network MIMO.

This section considers a downlink cellular network with receive diversity, where each BS is equipped with a single antenna and each UE is equipped with $N_r$ antennas. Note that in SG analysis, the multiple antennas are usually assumed to be collocated. The channel gain vector between a transmitting antenna and the $N_r$ receiving antennas is denoted by $\boldsymbol{h} \in \mathbb{C}^{N_r \times 1}$, which is assumed to be composed of i.i.d circularly symmetric unit variance complex Gaussian random variables. Also, it is assumed that the UEs have perfect channel information for the intended channel vector $\boldsymbol{h}_0$. Assuming maximum ratio combining (MRC) receivers, the baseband received signal at the input of the decoder can be represented as  

\begin{align}
y = {\boldsymbol{h}_0^H} \left( \sqrt{P} r_0^{-\frac{\eta}{2}}   \boldsymbol{h}_0 s_0 + \underset{r_j \in \tilde{\Psi} \setminus r_0}{\sum} \sqrt{P} r_j^{-\frac{\eta}{2}}  \boldsymbol{h}_j s_j + \boldsymbol{n} \right) \notag \\
= \sqrt{P} r_0^{-\frac{\eta}{2}} \boldsymbol{h}_0^H  \boldsymbol{h}_0  s_0 + \underset{r_j \in \tilde{\Psi}\setminus r_0}{\sum} \sqrt{P} r_j^{-\frac{\eta}{2}} {\boldsymbol{h}_0^H \boldsymbol{h}_j} s_j + {\boldsymbol{h}_0^H}  \boldsymbol{n} 
\end{align}

\noindent where $ \boldsymbol{n} \in \mathbb{C}^{N_r \times 1} $ is the noise vector with i.i.d complex Gaussian elements.  Conditioning on $\Xi=\left\{ \boldsymbol{h}_0, \boldsymbol{h}_i, \tilde{\Psi}\right\}$ and exploiting the Gaussian signaling assumption, the SINR can be expressed as

\begin{align} \label{eq:mrc}
\Upsilon(\Xi) &= \frac{ P r_0^{-\eta} {\boldsymbol{h}_0^H\boldsymbol{h}_0\boldsymbol{h}_0^H\boldsymbol{h}_0} }{ \underset{r_j \in \tilde{\Psi}\setminus r_0}{\sum} P r_j^{-{\eta}} {{\boldsymbol{h}_0^H \boldsymbol{h}_j \boldsymbol{h}_j^* \boldsymbol{h}_0}} + N_0 {\boldsymbol{h}_0^H \boldsymbol{h}_0} }  \notag \\
&{= \frac{ P r_0^{-\eta} \boldsymbol{h}_0^H\boldsymbol{h}_0 }{ \underset{r_j \in \tilde{\Psi}\setminus r_0}{\sum} P r_j^{-{\eta}} \frac{{{\boldsymbol{h}_0^H \boldsymbol{h}_j \boldsymbol{h}_j^* \boldsymbol{h}_0}}}{{\boldsymbol{h}_0^H \boldsymbol{h}_0}} + N_0  }}  \notag \\
&{=} \frac{ P r_0^{-\eta} {g_0} }{ \underset{r_j \in \tilde{\Psi}\setminus r_0}{\sum} P r_j^{-{\eta}} {{g}_j} + N_0 } 
\end{align}

\noindent {where $g_0$ and  $g_j$ in \eqref{eq:mrc} are the effective channel gains for the employed MIMO scheme. Let $h_{0,k}$ be the $k^{th}$ element of $\boldsymbol{h}_0$, then $g_0 = \sum_{k=1}^{N_r} h^*_{0,k} h_{0,k}$ is a summation of $N_r$ unit-mean exponential random variables. Hence, $g_0$ is gamma distributed with shape parameter $N_r$ and rate parameter $1$. On the other hand, due to the independence between $\boldsymbol{h}_0$ and $\boldsymbol{h}_i$,  the effective channel gain for the $i^{th}$ interfering link ($g_i$) is a unit-mean exponential random variable. Note that the exponential distribution of $g_i$ follows from the fact that $\frac{\boldsymbol{h}_0^H \boldsymbol{h}_j \boldsymbol{h}_j^H \boldsymbol{h}_0 }{|\boldsymbol{h}_0|^2} \overset{D}{=} h_{j1} h_{j1}^*$.}  Since the MRC receiver leads to a gamma distributed intended channel gain, ASEP and SINR CDF can be  obtained as in the case of Nakagami-m fading described in Section~\ref{sec:Nakagami}. For instance, the CDF of the SINR can be found as

\small
{\begin{align}  \label{CDF_div}
F_{\Upsilon}(T) {=} 1- \sum_{u=0}^{N_r-1} \frac{(-1)^u}{u!} \left(   \frac{T  r_0^{\eta}}{P }  \right)^u  \left.\frac{\mathrm{d}^u\mathcal{L}_{I_{\rm agg}}\left(z\right)}{\mathrm{d} z^u}\right|_{z=\frac{T  r_0^{\eta}}{P }}  
\end{align}}
\normalsize

\noindent{where} $\mathcal{L}_{I_{\rm agg}}\left(z\right)$ is given in \eqref{LT_start1} with $r_\mathcal{I}=r_0$. Fig.~\ref{fig:D_and_Re} validates \eqref{CDF_div} and shows the effect of receive diversity on the network outage probability.

From the simple example presented above, one can see that even in Rayligh fading environment, the fading in MIMO networks is no longer exponential, and hence, the analysis is more involved. Also, analyzing the distribution of the interfering signals is challenging as the interfering signal from each BS is multiplied by the precoding matrix tailored for processing the intended signal. Further, correlations within the interference at the antenna branches may impose additional complexity to the MIMO analysis. Nevertheless, the SG analysis has been greatly developed in recent years and modeled the performance of many MIMO setups with and without interference correlation \cite{STBC_harpret, Mimo_load_harpreet, Mimo_ordering, Corr_MRC_harpreet, Corr2_MRC_harpreet, eid_Mimo, marco_Mimo, asymptotic_SE, Dist_antenna_Zhang, SDMA_Zhang, MIMO_heath, Dist_antenna_Wei, Massive_harpret, Laila_MIMO}.

\begin{figure}[t]
	\begin{center}
	  \includegraphics[width=3.0in]{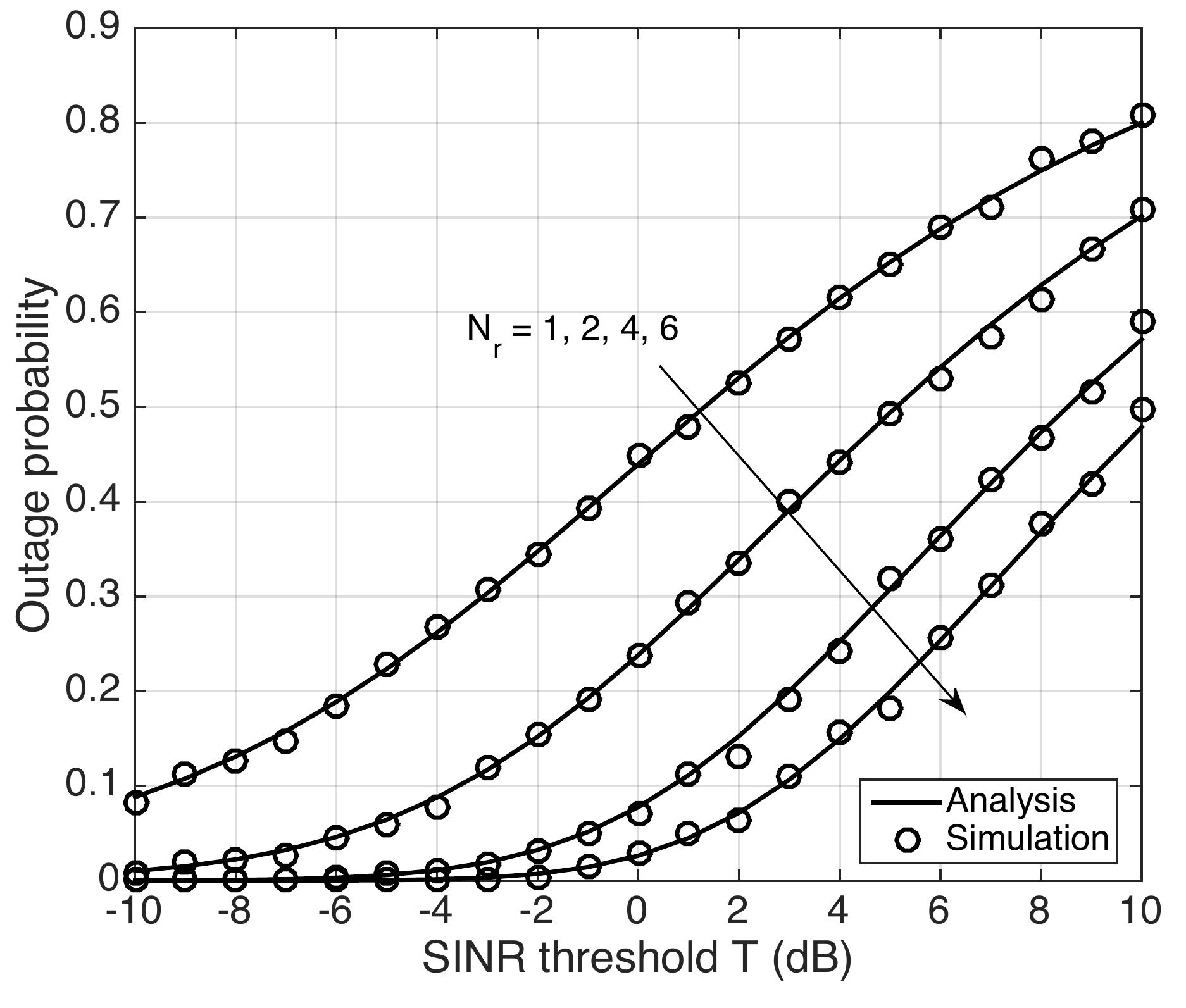}
	\end{center}
	\caption{ The effect of receive diversity on the outage probability.}
	\label{fig:D_and_Re}
\end{figure}

\subsection{Network MIMO}

In the previous section, it is implicitly assumed that the multiple antennas are collocated. In contrast, when several BSs cooperate to form a MIMO system, the antenna separations are prominent and should be taken into consideration. In this section, we consider a  downlink single-tier cellular network with single antenna BSs. User centric CSI agnostic coordinated multi-point (CoMP) transmission is enabled~\cite{CoMP_martin, CoMP_sakr}, in which each user is served by his nearest $n$ BSs. In this case, the test user receives $n$ non-coherent copies of his intended symbol from the $n$ nearest BSs, and the received baseband signal can be expressed as

\begin{align}
y = \sum_{i=0}^{n-1} \sqrt{P} r_i^{-\frac{\eta}{2}}   {h}_i s_0 + \!\!\!\!\!\!\!\!\!\!\!\!\!\!\!\!\!\!\!\!\!\!\!\!\!\!\!\!\!\!\!\! \underset{ \;\;\; \;\;\;  \;\;\;  \;\;\;  \;\;\;  \;\;\;  \;\;\;  r_j \in \tilde{\Psi} \setminus \{r_0, r_1, ..., r_{n-1} \} }{\sum} \!\!\!\!\!\!\!\!\!\!\!\!\!\!\!\!\!\!\!\!\!\!\!\!\!\!\!\!\! \sqrt{P} r_j^{-\frac{\eta}{2}}  {h}_j s_j + n  
\label{eq:cooperation2}
\end{align}

\noindent where the set $\{r_0, r_1, ..., r_{n-1} \}$ is excluded from $\tilde{\Psi}$ in \eqref{eq:cooperation2} as the nearest $n$ BSs do not contribute to the interference. The SINR can be written as

\begin{align}
\Upsilon = \frac{|\sum_{i=0}^{n-1} \sqrt{P} r_i^{-\frac{\eta}{2}}   {h}_i|^2}{  \!\!\!\!\!\!\!\!\!\!\!\!\!\!\!\!\!\!\!\!\!\!\!\!\!\!\!\!\!\!\!\! \underset{ \;\;\; \;\;\;  \;\;\;  \;\;\;  \;\;\;  \;\;\;  \;\;\;  r_j \in \tilde{\Psi} \setminus \{r_0, r_1, ..., r_{n-1} \} }{\sum} \!\!\!\!\!\!\!\!\!\!\!\!\!\!\!\!\!\!\!\!\!\!\!\!\!\!\!\!\! {P} r_j^{-{\eta}}  |{h}_j|^2 + N_0}  
\end{align}

\noindent where {$ |\sum_{i=0}^{n-1} \sqrt{P} r_i^{-\frac{\eta}{2}}   {h}_i|^2 $} is exponentially distributed with mean \small{$\sum_{i=0}^{n-1} P_i r_i^{-\eta}$}. \normalsize Substituting 

\begin{equation}
z = \frac{a}{{\sum_{i=0}^{n-1} P r_i^{-\eta}}} \notag
\end{equation}

\noindent {into} \eqref{LT_start1} and integrating over the joint PDF of the distances $f(r_0, r_1, ... r_n)$, the spatially averaged LT is given \eqref{Network_mimo}. Note that cooperation increases  the geographical interference protection region to $r_{\mathcal{I}}= r_{n-1}$ because the nearest $n$ BSs cooperate to serve the intended user and do not contribute to the aggregate interference. More advanced models for network MIMO with transmission  precoding and location aware cooperation are given in \cite{Cooperation1_martin, CoMP_martin, CoMP_sakr, Pairwise_Bacc}.

\begin{figure*}
\begin{align} \label{Network_mimo}
\mathcal{L}_{\mathcal{I}_{agg}}(a,\lambda) =  \!\!\! \!\!\!\!\!\!\!\!\!\! \underset{0<r_0<r_1<...<r_{n-1}<\infty}{\int \int ... \int}  \!\!\!\!\!\!\!\!\!\!\!\!\! \exp\left\{-\pi {\lambda} \sqrt{\frac{a}{\sum_{i=0}^{n-1} r_i^{-4}}}  \arctan\left( \frac{1}{r_{n-1}^2}\sqrt{\frac{a}{{\sum_{i=0}^{n-1} r_i^{-4}  }}}\right) \right\} f(r_0,r_1,...,r_n) \mathrm{d} r_0 \mathrm{d} r_1 .... \mathrm{d} r_{n-1}
\end{align}
\normalsize

\hrulefill
\end{figure*}

\normalsize

\subsection{Discussion}

\normalsize

This section discusses the numerical values obtained via SG analysis. Figs.~\ref{fig:ASER_DL_r0}, \ref{fig:Out_Spec}, and \ref{fig:Out_downlink} show high outage probability and ASEP values. Hence, it may be argued that the PPP results are quite pessimistic and do not reflect realistic system performance. However, we believe that the associated system model and assumptions, not the PPP, are the reasons for such pessimistic performance. That is, the naive universal frequency reuse, the saturated network model, and the peak transmit power pf the BSs are the main reason for the poor performance shown in Figs.~\ref{fig:ASER_DL_r0}, \ref{fig:Out_Spec}, and \ref{fig:Out_downlink}. To show that the system model, not the PPP,  are the main reasons for such pessimistic performance, we plot Fig.~\ref{fig:D_and_Re} which is obtained for a PPP cellular network with slightly different system model. Particularly, we incorporated receive diversity and frequency reuse, which are basic components of modern cellular networks. From the analysis perspective, we combined \eqref{LT_final} and \eqref{nakagami} to capture receive diversity and frequency reuse into the system model.

Fig.~\ref{fig:DR1} shows the explicit and combined effects of receive diversity and frequency reuse on the network outage probability. Fig.~\ref{fig:DR2} shows the combined effect of receive diversity and frequency reuse for different reuse factors and different numbers of receive antennas. Figs.~\ref{fig:DR1} and \ref{fig:DR2} show that incorporating simple network management techniques into the analysis leads to realistic values for the outage probability. For instance, with only two receive antennas and a reuse factor of 3, the outage probability at $T=0$ dB drops from almost $50\%$ (cf. Figs. \ref{fig:Out_Spec}, and \ref{fig:Out_downlink}) to below $5\%$. Incorporating more practical system parameters (e.g., power control and multi-slope path-loss) would further reduce the outage probability.

To recap, with the appropriate system model, SG analysis with the PPP assumption can capture realistic network performance and gives acceptable performance characterization. Sometimes we are interested in trends rather than absolute values. In this case, it is better to keep a simple system model to facilitate the analysis and to obtain insightful performance expressions. These expressions could be used to understand the network behavior in response to different network parameters and desing variables. However, it should be understood that the corresponding results are illustrative to the network behavior and do not give the true numerical values for the performance metrics.

\begin{figure}[t]
	\begin{center}
	  \subfigure[]{\label{fig:DR1}\includegraphics[width=3.0in]{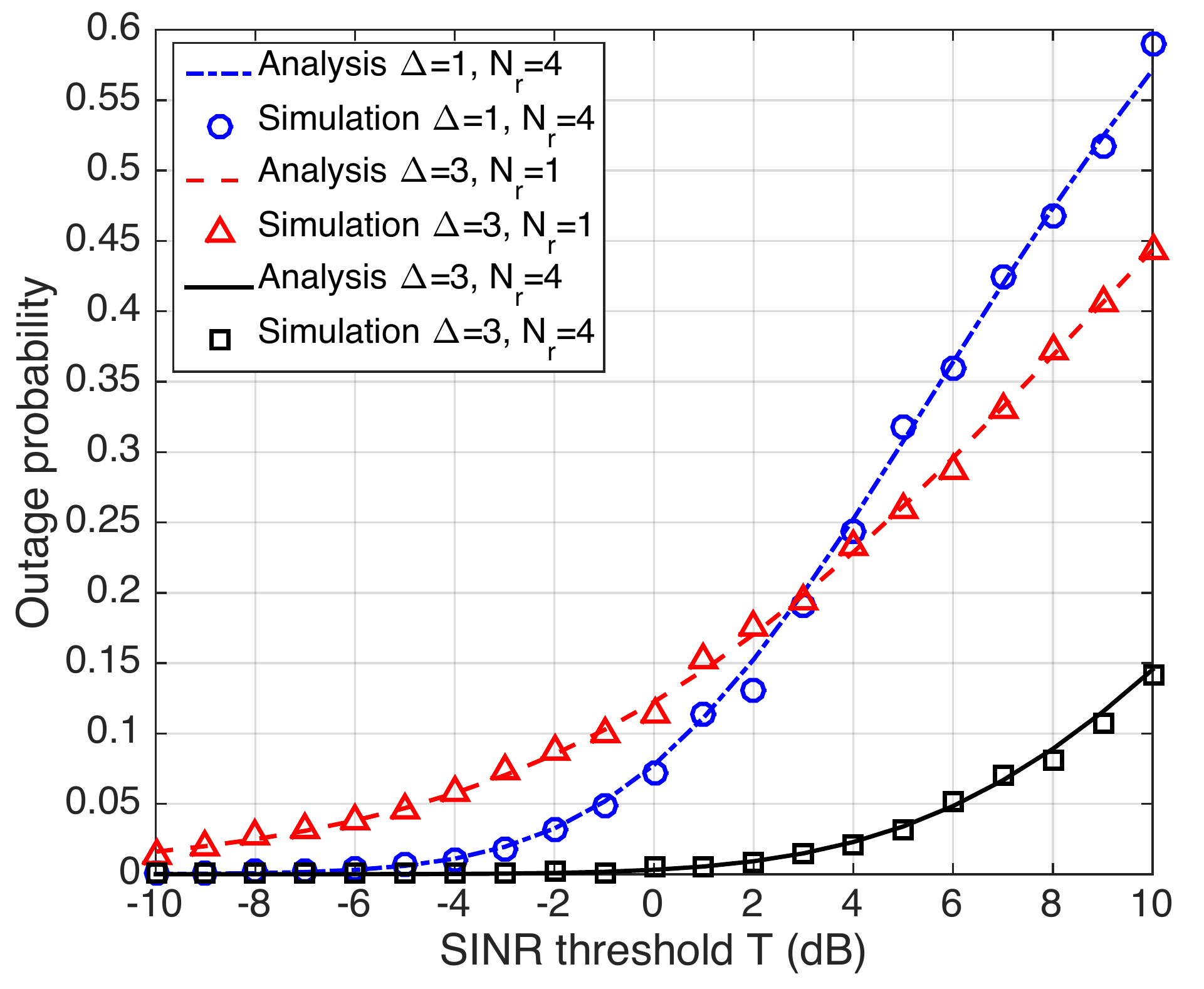}}
		\subfigure[]{\label{fig:DR2}\includegraphics[width=3.0in]{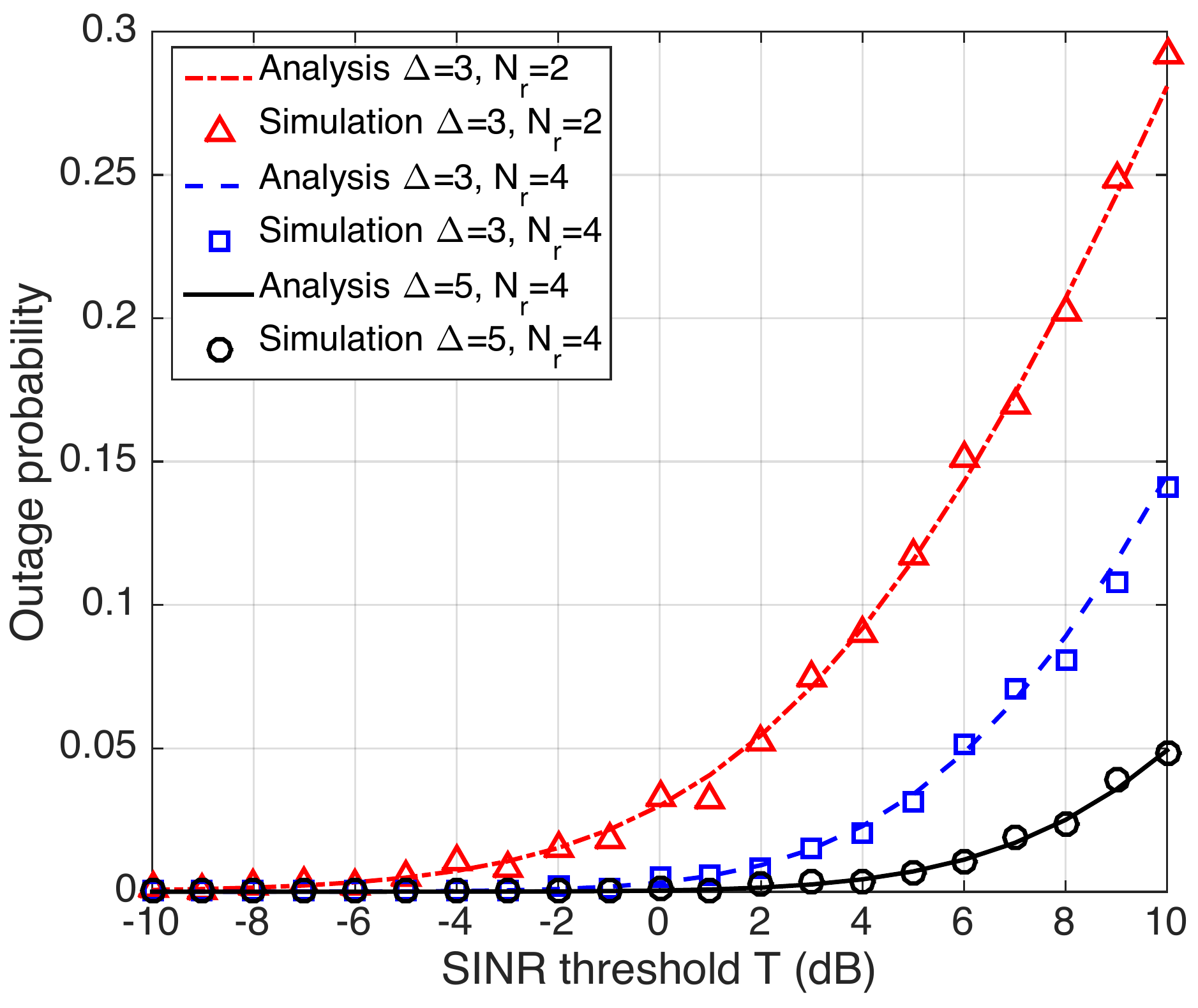}}
	\end{center}
	\caption{ The effect of receive diversity and coordinated frequency reuse on the outage probability.}
	\label{fig:D_and_Re}
\end{figure}


\section{Future Research Direction}  \label{sec:PPs}


SG analysis can be used to characterize the performance of large-scale setup wireless networks. For instance, it is well known that minimum Euclidean distance receivers are optimal if the intended symbol is disturbed by Gaussian noise. However, in large-scale networks where the intended symbol is disturbed by non-Gaussian interference in addition to the Gaussian noise, the optimal detector is unknown. Furthermore, results obtained for single point-to-point links cannot be directly generalized to large-scale networks. For instance, in a point-to-point link, the BER decreases with the transmit power. This fact does not hold for large-scale networks as the increased power of the useful signal is canceled by the increased interference power.  In this regard, SG paves the way to better understanding and more efficient operation of large-scale wireless networks. We highlight below some venues to extend  SG for better models of wireless networks.

\subsection{New Point Processes}

Exploring new tractable PP for modeling wireless networks is a fundamental research direction for SG analysis. Although we have shown that the PPP provides a good approximation for interference associated with repulsive point processes, the PPP alone is not enough to model all wireless networks. Wireless networks' topologies may include other complex correlations among the network elements rather than the simplified repulsion discussed in this paper. For instance, 5G networks define several types of communication including device-to-device (D2D) communication, vehicle-to-vehicle (V2V) communication, and machine-to-machine communication on top of the legacy device-to-BS communication~\cite{Elsawy_Vir_mag}. These various types of communications create complex topological structures that cannot be captured by PPP. This is because PPP is only characterized by its intensity and interference boundary, which offers limited degrees of freedom to model different topological structures. Hence, it is essential to develop SG models for wireless networks via new PPs. In this regard, there have been efforts invested to study new PPs in the context of cellular networks.  For instance, Poisson cluster processes for modeling attractive behavior between points are studied in \cite{now_martin, Ali_pcp, Harpreet_cluster}. Repulsive point processes such as the Mat\'ern hard core point process, the Ginibre point process, and the determinantal point process are studied in~\cite{Ginibre_martin, Ginibre_2, Giniber_uplink, cog_h, cog_h2, on_cog, on_ch, determinantal}. There are even efforts to characterize the asymptotic behavior of networks following general point processes \cite{High_SINR, High_SINR2}. In some cases when it is difficult to obtain explicit performance metrics in some network models, stochastic ordering can be exploited to compare their performances \cite{LTorder, Mimo_ordering}.  Note that the developed models using non-Poisson point processes are mostly for the baseline network model. Hence, besides exploring new point processes, extending existing non-Poisson based models to advanced network setup is also a potential research direction.

\subsection{Characterizing New Technologies}

Techniques used for transmissions and network management in wireless networks are continuously evolving to enhance the network performance and cope with the ever-increasing traffic demand. Usually, a proposal for a new technique starts with a theoretical idea followed by prototyping testbeds. However, it is challenging and costly to expose these techniques to realistic tests in large-scale setup. In this case, SG can serve as an initial and fast evaluation step for validating and quantifying the associated performance. For instance, in-band full-duplex (FD) communication, which emerges for recent advances in self-interference cancellation techniques, is optimistically promoted to double the spectral efficiency for wireless networks~\cite{FD1, FD2}. While this is true for a point-to-point link, it is not necessarily true in large-scale networks due to the increased interference level. In fact, \cite{h_FD} employed SG analysis to demonstrate the vulnerability of uplink to downlink interference and the negative effect that FD communication can impose on the uplink transmission. Then, in the light of the SG model in \cite{h_FD}, the authors proposed a solution to alleviate the negative impact of FD communication on the uplink transmission. Similar examples exist for other new technologies such as D2D communication~\cite{D2D_h, Jeffery_D2D, D2D_tony, D2D2_tony, Spatial_d2d}, coordinated multi-point transmission~\cite{CoMP_martin, CoMP_sakr, Pairwise_Bacc}, offloading and load balancing~\cite{Offloading_jeff, Joint_singh, load_aware_harpreet, Halim_1, User_Assoc_Wei}, uplink/downlink decoupling~\cite{uplink2_jeff}, massive MIMO~\cite{Massive_harpret}, and so on.

The above discussion shows the important role of SG in evaluating the gains associated with new technologies before the implementation step. Hence, it can be decided beforehand whether the new technology is worth the investment or not. Hence, performance characterization in large-scale networks via SG will always be a future research direction as long as new technologies are being proposed to enhance the performance of cellular networks as well as other types of large-scale networks.

\subsection{More Involved Performance Characterization}

In the context of cellular networks, SG is mainly confined to model interference and characterize outage, error probability, and transmission rate. An important direction for research is to extend SG analysis to model more performance metrics. For instance, SG can model other SINR related parameters such as secrecy rate~\cite{Sec_moe, Sec_moe_P1, Sec_moe_P2}, which is the fundamental performance metric in physical layer security. Looking into the literature, there are initiatives to asses physical layer security in cellular networks via the secrecy rate performance metric \cite{Security_harpreet, Security_Zhou, Security2_jeff, Bennis_sec}. However, this field of research is not mature enough to address the security problems imposed on 5G networks. In 5G networks there are massive D2D, M2M, and V2V communications on the top of the legacy user-to-BS communications. These different types of communications may serve applications (e.g., eHealth, smart city automation) which requires some level of privacy and confidentiality. Hence, developing secrecy rate models for modern cellular networks with D2D, M2M, and V2V communications is an interesting future research direction.

Stochastic geometry can also be extended beyond SINR characterization. For instance, cell boundary cross rate and cell dwell time are two fundamental performance metrics in cellular networks to design the handover procedure. The handover models available in the literature are mostly based on the circular approximation for the cell shape, which does not comply with recent measurements in \cite{martin_ppp, marco_fitting, tractable_app}. Hence, more accurate handover models for cellular network are required. In this regards, there are some initiatives to use SG to characterize handover in cellular networks as in \cite{Mobility_jeff, Mobility_liang, Adve_1}. However, complete handover designs based on SG are yet to be developed.

Developing new techniques for managing cellular networks may also define new performance metrics to be characterized. For instance, it is advised to transport and cache popular files in the cellular network edge during off-peak time to maximize the utilization of the core network and enhance the end user quality of service~\cite{proctive1}. In this case, the hitting probability, i.e., the probability that a user finds the requested file in a nearby BS, becomes a meaningful performance metric. Recently, models for hitting probability via stochastic geometry are developed and used to propose solutions to the caching problems based on file popularity~\cite{proctive2}.

\subsection{Statistical Network Optimization}

Cellular operators always seek an optimized operation of their networks. Modern cellular networks are composed of a massive number of network elements (i.e., BSs, users, devices, machines, etc.) which makes a centralized instantaneous optimization for the network infeasible. That is, it is infeasible to select serving BS, assign powers, allocate channels, and choose the mode of operation for each and every network element. In this context, SG analysis can be exploited for statistically optimized operation, which creates a tradeoff between complexity, signaling, and performance. While instantaneous optimization guarantees best performance at any time instant, statistical optimization provides optimal averaged performance on long-term scale to reduce signaling and processing overheads.  
 Note that statistical network parameters (e.g., distribution for channel gains, network elements spatial distribution and intensity, and so on.) change on longer time scales when compared to other instantaneous parameters such as channel realizations and users locations. For statistical network optimization, the performance objective functions and constraints can be formulated via SG analysis, which guarantees an optimal spatially averaged performance. Some efforts are invested in statistical network for cellular networks using SG\cite{Optimization_Adhoc, Optimization_cell}. However, to the best of the authors' knowledge, merging statistical and instantaneous optimization to balance performance, complexity, and signaling overhead is an open research problem.

\section{Conclusion}  \label{sec:conc}

We present a tutorial on stochastic geometry (SG) analysis for cellular networks. We first characterize interference in cellular networks by deriving its characteristic function and moments. Then, exact and approximated error performance analysis is conducted. We show that approximating the interfering symbols by Gaussian signals facilitates the analysis and simplifies the symbol error rate expressions without sacrificing accuracy. Then, we present the abstracted outage and ergodic rate analysis, which is used to further simplify the analysis and the performance expressions. To this end, we present a unified technique to compute error probability, outage probability, and ergodic rate for several system models in cellular networks. In particular, we show how the intensity and boundary of the PPP adapt to the network characteristics. We also present numerical examples and discussed the pessimistic performance obtained by SG. We show that with the proper network model, SG is capable of capturing realistic network performance. Finally, we point out future research directions for SG analysis.

\appendices 

\section{The Poisson Point Process}
\label{PPP}

The distance distribution between a generic location in $\mathbb{R}^2$ to the nearest point in a PPP $\Phi$ with intensity $\lambda$ is given by

\begin{align} \label{pdf_r_0}
f_{r_0}(r)  &=  2 \pi \lambda r e^{- \pi \lambda r^2},  \;\;\;\;\;\; r>0
\end{align}

The joint distance distribution between a generic location in $\mathbb{R}^2$ to the nearest and $n^{th}$ points in a PPP $\Phi$ with intensity $\lambda$ is given by

\begin{align} \label{eq:joint}
f_{r_0,r_n}(x,y)  &=  \frac{4 (\pi \lambda)^{n+1}}{\Gamma(n)} xy (y^2 - x^2)^{n-1} e^{- \pi \lambda y^2}, 
\end{align}

\noindent{where $0 \leq x \leq y \leq \infty$.

Let $f: \mathbb{R}^n \rightarrow \mathbb{R}$ be a measurable function and $\Phi \ in \mathbb{R}^n$ be a PPP, then by the PGFL we have: 
\begin{align} \label{PPP_PGFL}
\mathbb{E}\left[ \underset{x_i \in \Phi} \prod f(x_i) \right] = \exp\left\{ - \int_{\mathbb{R}^n} (1-f(x)) \Lambda(dx)   \right\}.
\end{align}

Let $V$ be the area of a generic PPP-Voronoi cell, then

\begin{equation}
f_V(v) \approx \frac{(\lambda c)^c v^{c-1} e^{- c\lambda v} }{\Gamma(c)}, \;\;\;\;\;\; 0 \leq v < \infty
\end{equation}

\noindent where $c=3.575$ is a constant defined for the Voronoi tessellation in the $\mathbb{R}^2$.

Consider two independent PPPs $\Phi_b$ and $\Phi_u$ with intensities $\lambda_b$ and $\lambda_u$. For the voronoi tessellation constructed w.r.t. $\Phi_b$, the probability mass function of the number of point of $\Phi_u$ existing in a generic voronoi cell of $\Phi_b$ is given by
\begin{align} \label{eq:load}
\mathbb{P}\left\{ \mathcal{U} =n \right\}  &=  \frac{\Gamma(n+c)}{\Gamma(n+1) \Gamma(c)} \frac{\left(\lambda_u\right)^n (\lambda_b c)^c}{( \lambda_b c + \lambda_u)^{n+c}},
\end{align}
\noindent{where $n= 0,1,2,\cdots.$ }

In a $K$-tier cellular network with intensities $\{\lambda_i\}_{k=1}^{K}$, bias factors $\{B_i\}_{k=1}^{K}$, and path-loss exponent $\{\eta_k\}_{k=1}^{K}$, the probability that a user associate with tier $k$ is given by

\small
\begin{align} \label{assoc_prob}
\mathcal{A}_k = 2 \pi \lambda_k \int\limits_0^\infty r \exp\left\{- \pi \sum_{i=1}^K \lambda_i \left(\frac{B_i P_i}{B_k P_k}\right)^\frac{2}{\eta_i} r^{\frac{2\eta_k}{\eta_j}} \right\} dr.
\end{align}
\normalsize

The service distance $r_{0,k}$ distribution for a user associated to a BS in the $k^{th}$ tier is given by
\small
\begin{equation} \label{mt_distance}
f_{r_{0,k}} = \frac{2 \pi \lambda_k x}{\mathcal{A}_k} \exp\left\{- \pi \sum_{i=1}^K \lambda_i \left(\frac{B_i P_i}{B_k P_k}\right)^\frac{2}{\eta_i} r^{\frac{2\eta_k}{\eta_j}} \right\}.
\end{equation}
\normalsize

\section{Lemma~1 in \cite{hamdi_useful_tech}.} \label{app:useful}


Let $Y\sim Gamma(m,m)$ be a unit mean gamma distributed random variable, $X$ be a real random variable with the LT $\mathcal{L}_x(\cdot)$, and $C$ be a constant. The authors in \cite{hamdi_useful_tech} proposed a technique to calculate averages in the form of $\mathbb{E}\left[ \text{erfc} \left(\sqrt{\frac{Y}{X+C}} \right) \right]$ and $\mathbb{E}\left[ \text{erfc}^2 \left(\sqrt{\frac{Y}{X+C}} \right) \right]$. These averages are given by 
\begin{align}
\label{eq:erfc1}
&\mathbb{E}\left[\text{erfc}\left(\sqrt{\frac{Y}{X+C}}\right) \right]=1- \frac{\Gamma(m+\frac{1}{2})}{\Gamma(m)} \frac{2}{\pi} \cdot \notag \\ & \int_0^{\infty} \frac{1}{\sqrt{z}}e^{-z\left( 1+m C\right)} {_1}F_1\left( 1-m; \frac{3}{2}, z\right) \mathcal{L}_{X}\left( m z\right) \mathrm{d} z, \\
&\overset{m=1}{=} 1- \frac{1}{\sqrt{\pi}} \int_0^{\infty} \frac{e^{-z\left( 1+ C\right)} \mathcal{L}_{X}\left( z\right)}{\sqrt{z}}  \mathrm{d} z,
\end{align}
and
\begin{align}
\label{eq:erfc2}
&\mathbb{E}\left[\text{erfc}^{2}\left(\sqrt{\frac{Y}{X+C}}\right) \right]=1- \frac{4 m}{\pi} \int_0^{\infty} \!\!\!\!\!\! e^{-z m C}   \notag \\ 
&\;\;\;\; \mathcal{L}_{X}\left( m z\right) \int_0^{\frac{\pi}{4}}\!\!\!\!\!{_1}F_1\left( m+1; 2, \frac{-z}{\sin^2 \vartheta}\right) \frac{ \mathrm{d}\vartheta}{\sin^2 \vartheta }\mathrm{d} z, \\
&\overset{m=1}{=} 1- \frac{2}{\sqrt{\pi}} \int_0^{\infty} \frac{e^{-z\left( 1+ C\right)} \text{erfc}(\sqrt{z})\mathcal{L}_{X}\left(  z\right)}{\sqrt{z}}  \mathrm{d} z.
\end{align}

\section{LT of $\zeta$} \label{app:LT_zeta}
Let $\zeta= \sum_{q=1}^\infty \frac{B_q \sigma^2_q}{P r_0^{-\eta}} $. Then the LT of $\zeta$ can be derived as
\begin{align} 
\mathcal{L}_\zeta(z) &= \mathbb{E}\left\{e^{-\frac{z (  \sum_{q=1}^\infty B_q \sigma^2_q )}{P r_0^{-\eta}}  }\right\} \notag \\
&=  \mathbb{E}\left\{\prod_{q=1}^\infty e^{-\frac{z B_q \sigma^2_q}{P r_0^{-\eta}}  }\right\} \notag \\
&= \prod_{q=1}^\infty \mathcal{L}_{B_q}\left(\frac{z \sigma^2_q}{P r_0^{-\eta}} \right) \notag \\
&=  e^{ \sum_{q=1}^\infty \left( - \frac{z \sigma^2_q}{P r_0^{-\eta}} \right)^q} \notag \\
&=  e^{\sum_{q=1}^\infty \left( \frac{z \left( \frac{2 \pi\lambda r_0^{2-\eta q} P^q (-1)^{q} \mathbb{E}\left\{ |s|^{2q}\right\}}{ (\eta q-2) q!} \right)^{\frac{1}{q}}}{P r_0^{-\eta}} \right)^q}  \notag \\
&=  e^{ 2 \pi\lambda r_0^{2} \mathbb{E}\left\{ \sum_{q=1}^\infty  \left( \frac{(-1)^{q}  (z|s|^{2})^q }{ (\eta q-2) q!} \right) \right\}} \notag \\
&=  e^{\pi\lambda r_0^{2}  \left(1-\mathbb{E}\left\{ _{1}F_1\left( -\frac{2}{\eta}; 1-\frac{2}{\eta}; -z |s|^2 \right)  \right\} \right) } \notag \\
&=  e^{\pi\lambda r_0^{2}  \left(1- \frac{1}{M} \sum_{m=1}^{M}  {}_{1}F_1\left( -\frac{2}{\eta}; 1-\frac{2}{\eta}; -z |s_m|^2 \right)   \right) }.
\label{eq:LT_zeta}
\end{align}

\section{Poof of Lemma~\ref{LT_I_power}} \label{LT_I_power_proof}

Let $\mathcal{I}_{agg}= { \sum_{r_k \in \tilde{\Psi}\setminus r_0} s P |h_i|^2 r_i^{-{\eta}}}$. Then the LT of $\zeta$ can be derived as

\small
\begin{align}
\mathcal{L}_{\mathcal{I}_{agg}}(s) &{=} \EXs{\tilde{\Psi}}{ \prod_{r_k \in \tilde{\Psi}\setminus r_0} \mathbb{E}_{h_i} \left\{ e^{- s P |h_i|^2 r_i^{-{\eta}}} \right\} } \notag \\
 &{=}  \exp \left\{-\int_{\mathbb{R}^2} \left(1- \EXs{h}{ e^{ -s {P} |h|^2 r^{-{\eta}}} } \right)  \Lambda(dr) \right\} \notag \\
 &{=}  \exp \left\{- 2 \pi \lambda \int_{r_0}^\infty \left(1- \EXs{h}{e^{-s P |h|^2 r^{-{\eta}}}}\right) r dr \right\} \notag \\
  &{=}  \exp \left\{-2 \pi \lambda \int_{r_0}^\infty \left(1-\frac{1}{1+sPr^{-\eta}}\right) r dr \right\}  \notag \\
  &{=}  \exp \left\{-2 \pi \lambda \int_{r_0}^\infty \left(\frac{sP}{r^{\eta}+sP}\right) r dr \right\} \notag \\
   &{=}  \exp \left\{-2 \pi \lambda (sP)^\frac{2}{\eta} \int_{\frac{r_0}{(sP)^\frac{1}{\eta}}}^\infty \left(\frac{x}{x^{\eta}+1}\right) dx \right\} \notag \\
      &{=}  \exp \left\{-\frac{2 \pi \lambda sP r_0^{2-\eta}}{\eta-2} {}_{2}F_1\left(1, 1-\frac{2}{\eta}; 2- \frac{2}{\eta}; - \frac{sP}{r^{\eta}}\right)  \right\}.
\end{align} \label{eq:LT_proof}
\normalsize

\bibliographystyle{IEEEtran}
\bibliography{IEEEabrv,Tutorial}

\end{document}